\documentclass[lettersize,journal]{IEEEtran}
\usepackage[T1]{fontenc}% optional T1 font encoding
\ifCLASSINFOpdf
\else
\fi
\usepackage[cmintegrals]{newtxmath}
\usepackage{url}
\usepackage{cite}
\usepackage{amsmath}

\usepackage{amsthm}
\usepackage{mathtools}
\usepackage{algorithmic}
\usepackage[ruled,linesnumbered]{algorithm2e}
\usepackage{graphicx}
\usepackage{textcomp}
\usepackage{xcolor}
\usepackage{bm}
\usepackage{enumitem}
\usepackage{subfigure}
\usepackage{url}
\usepackage{hyperref}
\usepackage{ulem}

\newtheorem{assumption}{Assumption}
\newtheorem{proposition}{Proposition}
\newtheorem{lemma}{Lemma}
\newtheorem{theorem}{Theorem}

\newcommand{\ie}{\textit{i}.\textit{e}.}

\def\prox{\mbox{prox}}
\newcommand{\cmmnt}[1]{}
\newcommand{\tr}[1]{}
\newcommand{\fp}{}

\allowdisplaybreaks[4]

\newboolean{showcomments}
\setboolean{showcomments}{true}
\newcommand{\comment}[1]{\ifthenelse{\boolean{showcomments}}
	{\textcolor{red}{(Comment: #1)}}
	{}
}
\newcommand{\red}[1]{\ifthenelse{\boolean{showcomments}}
	{\textcolor{red}{#1}}
	{}
}

\newboolean{showanswers}
\setboolean{showanswers}{true}
\newcommand{\answer}[1]{\ifthenelse{\boolean{showanswers}}
	{\textcolor{blue}{(Answer: #1)}}
	{}
}

\SetKwRepeat{Do}{do}{while}

\begin{document}

\title{PoPeC: PAoI-Centric Task Offloading with Priority over Unreliable Channels\vspace{-0pt}}
\author{
        Nan~Qiao,~\IEEEmembership{Student Member,~IEEE},
        % \textit{OCID: 0000-0001-2345-6789},
        Sheng~Yue,~\IEEEmembership{Member,~IEEE},
        Yongmin~Zhang,~\IEEEmembership{Senior Member,~IEEE}, and Ju Ren,~\IEEEmembership{Senior Member,~IEEE}
    \IEEEcompsocitemizethanks{\IEEEcompsocthanksitem Nan Qiao and Yongmin Zhang are with the School of Computer and Engineering, Central South University, Changsha, Hunan, 410083 China. E-mails: \texttt{\{nan.qiao,zhangyongmin\}@csu.edu.cn}.
    \IEEEcompsocthanksitem Sheng Yue and Ju Ren are with the Department of Computer Science and Technology, BNRist, Tsinghua University, Beijing, 100084 China. E-mails: \texttt{\{shengyue,renju\}@tsinghua.edu.cn}.
    \IEEEcompsocthanksitem Corresponding author: Ju Ren.
    }  
}
        
\maketitle

% \onecolumn
%abstract
% \IEEEtitleabstractindextext{
\begin{abstract}
Freshness-aware computation offloading has garnered increasing attention recently in the realm of edge computing, driven by the need to promptly obtain up-to-date information and mitigate the transmission of outdated data. However, most of the existing works assume that channels are reliable, neglecting the intrinsic fluctuations and uncertainty in wireless communication. More importantly, offloading tasks typically have diverse freshness requirements. Accommodation of various task priorities in the context of freshness-aware task scheduling and resource allocation remains an open and unresolved problem. To overcome these limitations, we cast the freshness-aware task offloading problem as a multi-priority optimization problem, considering the unreliability of wireless channels, prioritized users, and the heterogeneity of edge servers. Building upon the nonlinear fractional programming and the ADMM-Consensus method, we introduce a joint resource allocation and task offloading algorithm to solve the original problem iteratively. In addition, we devise a distributed asynchronous variant for the proposed algorithm to further enhance its communication efficiency. We rigorously analyze the performance and convergence of our approaches and conduct extensive simulations to corroborate their efficacy and superiority over the existing baselines.
% Specifically, we perform various transformations, such as fractional programming, and use ADMM to solve several special cases.
% Subsequently, building upon these methods and conclusions drawn earlier, we use an iterative solution algorithm to obtain the solution to the original problem.
% Furthermore, we propose an asynchronous non-convex ADMM method with a sublinear rate of communication defects and discuss the theoretical performance improvement due to the multi-priority mechanism.
%simulaiton
% Finally, simulation results demonstrate the performance gain.
\end{abstract}
% }

%introduction
\section{Introduction}
\label{sec:introduction}

%background
%importance of task offloading

Edge computing is an attractive computing paradigm in the era of the Artificial Internet of Things (AIoT)~\cite{liu2019survey,zhang2017survey}. By enabling end devices to offload computation-intensive tasks to nearby edge nodes, it is envisioned to provide real-time computing services, thereby facilitating the deployment of a wide range of intelligent applications (e.g., smart homes, smart cities, and autonomous vehicles)~\cite{tran2018joint}. In these applications, it is of paramount importance to promptly access up-to-date information while mitigating the transmission of outdated and worthless data~\cite{zou2021minimizing}. 
To this end, a great number of offloading solutions have been proposed to ensure timely status updates and rapid delivery of tasks from information sources, with the aim of enhancing the overall \textit{information freshness}~\cite{guo2021scheduling,sun2021age,li2021task,bedewy2020optimizing,pan2021minimizing}.

%the research community has shown considerable interest in developing solutions that can achieve timely status updates and rapid delivery of tasks from information sources, in order to enhance the \textit{freshness} performance~\cite{guo2021scheduling,sun2021age,li2021task,bedewy2020optimizing,pan2021minimizing}.

% how to enable timely status updates and rapid delivery of tasks from information sources, in order to achieve excellent freshness performance has received a lot of attention~\cite{guo2021scheduling,sun2021age,li2021task,bedewy2020optimizing,pan2021minimizing}.

% Recently, the Age of Information (AoI) and Peak Age of Information (PAoI) have been recognized as crucial indicators of information freshness, which refers to how long it has been since a user's last packet has been received~\cite{yates2021age}.
{Recently, the Age of Information (AoI) and Peak Age of Information (PAoI) 
have been recognized as important metrics for evaluating the freshness of information, which characterize the elapsed time since the reception of a user's most recent data packet~\cite{yates2021age}. Based on these metrics, several recent efforts have focused on AoI- or PAoI-centric computation offloading and resource allocation for efficient and concurrent transmission of freshness-sensitive information to the edge servers~\cite{li2021task,liu2021aion,li2021scheduling,moon2015minimax,pan2021minimizing,zou2021minimizing,sun2021age,guo2021scheduling,bedewy2020optimizing,lv2021strategy,yates2018status,bedewy2016optimizing,li2020waiting}.
}

% Other studies investigate offloading strategies for distributing fresh tasks from multiple sources to a collection of edge servers~\cite{lv2021strategy,yates2018status,bedewy2016optimizing,li2020waiting}.

% focus on how to offload fresh tasks from multiple sources to a collection of edge servers~\cite{lv2021strategy,yates2018status,bedewy2016optimizing,li2020waiting}.

{Unfortunately, there remain several challenges that need to be surmounted to achieve effective freshness-aware computation offloading in practice. First, many existing works assume channel homogeneity~\cite{sun2019closed,bedewy2020optimizing} or perfect knowledge of channel states~\cite{pan2021minimizing,sombabu2020age}, overlooking the dynamics and stochasticity of the limited wireless channels. As a result, such methods easily suffer from package loss or failure due to unreliable communication~\cite{moon2015minimax}. Second, computing resources on edge servers are typically constrained and heterogeneous, necessitating appropriate assignment of heterogeneous computing units to offloading tasks.
% neglected in the literature~\cite{bedewy2020optimizing, hsu2019scheduling, karabulut2018spatial, zeng2015optimized}. 
More importantly, designing a \textit{prioritized offloading strategy} is crucial because users may have diverse freshness requirements. 
% For example, safety-sensitive devices like temperature sensors in the Industrial Internet of Things and Automatic Emergency Braking (AEB) in autopilot require prompt offloading and processing to meet their high freshness demands. However, most existing methods struggle with measuring and handling the situations where the offloading tasks have different priorities.
For example, devices with safety-sensitive functions, such as temperature sensors in the Industrial Internet of Things and Automatic Emergency Braking (AEB) systems in autopilots, require prompt offloading and processing to meet their stringent freshness demands. Whereas, most of the existing methods struggle with measuring and handling the situations where offloading tasks possess different priorities.
In light of these considerations, this work seeks to answer the key question: 
% {
\textit{``How to design an efficient task offloading algorithm that can optimize the overall information freshness while effectively handling prioritized users, unreliable channels, and heterogeneous edge servers?''} }

To this end, we cast the freshness-aware task offloading problem as a multi-priority optimization problem, considering unreliability of wireless channels, heterogeneity of edge servers, and interdependence of multiple users with differing priorities. Given the high complexity of directly optimizing this problem, we first examine two special cases from the original problem and exploit nonlinear fractional programming to transform the problems into tractable forms, subsequently developing ADMM-Consensus-based solutions for both cases. Built upon these solutions, an iterative algorithm is devised to resolve the original problem effectively. We further discuss a distributed asynchronous variant of the proposed algorithm, capable of alleviating the overhead caused by unreliable iterations during the offloading policy acquisition process.
%To answer this question, we model the task offloading problem for users with multi-class priority in unreliable channels. Due to the unreliability of wireless channels, the heterogeneity of devices and servers, and the coupling of multi-class priority users, this problem, which has been demonstrated to be NP-Hard, is difficult to address directly. To this end, we decompose the problem into three subproblems that are solved step by step. We convert such complex problems using nonlinear fractional programming, and then employ the ADMM-Consensus method to address them. 
Theoretical analysis is carried out to establish the convergence property of the proposed algorithm and demonstrate the improvement in performance brought by the multi-priority mechanism.
%We theoretically analyze the performance improvement brought by the multi-class priority mechanism
% \sout{The efficacy and efficiency of the proposed algorithms are validated through extensive simulation experiments.} %Finally, we conduct extensive simulation experiments to confirm the performance and efficiency of the proposed algorithms.
% \textbf{Contribution}\\

In a nutshell, our main contributions are summarized below.
\begin{itemize}
    \item %We cast an offloading problem based on PAoI with heterogeneous users and unreliable channels. We transform the primary problem into a convex form and propose an ADMM-Consensus algorithm to resolve it efficiently in a distributed manner.
    We consider an M/G/1 offloading system and derive the precise Peak Age of Information (PAoI) expression for each user to characterize their information freshness. Then, we formulate the freshness-aware multi-priority task offloading problem under heterogeneous edge servers and unreliable channels.
    % To solve this problem, we transform the primary problem into a convex form and propose an efficient distributed solution using the ADMM-Consensus algorithm.
    \item Based on nonlinear fractional programming and the ADMM-Consensus method, we propose a joint resource allocation, service migration, and task offloading algorithm to solve the original problem effectively. We further devise a distributed asynchronous variant for the proposed algorithm to enhance its communication efficiency.
    % We assign different priority classes to users based on their sensitivity to freshness and formulate the task offloading problem as an M/G/1 system with multi-class priority. We derive the precise PAoI expression and propose an efficient distributed solution to optimize the problem. Furthermore, we carry out a large number of simulation experiments to demonstrate the effectiveness of the priority system.
    %We assign different priority classes to users based on their sensitivity regarding freshness. We formulate this task offloading problem with priority as an M/G/1 system with multi-class priority, derive its precise PAoI expression, and solve it efficiently in a distributed manner.
    \item %For the communication iteration problem of algorithms over unreliable channels, we propose a distributed asynchronous algorithm that iterates with limited information at each iteration. We establish the guarantees for the performance and convergence of the algorithm, and empirically show that its convergence rate is higher than those of the existing synchronous distributed algorithms.
    % We propose a distributed asynchronous algorithm for the communication iteration problem over unreliable channels. 
    % Our algorithm iterates with limited information at each iteration and 
    We establish theoretical guarantees for the proposed algorithms, in terms of performance and convergence. We conduct extensive simulations, and the results show that our algorithm can significantly improve the performance over the existing methods.
    % \item %To embed task offloading into the multi-server scenario, we propose a resource allocation problem that captures the trade-off between convergence, local offloading, and service migration. We propose a joint task offloading and service migration algorithm. Theoretical analysis and extensive experimental results show that we can solve the problem efficiently.
    % We propose a resource allocation problem that captures the trade-off between convergence, local offloading, and service migration to embed task offloading into the multi-server case. In addition, we propose a joint task offloading and service migration algorithm. Theoretical analysis and extensive experimental results demonstrate the algorithm's efficiency in solving the problem.
\end{itemize}

%The rest of the paper is organized as follows: Section \ref{sec:RelatedWork} demonstrates some similar existing work.Section \ref{sec:system_model} is dedicated to the system model where the required definitions and the relevant model are presented. In Section \ref{sec:PoPeC}, we gradually describe the three scenarios eventually proposing effective algorithms to tackle the PoPeC problem.In Section \ref{sec:Discussion}, we discussed the benefits of the multi-class priority mechanism and proposed an asynchronous parallel algorithm to make the algorithm more communication efficient.Simulation results that corroborate these findings are laid out in Section \ref{sec:simulation} while the paper is concluded in Section \ref{sec:conclusion}.  
The remainder of the paper is organized as follows: Section \ref{sec:RelatedWork} briefly reviews the related work. Section \ref{sec:system_model} introduces the system model, including relevant definitions and models. In Section \ref{sec:PoPeC}, we describe some special cases and propose algorithms to tackle the PoPeC problem. Section \ref{sec:Discussion} proposes an asynchronous parallel algorithm to improve communication efficiency and discusses the benefits of the multi-class priority mechanism. 
Section \ref{sec:simulation} presents the simulation results, followed by a conclusion drawn in Section \ref{sec:conclusion}.

\section{Related Work}\label{sec:RelatedWork}

{In the realm of edge computing, a multitude of research efforts have emerged to mitigate response delays by means of task offloading strategies~\cite{zhang2021joint,halder2022dynamic,saleem2020mobility,al2020task,wang2020multi,xujie2023mobility,chen2023qos,bai2023towards,chen2023joint,chen2023TECS,liu2023joint}. 
On the one hand, exploring the characteristics of channels is integral to this field. 
The reliability of channels has been investigated in scenarios such as real-time monitoring systems~\cite{abd2020reinforcement,hsu2019scheduling}. However, these approaches might not adequately address the challenges posed by unstable channel conditions stemming from factors like antenna beamforming and fading~\cite{karabulut2018spatial,zeng2015optimized}.
Moreover, while earlier studies assumed either homogeneous channels or the offloading of two separate channels with random arrivals~\cite{sun2019closed,pan2021minimizing,bedewy2020optimizing}, such assumptions fall short when dealing with the complexities of heterogeneous unreliable channels. 
On the other hand, one notable departure in our work is the consideration of the performance of synchronous parallel iterative algorithms in the context of unreliable channels. While many studies advocate for distributed methods to enhance efficiency, they often neglect the substantial communication costs of synchronous parallel algorithms on unreliable channels~\cite{sun2021age,zou2021minimizing,guo2021scheduling}. 
In contrast, we explore the potential of an asynchronous parallel algorithm to mitigate communication overheads and achieve the same performance.}

{Nonetheless, the sole emphasis on delay reduction may not ensure the necessary freshness of information for users, as highlighted in the works of Kosta et al.~\cite{kosta2017age} and Yates et al.~\cite{yates2021age}. This has led to the development of freshness-aware methodologies, leveraging metrics such as AoI and Peak Age of Information (PAoI)~\cite{liu2021aion,li2021scheduling,zou2021optimizing,zou2021minimizing,guo2021scheduling,sun2021age,li2021task,bedewy2020optimizing,pan2021minimizing}.
Therefore, recent advancements have delved into the customization of computation offloading strategies to accommodate distinct user types and preferences~\cite{pan2021minimizing,guo2021scheduling,zou2021minimizing,sun2021age,abd2022age,huang2015optimizing,liu2021anti,maatouk2020status,maatouk2019age,kaul2018age,xu2020peak}. 
Zou et al.~\cite{zou2021minimizing} introduce a novel partial-index approach that accurately characterizes indexing issues in heterogeneous multi-user multi-channel systems. Their SWIM framework optimizes resource allocation using maximum weights. Sun et al.~\cite{sun2021age} propose an age-aware scheduling strategy rooted in Lyapunov optimization to cater to diverse users, providing bounds on age that comply with throughput constraints. 
{
However, many of these contributions overlook the importance of user priorities, which is vital for practical prioritized systems~\cite{pan2021minimizing,guo2021scheduling,zou2021minimizing,sun2021age,abd2022age}.}}

{In contrast to prior works, which may focus on specific aspects such as channel types, reliability, or algorithmic choices, our research amalgamates these elements to address the intricate interplay of AoI optimization, edge computing, and heterogeneous channels. By investigating the challenges unique to these intersections, we contribute to a more comprehensive understanding of real-time computation offloading in dynamic environments.}

\section{System Model}\label{sec:system_model}
% In this section, our focus is on the Mobile Edge Computing (MEC) architecture for task scheduling in edge scenarios, encompassing task offloading and transmission models, as well as addressing the common limitations arising from resource constraints.  To supplement the limitations of the MEC architecture in practical situations, we present the capacity model with confidence evaluation.  We also delve into the information freshness model and problem formulation.

In this section, we introduce the system model, including the Mobile Edge Computing (MEC) architecture, the freshness model, and the problem formulation.

% In this section, we focus on the MEC architecture for task scheduling in edge scenarios, such as task offloading and transmission models, as well as certain common limitations.
% Due to resource constraints, we additionally present the capacity model with confidence evaluation to supplement the limitations of the MEC architecture in actual circumstances.
% Furthermore, we focus on the information freshness model and problem formulation.
\begin{figure}[ht!]
    \vspace{-0cm} 
    \setlength{\abovecaptionskip}{-0cm} 
    \setlength{\belowcaptionskip}{-0cm} 
    \centering
	\includegraphics[width=0.5\textwidth]{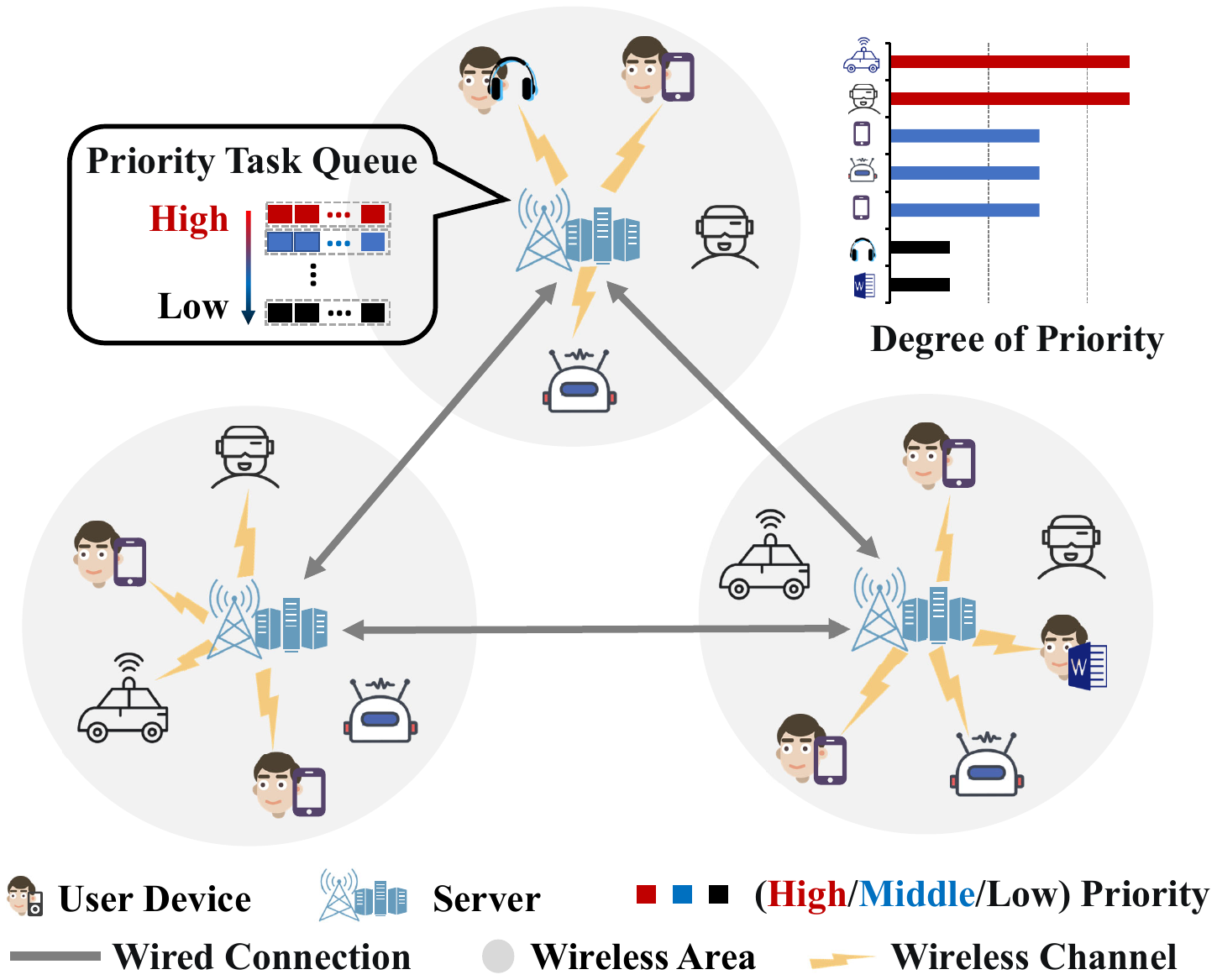}
    \vspace{-0pt}
	\caption{MEC Architecture with Unreliable Channles and Priority Users.}
	\label{fig:system-model}
    % \vspace{-18pt}
\end{figure}
\subsection{MEC Architecture}
{As shown in Fig.\ref{fig:system-model}, we consider a wireless network system consisting of a set of $M$ mobile edge servers (denoted by $\mathcal{M}$). Each server $m\in\mathcal{M}$ serves a set of $N_m$ users ($\mathcal{N}_m$), and each user $n\in\mathcal{N}$ can offload their tasks to the corresponding server $m_n$ via a limited number of wireless channels (denoted by $\mathcal{C} \doteq \{1,\dots, C\}$). Considering the effects of frequency-selective fading~\cite{zou2021minimizing}, we define $p_{n,c}$ as the probability of a successful transmission from user $n$ through channel $c$ to server $m_n$. In this context, users have (potentially) varying offloading priorities, represented as $\Delta\doteq\{1,2,\dots, \delta_\mathrm{max}\}$. Let $\mathcal{N}^{\delta}_m$ denote the set of users that are prioritized as $\delta\in\Delta$ and served by server $m$, which satisfies $\bigcup_{m\in\mathcal{M}}\bigcup_{\delta\in\Delta} \mathcal{N}^{\delta}_m = \bigcup_{m\in\mathcal{M}} \mathcal{N}_m = \mathcal{N}$. We represent $\mathcal{N}^{\delta} \doteq \bigcup_{m\in\mathcal{M}} \mathcal{N}^{\delta}_m$ as the set of users with priority $\delta$ across all servers. To simplify notation, we denote the priority of user $n$ as $\delta(n)$ and the set of users with the same or higher priority as $\Delta(\delta(n))$. 
Particularly, we distinguish user-side and server-side variables using the superscripts `u' and `s', respectively. For clarity, we summarize the key notations in Table \ref{table:notations}.
\begin{table}[ht]
	% \setlength{\tabcolsep}{4pt}
	%\small
	\caption{Key Notations}
	\label{table:notations}
	\vspace{-0.5em}
	\centering
	\renewcommand\arraystretch{1.25}
        \resizebox{\columnwidth}{!}{
		\begin{tabular}{l|l}
			\hline
			\textbf{Notation}                          & \textbf{Definition} \\ \hline
			$n$, $c$, $m$, $\delta$                    & indexes of users, channels, servers, and priority\\
			$u$, $s$                                   & superscripts of user-side and server-side variables\\
			$\lambda_{n}$, $\mu_{n}$                   & \begin{tabular}[c]{@{}l@{}}expected arrival and service rates of $n$'s offloaded tasks  \end{tabular} \\
			$p_{n,c}$                                  & \begin{tabular}[c]{@{}l@{}}probability of a successful update from user $n$ to channel $c$ \end{tabular} \\
			% $\frac{1}{\mu_{n}}$, $\nu_{n}$             & \begin{tabular}[c]{@{}l@{}}first and second moments of the general distribution describing\\ the processing time of server $m$ for user $n$'s tasks\end{tabular}\\
			$y_m$                                      & optimal migration decision of server $m$ \\
			$\eta^u_{n,c}$                             & \begin{tabular}[c]{@{}l@{}}probability of tasks offloading from user $n$ to channel $c$\end{tabular} \\
			$\eta^s_{m,m'}$                            & proportion of tasks delivered from server $m$ to server $m'$ \\
			$p^\mathrm{syn}$, $p^\mathrm{asyn}$                    & \begin{tabular}[c]{@{}l@{}}probability of a successful (a)synchronous update\end{tabular}  \\
			$\Gamma^\mathrm{syn}$, $\Gamma^\mathrm{asyn}$          & \begin{tabular}[c]{@{}l@{}}number of iterations of (a)synchronous algorithms\end{tabular}   \\	
			$T_{n}$, $I_{n}$, $W_{n}$, $Y_{n}$         & \begin{tabular}[c]{@{}l@{}}user $n$'s transmission time, arrival
interval, waiting time,\\ and processing time\end{tabular}  \\
            $A_{n}$                                    & PAoI of user $n$'s message  \\
			$\epsilon^{\mathrm{ck}}$, $\epsilon^{\mathrm{ac}}$           & \begin{tabular}[c]{@{}l@{}}stop criteria for iterations  CheckPointing Algorithm (NFP) \\and ADMM-Consensus Algorithm (NAC/ANAC)\end{tabular}   \\ \hline
		\end{tabular}
    }
% \vspace{-.5em}
\end{table}
}
% \sout{That is, users will offload tasks to a local server first.}  
% \comment{what's the meaning?} 
% \answer{We want to emphasize that task offloading occurs first from the user to its local server and has been merged with the previous sentence.}
% \sout{Also, the local server for user $n$ is denoted by $m(n)$.} 
% \comment{place this to where $m(n)$ is used}
% \answer{We've done that.}
%In addition, we assume that the probability of user $n$ successfully reaching the destination in order through channel $c$ is $p_{n,c}$, independent of all other transmissions, due to channel instability with frequency selective fading resulting in  packet errors~\cite{zou2021minimizing}.

%Denote the set of priorities as $\Delta\doteq\{1,2,\dots, \delta_\mathrm{max} \}$ and the set of users with priority $\delta\in\Delta$ as $\mathcal{N}^{\delta}$ which satisfies $\bigcup^{}_{\delta\in\Delta} \mathcal{N}^{\delta} = \mathcal{N}$. With slight abuse of notation, we denote $\delta(n)$ as the priority of user $n$  and $\Delta(\delta(n))$ as the set of users with the same or higher priority than that of $n$.

%$(\delta(n)-1)$ represents the priority level with priority $\delta(n)$ one level higher.

%Furthermore, with multi-class priorities mentioned, we specify some additional variables. For simplicity, 
% Meanwhile,  and we have .For convenience, we denote $N = |\mathcal{N}|$, $N_m = |\mathcal{N}_m|$, $N^{\delta} = |\mathcal{N}^{\delta}|$, $C = |\mathcal{C}|$, and $M = |\mathcal{M}|$, with the superscript `u' and `s' distinguishing between user-side and server-side variables, respectively.

\subsubsection{Task Offloading} 

The computational tasks of each user arrive according to a Poisson process with an expected arrival rate of $\lambda_{n}$. We denote the probability of user $n$ accessing channel $c$ as $\eta^u_{n,c}$, which satisfies
\begin{equation}
\label{eq:c1}
    0 \le \eta^u_{n,c} \le 1 , \quad \forall n \in \mathcal{N},~c \in \mathcal{C}. 
\end{equation}
{Clearly, the number of tasks offloaded by user $n$ through channel $c$ also adheres to a Poisson progress with an expected value of $\eta^u_{n,c}\lambda_{n}$. Given the fact that the number of offloaded tasks by user $n$ can not exceed that generated by the same user, we have
\begin{equation}
\label{eq:c2}
    \sum_{c\in\mathcal{C}}\eta^u_{n,c} \le 1 , \quad \forall n \in \mathcal{N}.
\end{equation}
% The inequality in Eq. \eqref{eq:c2} implies that the offloading amount $n$'s tasks will not be larger than the generation amount of user $n$'s tasks.
The inequality in Eq.~\eqref{eq:c2} holds from the fact that users may discard outdated tasks due to congestion in the channel, or a need to prioritize more urgent tasks~\cite{yue2021todg}.
% The inequality in Eq. \eqref{eq:c2} being true indicates that user $n$ does not offload any tasks.
}
Edge servers are connected to each other via a wired network and can collaborate to execute offloaded tasks by assigning a portion of tasks from one server to another. We represent $\eta^s_{m, m'}$ as the proportion of tasks delivered from server $m$ to server $m'$, which satisfies
\begin{align}
    \label{eq:c3a}
    0 \leq \eta^s_{m,m'} \leq 1, &\quad \forall m, m' \in \mathcal{M},\\
    \label{eq:c5a}
    \sum_{m' \in \mathcal{M}} \eta^s_{m,m'} = 1, &\quad \forall m \in \mathcal{M}.
\end{align}
% It is easy to see the following constraint should also be met:
% \begin{equation}
% \label{eq:c5a}
%      \sum_{m' \in \mathcal{M}} \eta^s_{m,m'} = 1 , \quad \forall m \in \mathcal{M}.
% \end{equation}
Accordingly, the number of computational tasks with priority $\delta$ delivered from server $m$ to $m'$ can be expressed as $\eta^s_{m,m'} \lambda^{s}_{\delta,m}$, where 
\begin{equation}
\label{eq:lambda-s}
\lambda^{s}_{\delta,m}\doteq\sum_{n\in\mathcal{N}^{\delta}_m}\sum_{c\in\mathcal{C}} p_{n,c} \lambda_{n,c}.
\end{equation}
$\lambda^{s}_{\delta,m}$ denotes the total number of received tasks prioritized as $\delta$ at server $m$. Since the total number of tasks arrived at each server cannot exceed its maximum capacity (denoted by $\lambda^{s,\mathrm{max}}_m$), we have
\begin{equation}
\label{eq:c4a}
    \sum_{\delta \in \Delta}\sum_{m' \in \mathcal{M}} \eta^s_{m',m} \lambda^{s}_{\delta,m'} \leq \lambda^{s,\mathrm{max}}_m, \quad \forall m \in \mathcal{M}. 
\end{equation}
%Thus, from $m$ to $m'$, the number of delivered tasks with $\delta$-th priority is $\eta^s_{m,m'} \lambda^{s}_{\delta,m}$, where $\lambda^{s}_{\delta,m}=\sum_{n\in\mathcal{N}^{\delta}_m}\sum_{c\in\mathcal{C}} p_{n,c} \lambda_{n,c}$ denote the total number of received tasks with $\delta$-th priority at server $m$.
% For server collaboration, scheduling tasks with priority level $\delta$ from server $m$ to $m'$ is represented by  $\eta^s_{m,m'} \lambda^{s}_{\delta,m}$. 

\subsubsection{Transmission Model}
% \subsubsection{Communication between user and server}
{{% Besides the assumption of an M/G/1 queuing model with FCFS (first-come-first-service) for server computation,
We assume that the transmission process of the multiple access channel follows the M/M/1 queuing model, as in~\cite{wang2020latency,ren2022efficient}.\footnote{Task arrivals follow a Poisson process, and channel transmission times are exponential. It's like an M/M/1 queue, where tasks arrive randomly, wait in a queue, and are served one at a time, with exponential service times.}$^,$\footnote{The M/M/1 wireless transmission queue model can be simplified to the M/G/1 model when their second moments are equal, making them equivalent.}
%\footnote{The contention and queuing process for multiple tasks for each channel can be modeled as an M/M/1 process, where task arrivals to a certain channel follow a Poisson process and the channel transmission time (service time) for each task follows an exponential distribution.}
%\footnote{M/M/1 queueing model for wireless transmission can be extended to M/G/1 queueing model. When the second moment of the M/M/1 queueing model is equal to the second moment of the M/G/1 queueing model, the two models are considered equivalent.}
We denote the task arrival rate of channel $c$ as $\lambda_c\doteq\sum_{n\in\mathcal{N}}\eta^u_{n,c}\lambda_{n}$, \ie, the sum of the offloading arrival rates of all users accessing this channel. 
% is the rate at which tasks are offloaded through channel $c$.
% The service rate of channel $c$ is $r_c/S$, where
% $r_c$ and $S$ represent the communication rate and the size of each package transmitted through channel $c$, respectively.
Meanwhile, $r_c$ stands for the communication for transmission rate through channel $c$, and $S$ denotes the package size.
} As the currently achievable wireless communication rate approaches the Shannon limit~\cite{ghanem2020resource}, the transmission rate can be expressed as $r_c = B_c \log \bigl(R_c + 1 \bigr)$, where $B_c$ and $R_c$ are the bandwidth and the signal-to-noise ratio of channel $c$, respectively. Thus, if user $n$ accesses channel $c$, the expected transmission time can be expressed by
\begin{equation}  
\label{eq:TransmissionDelay-Channel}
    T^\mathrm{tr}_{n,c} = \left\{
    \begin{array}{ll}
     \frac{1}{r_c/S - \lambda_c} + t_{n,c}, &\quad \frac{r_c}{S} - \lambda_c > 0,\\
     \infty , &\quad \mathrm{otherwise},
    \end{array}
    \right.	
\end{equation}
where $\frac{1}{r_c/S - \lambda_c}$ is the channel contention time, and {$t_{n,c}$ is the constant end-to-end propagation delay of the offloading tasks.\footnote{Since $d_{n,c}$ does not change significantly relative to the speed of light $v_{c}$, $t_{n,c} = \frac{d_{n,c}}{v_{c}}$ can be regarded as a very small quantity or a constant.}
The propagation delay, in a short time slice, is calculated by dividing the distance between the user and the server ($d_{n,c}$) by the speed at which the wireless signal propagates through the air ($v_{c}$), \ie, $t_{n,c} = \frac{d_{n,c}}{v_{c}}$ \cite{callegati1999packet}.
As in~\cite{mao2017survey}, since MEC servers are linked via wired core networks, we assume that the end-to-end propagation delay between servers $m$ and $m'$ is constant, denoted by $t^\mathrm{tr}_{m, m'}$. 
}
}

% from user $n$ via channel $c$ to local server.	
 
% \subsubsection{Communication between servers}

\subsection{Limited Capacity Model with Confidence Evaluation\label{subsec:capacity_model}}
{To meet the requirements of real-time response, we consider the limited channel capacity/computation models with confidence, which enable real-time estimation and the control of system stability~\cite{guo2021scheduling,zhou2018channel}. 
% Considering that our study addresses real-time and practical scenarios, we consider limited (Channel/Computation) capacity model with confidence, which offer dual advantages: real-time estimation and the control of system stability~\cite{guo2021scheduling,zhou2018channel}.  
}
\subsubsection{Channel Capacity}
% The effective re-consumption of the channel by competing activities often results in a decrease in the freshness of the information in many applications with high real-time needs. We provide channel confidence to lessen the possible issue of excessive task inflow and to distribute tasks more equally across many heterogeneous channels since the in-channel arrival tasks follow a Poisson distribution. 

Denote the capacity of channel $c$ as $M_c^\mathrm{max}$. 
% Due to the limited channel capacity, the total number of tasks transmitted through a channel cannot exceed its maximum capacity ($M_c^\mathrm{max}$). 
Based on the properties of the cumulative distribution function of the Poisson distribution~\cite{patil2012comparison}, we introduce the following constraint to ensure that each channel is conflict-free with confidence level $1-\alpha$~\cite{zhou2018channel,guo2021scheduling}:
\begin{equation}
    \frac{z^2_1}{2} +  z_2\cdot\left(\sum_{n\in\mathcal{N}} \eta^u_{n,c}\lambda_{n} + \frac{z^2_2}{4}\right)^\frac{1}{2} + \sum_{n\in\mathcal{N}} \eta^u_{n,c}\lambda_{n}  \le M_c^\mathrm{max}.
    % , \quad \forall c. 
    \label{eq:c7a}
\end{equation}
Here, $z_1 \doteq (\alpha/2 - \lambda^{\mathrm{chl}}_c)/\sqrt{\lambda^{\mathrm{chl}}_c}$ and $z_2 \doteq (1-\alpha/2 - \lambda^{\mathrm{chl}}_c)/\sqrt{\lambda^{\mathrm{chl}}_c}$ are two statistics standardized from a normal distribution,
% satisfying 
% % $P\{\alpha < z_1\} = 1 - \Phi ((2\alpha - \lambda^{\mathrm{chl}}_c)/\sqrt{\lambda^{\mathrm{chl}}_c})$
% $z_1 = (\alpha/2 - \lambda^{\mathrm{chl}}_c)/\sqrt{\lambda^{\mathrm{chl}}_c}$
% and 
% $z_2 = (1-\alpha/2 - \lambda^{\mathrm{chl}}_c)/\sqrt{\lambda^{\mathrm{chl}}_c}$,
% respectively. 
where $\lambda^{\mathrm{chl}}_c \doteq \sum_{n\in\mathcal{N}} \eta^u_{n,c}\lambda_{n}$ represents the total number of tasks transmitted through channel $c$. 

\begin{figure*}[t]
\begin{align}
% \begin{split}
    \mathbb{E}[W^{p}_{n}] =\,& \frac{\frac{1}{2}\sum_{\delta\in\Delta}\sum_{n'\in\mathcal{N}^{\delta}}\sum_{c\in\mathcal{C}}p_{n',c}\eta^u_{n',c}\lambda_{n'}\nu_{n'}}{\left(1-\sum_{\delta\in\Delta(\delta(n))}\sum_{n'\in\mathcal{N}^{\delta}}\sum_{c\in\mathcal{C}}p_{n',c}\frac{\eta^u_{n',c}\lambda_{n'}}{\mu_{n'}}\right)\left(1-\sum_{\delta\in\Delta(\delta(n)-1)}\sum_{n'\in\mathcal{N}^{\delta}}\sum_{c\in\mathcal{C}}p_{n',c}\frac{\lambda_{n',c}}{\mu_{n'}}\right)}\label{eq:paoi-priority-waitingtime-expectation}\\\nonumber\\
    \pi_{n,m}(\bm\eta^s) \doteq\,& \frac{\sum_{\delta\in\Delta} \eta^s_{m_n,m} \lambda^{s}_{\delta,m_n}}{\sum_{\delta\in\Delta} \sum_{m'\in\mathcal{M}} \eta^s_{m_n,m'} \lambda^{s}_{\delta,m_n}} \cdot
    \left(
    t^\mathrm{tr}_{m_n,m} + \frac{1}{\mu_{n,m}}+\frac{1}{1-\sum_{\delta\in{\Delta(\delta(n))}}\sum_{m'\in\mathcal{M}} \eta^s_{m',m} \lambda^{s}_{\delta,m'} \frac{1}{\mu_{n,m}}}\right.\nonumber\\
    &\left.\cdot\frac{\frac{1}{2}\sum_{\delta\in\Delta} \sum_{m'\in\mathcal{M}} \eta^s_{m',m} \lambda^{s}_{\delta,m'} \nu_{n,m} }{1-\sum_{\delta\in{\Delta(\delta(n)-1)}}\sum_{m'\in\mathcal{M}} \eta^s_{m',m} \lambda^{s}_{\delta,m'} \frac{1}{\mu_{n,m}}}\right)
% \end{split}
\label{eq:multi-server-DataMigration}
\end{align}
\noindent\makebox[\linewidth]{\rule{\textwidth}{0.5pt}}
\end{figure*}

\subsubsection{Computation Capacity} 
% Considering the dynamic changes in the computing power of each server \comment{polish it},
{Due to the fluctuation in servers' available computing capacities, we assume that server $m$'s processing time for user $n$'s tasks follows a general distribution with mean $1/\mu_{n,m}$ and second moment $\nu_{n,m}$~\cite{huang2015optimizing}.
\footnote{With task arrivals following a Poisson distribution, we consider an M/G/1 queueing system for server computations that adhere to an FCFS manner.}
To simplify notations, we use $1/\mu_n$ and $\nu_n$ to represent $1/\mu_{n,m_n}$ and $\nu_{n,m_n}$, respectively.\footnote{$1/\mu_{n,m_n}$ and $\nu_{n,m_n}$ represent the mean and second moment, respectively, of the processing time for user $n$'s tasks on its local server $m_n$.}}
% We let $\frac{1}{\mu_n}$ and $\nu_n$ represent the mean and second moment of the general distribution of tasks performed on the local server, respectively. \comment{?}
% we assume that each task offloaded by user $n$ is executed according to an unknown distribution with mean $\frac{1}{\mu_n}$ and second moment $\nu_n$. 
To ensure each task is executed with a confidence level of $1 - \beta$, we impose the following constraint:
%Analogous to channel capacity, we provide the following constraint to guarantee that the task can be executed with a confidence level of $(1 - \beta)$:
% Without loss of generality, we assume that the processing time of each update packet from entity $n$ obeys general distribution with mean $\frac{1}{\mu_n}$ and second moment $\nu_n$.
% Similarly, we have constraints in order to ensure that the task can be completed with a confidence level of $(1 - \beta)$.
\begin{equation}
    \label{eq:c8a}
\begin{split}
     \frac{z^2_{3}}{2} &+ z_{4}\cdot\left(\sum_{n\in\mathcal{N}} \sum_{c\in\mathcal{C}} p_{n,c}\cdot \frac{\eta^u_{n,c}\lambda_{n}}{\mu_n} +\frac{ z^2_{4}}{4}\right)^\frac{1}{2} \\    &+\sum_{n\in\mathcal{N}} \sum_{c\in\mathcal{C}} p_{n,c}\cdot \frac{\eta^u_{n,c}\lambda_{n}}{\mu_n}  <1,
\end{split}
\end{equation}
\noindent where $z_3 \doteq (\beta/2 - \lambda^{\mathrm{comp}}_c)/\sqrt{\lambda^{\mathrm{comp}}_c}$
and 
$z_4 \doteq (1-\beta/2 - \lambda^{\mathrm{comp}}_c)/\sqrt{\lambda^{\mathrm{comp}}_c}$ are two statistics standardized from a normal distribution, with $\lambda^{\mathrm{comp}} \doteq \sum_{n\in\mathcal{N}} \sum_{c\in\mathcal{C}} p_{n,c} \eta^u_{n,c}\lambda_{n}/\mu_n$ the total number of tasks arriving at server $m(n)$. 

% where $\beta_1 = \beta/2$, $\beta_2 = 1 - \beta/2$,
% $\lambda^{\mathrm{comp}} = \sum_{n\in\mathcal{N}} \sum_{c\in\mathcal{C}} p_{n,c}\cdot \eta^u_{n,c}\lambda_{n}/\mu_n$
% and $P\{\beta < z(\beta)\} = 1 - \Phi ((\beta - \lambda^{\mathrm{comp}})/\sqrt{\lambda^{\mathrm{comp}}})$. \comment{Revise according to the above}

% \begin{figure*}[t]
%     \begin{align}
%     \label{eq:paoi-priority-waitingtime-expectation}
%     \mathbb{E}[W^{p}_{n}] = \frac{\frac{1}{2}\sum_{\delta\in\Delta}\sum_{n'\in\mathcal{N}^{\delta}}\sum_{c\in\mathcal{C}}p_{n',c}\eta^u_{n',c}\lambda_{n'}\nu_{n'}}{\left(1-\sum_{\delta\in\Delta(\delta(n))}\sum_{n'\in\mathcal{N}^{\delta}}\sum_{c\in\mathcal{C}}p_{n',c}\frac{\eta^u_{n',c}\lambda_{n'}}{\mu_{n'}}\right)\left(1-\sum_{\delta\in\Delta(\delta(n)-1)}\sum_{n'\in\mathcal{N}^{\delta}}\sum_{c\in\mathcal{C}}p_{n',c}\frac{\lambda_{n',c}}{\mu_{n'}}\right)}
%     \end{align}
% % \noindent\makebox[\linewidth]{\rule{\textwidth}{0.5pt}}
% \end{figure*}

\subsection{Freshness Model and Problem Formulation}
We characterize the freshness of information using the PAoI metric. Specifically, the PAoI of user $n$'s message, denoted by $A_{n}$, is determined by four key factors: transmission time $T_{n}$, arrival interval $I_{n}$, waiting time $W_{n}$, and processing time $Y_{n}$~\cite{zou2021minimizing}, \ie, 
\begin{align}
    \label{eq:paoi-expectation2}
    \mathbb{E}[A_{n}] &= \mathbb{E}[T_{n}]+\mathbb{E}[I_{n}]+\mathbb{E}[W_{n}]+\mathbb{E}[Y_{n}],
\end{align}
% $\mathbb{E}[T_{n}] = T^\mathrm{tr}_{n,c}$ \comment{revise}
% It is straightforward to figure out that 
with $\mathbb{E}[I_{n}] = 1/\sum_{c\in\mathcal{C}}p_{n,c}\eta^u_{n,c}\lambda_{n}$ and $\mathbb{E}[T_{n}] = \sum_{c\in\mathcal{C}} \eta^u_{n,c} T^\mathrm{tr}_{n,c}$. 
The expressions of $\mathbb{E}[W_{n}]$ and $\mathbb{E}[Y_{n}]$ are contingent upon the way in which user $n$'s tasks are executed. To derive the precise expressions, we use binary variable $y_{m}$ to denote the migration decision of server $m$, \ie, 
\begin{equation}
\label{eq:c-y}
    y_m \in \{0,1\}, \quad \forall m \in \mathcal{M}.
\end{equation}
%local offloading 
{
If $y_m$ comes to 0, server $m$ executes its tasks locally; 
otherwise, it resorts to the other servers for collaboration.
In a slightly abusive notation, we use the superscript `p' to represent the case when $y_m=0$ and the superscript `s' to represent the case when $y_m=1$.
Thus, when $y_{m_n}=0$, $\mathbb{E}[Y^{p}_{n}]=1/\mu_n$, and $\mathbb{E}[W^{p}_{n}]$ can be expressed in Eq.~\eqref{eq:paoi-priority-waitingtime-expectation} from the M/G/1 queuing model with FCFS \tr{(see Appendix \ref{app:pro:paoi_mg1_pri} of technical report~\cite{qiao2023popec} for more details).}\fp{(see Appendix A of technical report~\cite{qiao2023popec} for more details).}
When $y_{m_n}=1$, according to Little's Law\footnote{Little's law in queuing theory states that the average number of customers in a stationary system equals the product of arrival rate and waiting time~\cite{kim2013statistical}.}, we can obtain
\begin{equation}
    \mathbb{E}[W^s_{n}]+\mathbb{E}[Y^s_{n}]=\sum_{m\in\mathcal{M}} \pi_{n,m}(\bm\eta^s),
\end{equation}
where $\bm\eta^s\doteq\{\eta^s_{m, m'}\}_{m, m'\in \mathcal{M}}$, and $\pi_{n,m}(\bm\eta^s)$ is defined in Eq.~\eqref{eq:multi-server-DataMigration}. Based on this, the expected PAoI of user $n$ can be specified as follows:
\begin{align}
\label{eq:pomis}
    \mathbb{E}[A_{n}]
    =  (1 - y_{m_n})\cdot\mathbb{E}[A_{n}|y_{m_n}=0] + y_{m_n}\cdot\mathbb{E}[A_{n}|y_{m_n}=1].
    % = & (1 - y_{m_n})\mathbb{E}[A^p_{n}] + y_{m_n}\mathbb{E}[A^s_{n}].
\end{align}
}

Given the aforementioned constraints, our objective is to find the optimal offloading decision for users and the collaboration decision for servers to minimize the average expected PAoI across all users. 
{
We cast the problem as the \textit{\underline{P}AoI-Centric Task \underline{O}ffloading with \underline{P}riority over Unr\underline{e}liable \underline{C}hannels (PoPeC)}:
\begin{equation}
\label{eq:P}
\begin{split}
    &\textbf{(PoPeC)}~\mathop{\min}_{\bm\eta^u, \bm\eta^s, \bm y}
    \frac{1}{N}\sum_{n\in\mathcal{N}} \mathbb{E}[A_n], \\
    &\qquad\quad\text{s.t.}~ \eqref{eq:c1}\textrm{--}\eqref{eq:c4a},
    % \eqref{eq:c6},
    \eqref{eq:c7a},\eqref{eq:c8a},\eqref{eq:c-y},
\end{split}
\end{equation}
with $\bm\eta^u\doteq\{\eta^u_{n,c}\}_{n\in \mathcal{N}, c\in \mathcal{C}}$ the offloading decision, and $\bm\eta^s\doteq\{\eta^s_{m, m'}\}_{m, m'\in \mathcal{M}}$ along with $\bm{y}\doteq\{y_m\}_{m\in\mathcal{M}}$ the collaboration decision.
}

% \begin{equation}
% \label{eq:P}
% \begin{split}
%     \textbf{(PoPeC)}~&\mathop{\min}_{\bm\eta^u, \bm\eta^s, \bm y}
%     \frac{1}{N}\sum_{n\in\mathcal{N}} F_n(\bm\eta^u, \bm y, \bm\eta^s) \\
%     \text{s.t.}&~ \eqref{eq:c1},\eqref{eq:c2},\eqref{eq:c3a},\eqref{eq:c5a},\eqref{eq:c4a},
%     % \eqref{eq:c6},
%     \eqref{eq:c7a},\eqref{eq:c8a},\eqref{eq:c-y},
% \end{split}
% \end{equation}

% where $F_n(\bm\eta^u, \bm y, \bm\eta^s)=(1 - y_{m_n}) f^p_n(\bm\eta^u) + y_{m_n} f^s_n(\bm\eta^u, \bm\eta^s)$, $\bm y = \{y_{m}\}$, $f^p_n(\bm\eta^u) = \mathbb{E}[A^p_n]$, and $f^s_n(\bm\eta^s) = \mathbb{E}[A^s_n]$.

% \begin{figure*}[t]
% 	\centering
% 	\subfigure[PoSiP]{\label{fig:partA}\includegraphics[width=0.26\textwidth]{./figure/SystemFig1-a.pdf}}
% 	\subfigure[PoMiP]{\label{fig:partB}\includegraphics[width=0.35\textwidth]{./figure/SystemFig1-b.pdf}}
% 	\subfigure[PoMiS]{\label{fig:partC}\includegraphics[width=0.35\textwidth]{./figure/SystemFig1-c.pdf}}
% 	\caption{Exemplary illustrations of three scenarios based on the edge network model.}
% 	\label{fig:SystemModel}
% \end{figure*}

\section{PoPeC: PAoI-Centric Task Offloading with Priority over Unreliable Channels}\label{sec:PoPeC}

In light of the challenges in directly addressing the offloading problem, this section begins by examining two special cases -- priority-free and multi-priority task scheduling with no server collaboration. Building on the insights gained from these solutions, we subsequently devise an effective and efficient algorithm for optimizing the original problem.

\subsection{Priority-Free Task Scheduling
}{\label{subsec:POSIP}}
\subsubsection{Problem Transformation}

We first focus on a special case in which all tasks are of the same type and do not possess any priority distinctions, \ie, the problem of \textit{offloading to the local servers}. To solve this problem, the first step is to derive the expression of the expected PAoI for each user in this context.
% considering only the process of user offloading to the local server. In this case, it makes sense to first find out what each user's PAoI expectation is for the system.
% The usual situation with no priority is equal to such a priority-free scenario.
% Across such a priority-free and local-server-only example, tasks from various but priority-neutral users are delivered to the local server in a series of heterogeneous and unreliable channels.
% In this situation, finding out what each user's PAoI expectations are for the system is of relevance.
Recall that each user's task arrivals follow a Poisson distribution, and task execution times follow a general distribution. 
{
We re-evaluate the waiting time based on Eq.~\eqref{eq:paoi-priority-waitingtime-expectation} as follows: 
\begin{equation}
    \mathbb{E}[W_{n}]=\frac{\sum_{n\in\mathcal{N}}\sum_{c\in\mathcal{C}}p_{n,c}\eta^u_{n,c}\lambda_{n}\nu_n}{2\cdot\left(1-\sum_{n\in\mathcal{N}}\sum_{c\in\mathcal{C}}p_{n,c}\cdot\frac{\eta^u_{n,c}\lambda_{n}}{\mu_n}\right)}.
\end{equation}
}
We introduce an auxiliary variable, $\hat t^\mathrm{tr}_n>0$, to represent the upper bound of user $n$'s transmission time through all the available channels, \ie
\begin{equation}
\label{eq:c6}
    \hat t^\mathrm{tr}_n \ge  T^\mathrm{tr}_{n,c} , \quad \forall c \in \mathcal{C}. 
\end{equation}
% Intuitively, it holds 
%  $\mathbb{E}[T_{n}]=\hat t^\mathrm{tr}_n$.

Akin to~\cite{he2016optimal,du2017computation}, we transform the original problem into the following tractable form, \ie, optimizing a tight upper bound of user $n$'s expected PAoI: 
\begin{equation}
\label{eq:P1}
\begin{aligned}
    \textbf{(P1)}~~&\mathop{\min}_{\hat t^\mathrm{tr}_n,\bm\eta^u} \frac{1}{N}\sum_{n\in\mathcal{N}} f_n(\hat t^\mathrm{tr}_n,\bm\eta^u), \\
    \text{s.t.}&~ \eqref{eq:c1},\eqref{eq:c2},\eqref{eq:c7a},\eqref{eq:c8a},\eqref{eq:c6}.
\end{aligned}
\end{equation}
$f_n(\hat t^\mathrm{tr}_n,\bm\eta^u)$ is defined as:
\begin{align}
\label{eq:paoi-expectation1.1}
    f_n(\hat t^\mathrm{tr}_n,\bm\eta^u)
    % \doteq\,&\frac{\sum_{n\in\mathcal{N}}\sum_{c\in\mathcal{C}}p_{n,c}\eta^u_{n,c}\lambda_{n}\nu_n}{2\cdot\left(1-\sum_{n\in\mathcal{N}}\sum_{c\in\mathcal{C}}p_{n,c}\cdot\frac{\eta^u_{n,c}\lambda_{n}}{\mu_n}\right)}\nonumber\\
    % &+ \hat t^\mathrm{tr}_n + \frac{1}{\sum_{c\in\mathcal{C}}p_{n,c}\eta^u_{n,c}\lambda_{n}} + \frac{1}{\mu_n}\nonumber\\
    \doteq\,&\hat t^\mathrm{tr} + \frac{1}{\Psi_n(\bm\eta^u)}
	    + \frac{\upsilon_n(\bm\eta^u)}{\phi_n(\bm\eta^u)}
	    + \frac{1}{\mu},
\end{align}
where we represent $\upsilon_n(\bm\eta^u) \doteq \frac{1}{2}\sum_{n'\in\mathcal{N}}\sum_{c\in\mathcal{C}}p_{n',c}\cdot\eta^u_{n',c}\lambda_{n'}\nu$, $\phi_n(\bm\eta^u) \doteq 1-\sum_{n'\in\mathcal{N}}\sum_{c\in\mathcal{C}}p_{n',c}\cdot\eta^u_{n',c}\lambda_{n'}/\mu$, 
    % $\bm\eta^u:=\{\eta^u_{n,c}|n\in \mathcal{N}, c\in \mathcal{C}, \eta^u_{n,c} \in [0,1]\}$,
    and $\Psi_n(\bm\eta^u)\doteq\sum_{c\in\mathcal{C}}p_{n,c}\eta^u_{n,c}\lambda_{n}$.

% {In the same spirit of min-max optimization, we can increase the communication efficiency across all channels by minimizing $\hat{t}^\mathrm{tr}_n$ instead of $\mathbb{E}[T_n]$, thereby enhancing the robustness of our proposed solution.}
In line with the principles of min-max optimization, we can enhance the overall communication efficiency across all channels by focusing our efforts on minimizing $\hat{t}^\mathrm{tr}_n$ rather than $\mathbb{E}[T_n]$.  This strategic shift not only aligns with our objective of improving performance but also bolsters the resilience and robustness of our proposed solution.

Eq.~\eqref{eq:c7a} and Eq.~\eqref{eq:c8a} suggest that \textbf{P1} is a non-convex problem. To identify well-structured solutions, we transform these constraints into equivalent convex ones.

% In order to identify well-structured solutions, it is necessary to undertake certain transformations and assumptions on \textbf{P1}. Due to the non-convex nature of \eqref{eq:c7a} and \eqref{eq:c8a}, \textbf{P1} does not seem to be a convex problem in the current form.
\begin{lemma}
\label{lemma:P1-1_convex_problem}
% We apply some equivalent alterations to convert the constraints, \eqref{eq:c7a} and \eqref{eq:c8a}, into convex set \eqref{eq:c7b} and \eqref{eq:c8b} via some equivalent transformations.
In Problem \textbf{P1}, Constraints \eqref{eq:c7a} and \eqref{eq:c8a} are equivalent to following Constraints \eqref{eq:c7b} and \eqref{eq:c8b} respectively:
\begin{align}
    &\sum_{n\in\mathcal{N}} \eta^u_{n,c}\lambda_{n} \le M_c^\mathrm{max} + \frac{z^2_{2}}{2} - \frac{z^2_{1}}{2} - z_{2}\cdot\left(M_c^\mathrm{max} + \frac{z^2_{2}}{2} - \frac{z^2_{1}}{2}\right)^\frac{1}{2},
    \label{eq:c7b}\\
    &\sum_{n\in\mathcal{N}} \sum_{c\in\mathcal{C}} p_{n,c} \frac{\eta^u_{n,c}\lambda_{n}}{\mu_n} \leq 1 + \frac{z^2_{4}}{2} - \frac{z^2_{3}}{2} - z_{4}\cdot\left(1 + \frac{z^2_{4}}{2} - \frac{z^2_{3}}{2}\right).
    \label{eq:c8b}
\end{align}
% where $\frac{z^2_{2}}{2} - \frac{z^2_{1}}{2} = \frac{z^2_{2}}{2} - \frac{z^2_{1}}{2}$ and $\frac{z^2_{4}}{2} - \frac{z^2_{3}}{2} = \frac{z^2_{4}}{2} - \frac{z^2_{3}}{2}$
\end{lemma}
\begin{proof}
  \tr{The detailed proof can be found in Appendix \ref{app:lemma:P1-1_convex_problem} of technical report~\cite{qiao2023popec}.}\fp{The detailed proof can be found in Appendix B of technical report~\cite{qiao2023popec}.}
\end{proof}

{According to Lemma \ref{lemma:P1-1_convex_problem}, Constraints \eqref{eq:c7b} and \eqref{eq:c8b} restrict the decision variables to convex sets, and hence, Problem \textbf{P1} can be rewritten as:
% Besides, we assume updates from any users are provided at the service time of $1/\mu$ and a second moment of $\nu$.
% This assumption is sufficient to describe systems where the users deliver tasks with uniformly distributed lengths but various timeliness demands~\cite{yates2018age}.
% \begin{equation}
% \label{eq:P1-1}
% \begin{aligned}
%     \textbf{(P1-1)}~~&\mathop{\min}_{\hat t^\mathrm{tr}_n, \bm\eta^u} \frac{1}{N}\sum_{n\in\mathcal{N}} \left( \hat t^\mathrm{tr} + \frac{1}{\Psi_n(\bm\eta^u)}
% 	    + \frac{\upsilon_n(\bm\eta^u)}{\phi_n(\bm\eta^u)}
% 	    + \frac{1}{\mu}\right) \\
%     \text{s.t.}&~ \eqref{eq:c1},\eqref{eq:c2},\eqref{eq:c6},\eqref{eq:c7b},\eqref{eq:c8b}
% \end{aligned}
% \end{equation}
\begin{equation}
\label{eq:P1-1}
\begin{aligned}
    \textbf{(P1-1)}~~&\mathop{\min}_{\hat t^\mathrm{tr}_n,\bm\eta^u} \frac{1}{N}\sum_{n\in\mathcal{N}} f_n(\hat t^\mathrm{tr}_n,\bm\eta^u), \\
    \text{s.t.}&~ \eqref{eq:c1},\eqref{eq:c2},\eqref{eq:c6},\eqref{eq:c7b},\eqref{eq:c8b}.
\end{aligned}
\end{equation}
}

Building on the above transformation, we can obtain the following property.
% where we denote $\upsilon_n(\bm\eta^u) \doteq \frac{1}{2}\sum_{n'\in\mathcal{N}}\sum_{c\in\mathcal{C}}p_{n',c}\cdot\eta^u_{n',c}\lambda_{n'}\nu$, $\phi_n(\bm\eta^u) \doteq 1-\sum_{n'\in\mathcal{N}}\sum_{c\in\mathcal{C}}p_{n',c}\cdot\eta^u_{n',c}\lambda_{n'}/\mu$, 
%     % $\bm\eta^u:=\{\eta^u_{n,c}|n\in \mathcal{N}, c\in \mathcal{C}, \eta^u_{n,c} \in [0,1]\}$,
%     and $\Psi_n(\bm\eta^u)\doteq\sum_{c\in\mathcal{C}}p_{n,c}\eta^u_{n,c}\lambda_{n}$.
\begin{theorem}
\label{thm:admm-consensus}
    Problem \textbf{P1-1} is a convex problem. 
    % We can obtain the optimal solution to the \textbf{P1-2} problem by solving \textbf{P1-3} in a distributed manner with the ADMM-Consensus algorithm.
\end{theorem}
\begin{proof}
    It is easy to see that the constraints of \textbf{P1-1} are convex because they are affine sets. \tr{The convexity of all sub-functions is proven in Appendix \ref{app:lemma:paoi_mg1} of technical report~\cite{qiao2023popec}.}\fp{The convexity of all sub-functions is proven in Appendix C of technical report~\cite{qiao2023popec}.}
\end{proof}
% Theorem \ref{thm:admm-consensus} implies that we can take advantage of the fact that the sub-problems are all convex and adopt the ADMM method to solve the problem \textbf{P1-1}.
% Based on the above, we have the following part.
Based on the convexity of Problem \textbf{P1-1}, we can employ the ADMM technique to solve this problem, detailed in the following section.

% \begin{align}
% \label{eq:P1-2}
%     \textbf{(P1-1)}~&\mathop{\min}_{\hat t^\mathrm{tr},\bm x}\frac{1}{N}\sum_{n\in\mathcal{N}} f_n(\hat t^\mathrm{tr},\bm x) \\
%     \text{s.t.}&~ \eqref{eq:c1},\eqref{eq:c2},\eqref{eq:c6},\eqref{eq:c7b},\eqref{eq:c8b}
% \end{align}
% In other word, \textbf{(P1-1)} is equivalent to \textbf{(P1)}. 
    
\subsubsection{Problem Decomposition and Solving}\label{subsubsec:POSIP2}
We define an auxiliary function, $g_n(\bm x)$, as follows:
\begin{align}
\label{eq:feasible_solution1}
    g_n(\bm x) \doteq \left\{\begin{array}{ll}
    \hat t^\mathrm{tr} + \frac{1}{\Psi_n(\bm\eta^u)}
	    + \frac{\upsilon_n(\bm\eta^u)}{\phi_n(\bm\eta^u)}
	    + \frac{1}{\mu}, & \bm x \in \Omega, \\
    \infty, & \text {otherwise.}
    \end{array}\right.
\end{align}
where $\bm x \doteq \{\hat t^\mathrm{tr}, \bm\eta^u\}$, and $\Omega\doteq\{\bm x|\eqref{eq:c1},\eqref{eq:c2},\eqref{eq:c6},\eqref{eq:c7b},\eqref{eq:c8b}\}$ is the feasible set of Problem \textbf{P1-1}. 
% In addition, we use the subscripts `n' and `o' to indicate whether the variable $\bm x$ is controlled on the user $n$ side or the server side, respectively.

\textbf{P1-1} is equivalent to the following consensus problem:
\begin{align}
    \label{PB:P1-2}
    \textbf{(P1-2)}~\mathop{\min}_{\{\bm x_n\}}~~&\sum_{n\in\mathcal{N}} g_n(\bm x_n), \\
    \label{eq:P1-2-consensus}
    \text{s.t.}&~\bm x_n = \bm x_{o}.
\end{align}

According to Theorem \ref{thm:admm-consensus}, it is clear that the well-known ADMM-Consensus algorithm can be used to obtain the optimal solution of Problem \textbf{P1-2}\cite{boyd2011distributed} \tr{(For more details, please refer to Appendix \ref{app:det:admm-consensus} of technical report~\cite{qiao2023popec})}\fp{(For more details, please refer to Appendix D of technical report~\cite{qiao2023popec})}.
    % since the objective function is convex and the feasible constraint is convex set, \textbf{P1-3} is a convex problem. According to~\cite{boyd2011distributed}, we can obtain the optimal solution distributively by means of the ADMM-Consensus algorithm. 

\subsection{Multi-Priority Task Scheduling\label{subsec:POMIP}}
\subsubsection{Problem Transformation with Nonlinear Fractional Programming}
\label{subsubsec:NFP}
{Different from priority-free task scheduling in the previous section, we delve deeper into the multi-priority task scheduling problem in this section.
% Based on their unique characteristics, users with multi-class priority are divided into different priority classes.
In this case, according to the priorities of received tasks, servers will execute the tasks with higher priorities more promptly. Since the multi-priority task scheduling, in this case, is NP-hard
% (see~\cite{moghadas2011queueing}).
\tr{(please refer to Appendix \ref{app:NP-hard} of technical report~\cite{qiao2023popec} for the proof)}\fp{(please refer to Appendix M of technical report~\cite{qiao2023popec} for the proof)}, 
% Due to the NP-Hardness of the M/G/1 model with multi-class priorities~\cite{moghadas2011queueing}, we cannot ignore the need to address it. 
similar to Problem \textbf{P1}, we minimize a tight upper bound for the average expected PAoI of multi-priority users: }
% In this subsection, we delve deeper into the task offloading problem of users with multi-class priority over unreliable channels.  With the concept of multi-class priority, users are categorized into distinct priority classes based on their unique characteristics.  This enables the system to prioritize and complete tasks from high-priority users promptly.  Since users with different priorities have statistically significant differences in their task types, the execution times of tasks performed by different users are distributed independently but differently.  Using the problem formulation in the system model and taking into account multi-class priority, we arrive at the problem in Case 2 as follows
% In this subsection, we further investigate the task offloading problem of multi-class priority users in unreliable channels.
% In the light of multi-class priority, users are regarded as distinct classes priority based on their unique features and characteristics.
% Thanks to this mechanism, tasks from high-class priority users will be completed as promptly as possible.
% Given that the task types of users with varying priorities are statistically significantly inconsistent, the execution times of tasks performed by different users are distributed independently but differently.
% Based on the discussion of multi-class priority and problem formulation in the system model, we obtain the problem in PoMiP as
\begin{equation}
\label{eq:P2}
\begin{aligned}
    \textbf{(P2)}~~&\mathop{\min}_{\bm x} \frac{1}{N}\sum_{n\in\mathcal{N}} f^p_n(\bm x), \\
    \text{s.t.}&~ \eqref{eq:c1},\eqref{eq:c2},\eqref{eq:c7a},\eqref{eq:c8a},\eqref{eq:c6},
\end{aligned}
\end{equation}
where 
% $\bm x = \{\hat t^\mathrm{tr}_n,\bm\eta^u\}$,
$f^{p}_n(\bm x) = \hat t^\mathrm{tr} 
+\frac{1}{\mu_n} 
+ \frac{1}{\Psi_n(\bm x)}
+ \frac{\Upsilon(\bm x)}
{\Phi_{\delta(n)}(\bm x)
\Phi_{\delta(n)-1}(\bm x)}$ is the upper bound of $\mathbb{E}[A_n]$,
$\bm x = \{\hat t^\mathrm{tr},\bm\eta^u\}$ denotes the decision variables,
$\Phi_{\delta(n)}(\bm x) = 1-\sum_{\delta\in\Delta(\delta(n))}\sum_{n'\in\mathcal{N}^{\delta}}\sum_{c\in\mathcal{C}}p_{n',c}\eta^u_{n',c}\lambda_{n'}/\mu_{n'}$,
$\Upsilon(\bm x) = \frac{1}{2}\sum_{\delta\in\Delta}\sum_{n'\in\mathcal{N}^{\delta}}\sum_{c\in\mathcal{C}}p_{n',c}\eta^u_{n',c}\lambda_{n'}\nu_{n'}$, and
$\Psi_n(\bm x)=\sum_{c\in\mathcal{C}}p_{n,c}\eta^u_{n,c}\lambda_{n}$.
% Note that while the PoMiP's problem formulation is nearly identical to PoSiP's, it differs slightly, and the multi-class priority mechanism modifies the objective function relatively slightly.
 To obtain the effective offloading decision in this case, we next transform the original problem and 
decouple users' offloading decisions.

% Due to the coupling of multi-class priority user offloading decisions, some mathematical transformations of the objective function are necessary.
% Such M/G/1 with priority problem has been shown to be an NP-Hard problem\cite{moghadas2011queueing}.
% Due to the coupling of multi-class priority user offloading decisions, some mathematical transformations of the objective function are necessary in order to generate efficient solutions.
% Obviously, $\frac{1}{\mu_n}$ is not a decision variable and has no impact on the solution of $\bm x = \{\hat t^\mathrm{tr},\bm\eta^u\}$ in  \textbf{P2}.
First, we define $\theta_n \doteq f^{p,u}_n(\bm x_n)/f^{p,l}_n(\bm x_n) =  f^{p}_n(\bm x_n) - 1/\mu_n$ and write
\begin{align} 
    \mathop{\min}_{\boldsymbol{\theta}}\sum_{n\in\mathcal{N}}\theta_n
    = \mathop{\min}_{\bm x}\sum_{n\in\mathcal{N}}\frac{f^{p,u}_n(\bm x_n)}{f^{p,l}_n(\bm x_n)}
    % \nonumber\\
   = \sum_{n\in\mathcal{N}}\frac{f^{p,u}_n(\bm x_n^*)}{f^{p,l}_n(\bm x_n^*)},
\label{eq:NFP-theta}
\end{align}
where $\bm x_n^*$ is the optimal solution and $f^{p,u}_n,f^{p,l}_n$ are denoted by
% $f^{p,u}_n$ and $f^{p,l}_n$, which are \eqref{eq:NFP-upperfunction} and \eqref{eq:NFP-lowerfunction}.
\begin{align}
\label{eq:NFP-upperfunction}
    f^{p,u}_n(\bm x_n)   =\; & \Phi_{\delta(n)}(\bm x_n)\Phi_{\delta(n)-1}(\bm x_n) + \Upsilon(\bm x_n)\Psi_n(\bm x_n)\nonumber\\
    &+\hat t^\mathrm{tr}\Psi_n(\bm x_n)\Phi_{\delta(n)}(\bm x_n)\Phi_{\delta(n)-1}(\bm x_n),\\
\label{eq:NFP-lowerfunction}
    f^{p,l}_n(\bm x_n)   =\; & \Psi_n(\bm x_n)\Phi_{\delta(n)}(\bm x_n)\Phi_{\delta(n)-1}(\bm x_n).
\end{align}

% Define $\theta_n = \frac{f^{p,u}_n(\bm x_n)}{f^{p,l}_n(\bm x_n)}$ and $\boldsymbol{\theta} = \{\theta_n| n\in \mathcal{N}\}$. 
% Furthermore, $\Theta$ can be expressed as

\begin{proposition}
\label{pro:npl-mp}
Problem \textbf{P2} can be recast into an equivalent problem as follows:
    \begin{equation}
    \begin{split}
    \label{eq:P2-1}
        \textbf{\textrm{(P2-1)}}~&\mathop{\min}_{\{\bm x_n\}} \sum_{n\in\mathcal{N}} f^{p,u}_n(\bm x_n) - \theta^{*}_n f^{p,l}_n(\bm x_n), \\
        \text{s.t.}&~~\eqref{eq:c1},\eqref{eq:c2},\eqref{eq:c6},\eqref{eq:c7b},\eqref{eq:c8b},
    \end{split} 
\end{equation}
where $\theta^{*}_n$ is the minimum value of $f^{p,u}_n(\bm x_n)/f^{p,l}_n(\bm x_n)$ with regard to $\bm x_n$.
\end{proposition}
\begin{proof}
\tr{The detailed proof can be found in Appendix \ref{app:pro:npl-mp} of technical report~\cite{qiao2023popec}.}\fp{The detailed proof can be found in Appendix E of technical report~\cite{qiao2023popec}.}
\end{proof}

Proposition \ref{pro:npl-mp} demonstrates that problem transformation can be employed to deal with the nonlinear fractional programming problem \textbf{P2-1}. Hence, we can use the iterative Dinkelbach techniques~\cite{dinkelbach1967nonlinear}, outlined in Algorithm \ref{alg:NFP}, to solve the transformed problem.
% we utilized the iterative Dinkelbach technique to solve the nonlinear fractional programming problem, \ie \textbf{P2-1}, which is shown by Algorithm \ref{alg:NFP}.
Algorithm \ref{alg:NFP} updates the value of $\tilde{\theta}_n$ in each iteration based on the current $\tilde{\bm x}^k_n$, \ie,
\begin{equation}
    \tilde{\theta}^{k+1}_n = \frac{f^{p,u}_n(\tilde{\bm x}^{k+1}_n)}{f^{p,l}_n(\tilde{\bm x}^{k+1}_n))}.
\end{equation}
This process continues until:
% every iteration $k$ of Algorithm \ref{alg:NFP} changes the new value of $\tilde{\theta}^{k+1}_n = \frac{f^{p,u}_n(\tilde{\bm x}^{k+1}_n)}{f^{p,l}_n(\tilde{\bm x}^{k+1}_n))}$  until:
\begin{equation}
     f^{p,u}_n(\tilde{\bm x}^{k+1}_n) - \tilde{\theta}^{k}_n f^{p,l}_n(\tilde{\bm x}^{k+1}_n)) > \epsilon^{\mathrm{ck}},
\end{equation}
where $\epsilon^{\mathrm{ck}}$ represents the stop criteria for iterations.

\begin{algorithm}[t]
	\caption{Nonlinear fractional programming based on ADMM-Consensus (NFPA)}
	\label{alg:NFP}
	\LinesNumbered
	\tcp{CheckPointing Algorithm - NFP}
	\KwIn{$\epsilon^{\mathrm{ck}}$, $\{\tilde{\theta}^0_n\}$, $k$}
	\For{$n = 1$ \KwTo $N$}{
	{Convergence = False}\\
    \While{Convergence = False}{
        {Apply Algorithm \ref{alg:nonconvexadmm} or \ref{alg:AsynchronousADMM} with $\tilde{\bm x}^k_n$ and $\tilde{\theta}^k_n$ to obtain the solution $\tilde{\bm x}^{k+1}_n$}\\
        \eIf{$f^{p,u}_n(\tilde{\bm x}^{k+1}_n) - \tilde{\theta}^{k}_n f^{p,l}_n(\tilde{\bm x}^{k+1}_n)) > \epsilon^{\mathrm{ck}}$}{
        {Select $\tilde{\theta}^{k+1}_n = \frac{f^{p,u}_n(\tilde{\bm x}^{k+1}_n)}{f^{p,l}_n(\tilde{\bm x}^{k+1}_n))}$}\\
        {Convergence = True}}
	    {{Convergence = False}\\}
	    {Update $s = s+1$}\\
	}
	{Select $\bm x^{p,*}_n = \tilde{\bm x}^{k}_n$}\\
	{Select $\theta^{p,*}_n = \tilde{\theta}^{k}_n$}\\
	{Calculate \textbf{PAoI} according to Eq.~\eqref{eq:P2}}\\
	\KwOut{$\bm x^{p,*}_n$ and \textbf{PAoI}}
}	
\end{algorithm}

\subsubsection{ADMM-Consensus Based Solution}
\label{subsubsec:nonconvex-admm}
The procedure mentioned above can be considered as a checkpointing algorithm. We employ non-convex ADMM-Consensus methods to determine the new value of $\{\bm x_n\}$. Analogous to section \ref{subsec:POSIP}, the consensus problem for \textbf{P2-1} can be expressed as:
\begin{equation}
\label{eq:P2-2}
    \begin{split}
    \textbf{(P2-2)}~\mathop{\min}_{\{\bm x_n\}}~~&\sum_{n\in\mathcal{N}} g^{p}_n(\bm x_n), \\
    \text{s.t.}&~\bm x_n = \bm x_{o},
    \end{split}
\end{equation}
where
\begin{align}
\label{eq:feasible_solution2}
    g^{p}_n(\bm x_n) = \left\{\begin{array}{ll}
    f^{p,u}_n(\bm x_n) - \theta^{*}_nf^{p,l}_n(\bm x_n), &\bm x\in \Omega, \\
    \infty, &\text{otherwise.}
    \end{array}\right.
\end{align}
The augmented Lagrangian for Problem \textbf{P2-2} can be expressed as:
\begin{align}
\label{eq:PLadmm-lagrange}
    L^{p}(\{\bm x_n\}, \bm x_{o}, \{\bm\sigma_n\})
   = \sum_{n\in\mathcal{N}}L^{p}_{n}(\bm x_n, \bm x_{o}, \bm\sigma_n)\nonumber\\
   = \sum_{n\in\mathcal{N}}\left(g^{p}_n(\bm x_n) + \langle\bm\sigma_n, \bm x_n - \bm x_{o} \rangle
    + \frac{\rho_n}{2}\|\bm x_n - \bm x_{o}\|^2_2\right),
    % \\ \text{s.t.}  \mathcal{C} = \{(\bm x_1, \ldots, \bm x_N)|\bm x_1 = \ldots = \bm x_N\},
\end{align}
where $\rho_n$ is a positive penalty parameter with respect to Problem \textbf{P2-2}.
Based on the non-convex ADMM-Consensus algorithm, we update the variables in each iteration $t$ as follows:
% Therefore, we can obtain the efficient solution $\{\bm x_n\}$ to  Algorithm \ref{alg:NFP} via internal circulation ADMM-Consensus method (\ie Algorithm \ref{alg:nonconvexadmm}).
\begin{align}
    \label{eq:PLadmm2-relax-i} 
    \bm x^{t+1}_n & = \arg\min_{\bm x_n}L^p_n(\{\bm x_n\}, \bm x^t_{o}, \{\bm\sigma^t_n\}),
    \\
    % \{g^{p}_n(\bm x^t_n)\nonumber\\
    % & + \langle\bm\sigma^t_n, \bm x^t_n - \bm x^t_{o}\rangle + \frac{\rho_n}{2}\|\bm x^t_n - \bm x^t_{o}\|^2_2\} \\ 
    \label{eq:PLadmm2-relax-o}
    \bm x^{t+1}_{o} & = \frac{\sum_{n\in\mathcal{N}}(\rho_n \bm x^{t+1}_n + \sigma^{t}_n)}{\sum_{n\in\mathcal{N}}\rho_n}, \\
    \label{eq:PLadmm2-relax-sigma}
    \sigma^{t+1}_n & = \sigma^{t}_n + \rho_n(\bm x^{t+1}_n - \bm x^{t+1}_{o}).
 \end{align}
The premise that the iterative update can converge is that the function $g^{p}_n$ is Lipschitz continuous.

\begin{proposition}
\label{theorem:NFP_Problem}
The first-order derivative of $g^{p}_n$ is Lipschitz continuous with constant $\ell_n$, which is defined as
\begin{align}
    \ell_n=\frac{\lambda^2_{max}}{\mu^2_\mathrm{min}}\Big(
	        \sum_{n_1\in \mathcal{N}} \frac{\nu_{n}}{2\mu_{n}\mu_{n_1}}
	        + (1+\theta^{*}_n \frac{\nu_{n}}{2\mu_{n}})|\mathcal{N}_1|
	        + \theta^{*}_n \frac{\nu_{n}}{2\mu_{n}}|\mathcal{N}_2|\Big),
\end{align}
where $\mu_\mathrm{min} = \mathop{\min}_{n} \{\mu_n\}$,
$\lambda_{max} = \mathop{\max}_n \lambda_{n}$,
$\mathcal{N}_1 = \{\mathcal{N}^{\delta}|\delta<\delta(n)\}\bigcup \{n\}$, $\mathcal{N}_2 = \{\mathcal{N}^{\delta}|\delta\leq\delta(n)\}$.
% and $\{\leq\delta(n)\} = \{\delta|\delta\leq\delta(n)\}$.

% Based on this, we can apply the non-convex ADMM-Consensus method to acquire solution $\{\bm x_n\}$ of \textbf{P2-2}.
\end{proposition}
\begin{proof}
    % For more detailed information, please refer to
    \tr{The detailed proof can be found in Appendix \ref{app:theorem:NFP_Problem} of technical report~\cite{qiao2023popec}.}\fp{The detailed proof can be found in Appendix F of technical report~\cite{qiao2023popec}.}
\end{proof}

According to Proposition \ref{theorem:NFP_Problem}, we adopt the non-convex ADMM-Consensus method to solve Problem \textbf{P2-2}, with given $\left(\{\bm x^0_n\}, \{\theta^*_n\}\right)$ (as outlined in Algorithm \ref{alg:NFP}).

\begin{algorithm}[ht]
	\caption{Non-Convex ADMM-Consensus (NAC)}
	\label{alg:nonconvexadmm}
	\LinesNumbered
	\KwIn{$\epsilon^{\mathrm{ac}}$, $t$ and $\left(\{\bm x^0_n\}, \{\theta^*_n\}\right)$ from Algorithm \ref{alg:NFP}}
    \tcp{ADMM-Consensus Algorithm}
    \Do{$\eta\left(\bm x^t, \bm\sigma^t\right) \ge \epsilon^{\mathrm{ac}}$}{
        {Calculate $\bm x^{t+1}_{o}$ in \textbf{MEC}, according to Eq.~\eqref{eq:PLadmm2-relax-o}}\\
    	{Calculate $\bm x^{t+1}_n$ in \textbf{user}, simultaneously according to Eq.~\eqref{eq:PLadmm2-relax-i}}\\
        {Calculate $\sigma^{t+1}_n$ in \textbf{user}, simultaneously according to Eq.~\eqref{eq:PLadmm2-relax-sigma}}\\
        {Update $t = t + 1$}\\
	    }
	\KwOut{$\bm x^{t}_n$}
\end{algorithm}
It is worth mentioning that the value of $\rho_n$ and the convergence property of the Algorithm \ref{alg:nonconvexadmm} differ from those in the Convex ADMM-Consensus method. Due to the non-convexity of the objective function, the commonly used \textit{gap function} cannot be adapted to the analysis of Algorithm \ref{alg:nonconvexadmm}. Next, we design a special gap function capable of characterizing the convergence of NAC as follows:
\begin{equation}
\label{eq:gap-function}
    \eta\left(\bm x^t, \bm\sigma^t\right)=\left\|\tilde{\nabla} L^{p}\left(\left\{\bm x^t_{n}\right\}, \bm x^t_{o}, \bm\sigma^t\right)\right\|^{2}+\sum_{n\in\mathcal{N}}\left\|\bm x^t_{n}-\bm x^t_{o}\right\|^{2},
\end{equation}
where 
\begin{align*}
    \tilde{\nabla} L^{p}\left(\left\{\bm x_{n}\right\}, \bm x_{o}, \bm\sigma^t\right)
   =\left[\begin{array}{c}
    \nabla_{\bm x_{o}} L^{p}\left(\left\{\bm x_{n}\right\}, \bm x_{o}, \bm\sigma^t\right)  \\
    \nabla_{\bm x_{1}} L^{p}\left(\left\{\bm x_{n}\right\}, \bm x_{o}, \bm\sigma^t\right) \\
    \vdots \\
    \nabla_{\bm x_{N}} L^{p}\left(\left\{\bm x_{n}\right\}, \bm x_{o}, \bm\sigma^t\right)
    \end{array}\right].
\end{align*}
% When $\lim_{t \to \infty} \eta\left(\bm x^t, \bm\sigma^t\right)=0$, we have the limit solution is a stable solution.
When $\eta\left(\bm x^t, \bm\sigma^t\right) < \epsilon^{\mathrm{ac}}$, Algorithm \ref{alg:nonconvexadmm} will find a stationary solution. 
The gap function that ADMM used to deal with convex functions is no longer suitable to judge the convergence of non-convex functions ADMM.
We further explain why the gap function Eq. \eqref{eq:gap-function} can be used as a stopping criterion
\tr{(A comprehensive explanation of this concept can be found in Appendix \ref{app:GapFunction-Convergence} of technical report~\cite{qiao2023popec}.)}\fp{(A comprehensive explanation of this concept can be found in Appendix L of technical report~\cite{qiao2023popec}.)}
In addition to the change of the gap function, the convergence of Algorithm 2 requires some parameters to satisfy special conditions.
% Obviously, above algorithm is too ideal, since the network system is in a unreliable channel and every single communication has the potential to fail. To accommodate this complexity, we propose an algorithm that does not have to iterate all the nodes together at each iteration, at the cost of possibly losing accuracy.

\begin{theorem}
	\label{theorem:innerConvergence1}
    If $\rho_n>2\ell_n$, Algorithm \ref{alg:nonconvexadmm} converges to an $\epsilon^{\mathrm{ac}}$-stationary point within $O(1/(p^\mathrm{syn}\epsilon^\mathrm{ac}))$,
    where $\epsilon^{\mathrm{ac}}$ is a positive iteration factor, and $p^\mathrm{syn} = \Pi_{n\in\mathcal{N}}(\frac{1}{C} \sum_{c\in\mathcal{C}}p_{n,c})$ is the probability of successfully completing a synchronous update. 
    % \begin{equation}
    %     \Gamma^\mathrm{syn}<\frac{k^{\Gamma}(L^{p}\left(\left\{\bm x_{n}^{1}\right\}, \bm x_{o}^{1}, \bm \sigma^{1}\right)-\underline{G^{p}})}{\epsilon^{\mathrm{ac}} p^\mathrm{syn}},
    % \end{equation}
    % where  
    % $\epsilon^{\mathrm{ac}}$ is a positive iteration factor, 
    % $k^{\Gamma}$ is a constant, 
    % $p^\mathrm{syn} = \Pi_{n\in\mathcal{N}}(\frac{1}{C} \sum_{c\in\mathcal{C}}p_{n,c})$ probability of successfully completing a synchronous update,
    % $\underline{G^{p}}$ is the lower bound of $\sum_{n\in\mathcal{N}} g^{p}_n(\bm x_n)$,
    % % \ie $\underline{G^{p}} = \inf_{\bm x_n} \sum_{n\in\mathcal{N}} g^{p}_n(\bm x_n)$,
    % and $\Gamma^\mathrm{syn}$ is number of iterations, \ie 
    % $\Gamma^\mathrm{syn} = \min \left\{t\mid \eta\left(\bm x^t, \bm\sigma^t\right) \leq \epsilon, t \geq 0\right\}$.
\end{theorem}
\begin{proof}
    \tr{The detailed proof can be found in Appendix \ref{app:theorem:innerConvergence1} of technical report~\cite{qiao2023popec}.}\fp{The detailed proof can be found in Appendix G of technical report~\cite{qiao2023popec}.}
\end{proof}
{
Theorem \ref{theorem:innerConvergence1} implies that when specific parameters satisfy certain criteria, Algorithm \ref{alg:nonconvexadmm} can sublinearly converge.
Additionally, since the iteration will continue if the parameters given by Algorithm \ref{alg:nonconvexadmm} do not satisfy the stop condition of Algorithm \ref{alg:NFP}, we can use Algorithm \ref{alg:NFP} as a condition for the termination
of the entire solution method.}

\subsection{Multi-Priority Task Scheduling and Multi-Server Collaboration\label{subsec:POMIS}}
\subsubsection{Problem Decomposition}
Different from Section \ref{subsec:POSIP} and Section \ref{subsec:POMIP}, which concentrate solely on a single server, we further study multi-priority and multi-server collaboration-based offloading in this subsection. 
In order to resolve the original problem,  we first derive the expression of the optimal migration decision variable $\bm y^*$.
\begin{theorem}
\label{thm:single-or-multi}
    The optimal migration decision of (\textbf{PoPeC}) is
    \begin{equation}  
	\label{eq:opy-y}
		y_{m_n} = \left\{
		\begin{array}{ll}
		 1, &\quad  \phi^{in}(\bm \eta^s)+\phi^{out}(\bm \eta^s)> 0,\\
		 0 , &\quad otherwise,
		\end{array}
		\right.	
	\end{equation}
	where $\phi^{in}(\bm \eta^s) = \sum_{m'\in\mathcal{M}/m_n} \eta^s_{m',m_n} \sum_{\delta \in \Delta} \lambda^{s}_{\delta,m'}$ and $\phi^{out}(\bm \eta^s)$ $= \sum_{m'\in\mathcal{M}/m_n} \eta^s_{m_n,m'} \sum_{\delta \in \Delta} \lambda^{s}_{\delta,m_n}$.
\end{theorem}
\begin{proof}
    \tr{The detailed proof can be found in Appendix \ref{app:thm:single-or-multi-A} of technical report~\cite{qiao2023popec}.}\fp{The detailed proof can be found in Appendix H-A of technical report~\cite{qiao2023popec}.}
\end{proof}
    % Combining Theorem \ref{thm:single-or-multi} and discussion of constraints, there is a comparable solution
    % $\hat{\bm y^*} = \bm 1$.
    % which holds $\sum_{n\in\mathcal{N}} F_n(\bm x, \bm y^{*}, \bm z)=\sum_{n\in\mathcal{N}} F_n(\bm x, \hat{\bm y^{*}}, \bm z)$.
    % Hence, we have 
    % \begin{equation}
    % \label{eq:fn=f1n}
    %     \sum_{n\in\mathcal{N}} F_n(\bm x, \hat{\bm y^{*}}, \bm z) \leq \sum_{n\in\mathcal{N}} F_n(\bm x, \bm y, \bm z),
    % \end{equation}

\tr{Combining Theorem \ref{thm:single-or-multi} and Appendix \ref{app:thm:single-or-multi-B} of technical report~\cite{qiao2023popec},}\fp{Combining Theorem \ref{thm:single-or-multi} and Appendix H-B of technical report~\cite{qiao2023popec},}
\textbf{(PoPeC)} is equivalent to the following problem:
\begin{equation}
\label{eq:P3}
\begin{split}
    \textbf{(P3)}~&\mathop{\min}_{\bm x, \bm z}
    \frac{1}{N}\sum_{n\in\mathcal{N}} F^1_n(\bm x, \bm z) \\
    \text{s.t.}&~ 
    \eqref{eq:c1},\eqref{eq:c2},%to
    \eqref{eq:c3a},\eqref{eq:c5a},\eqref{eq:c4a},%tm
    \eqref{eq:c6},%transmission
    % \eqref{eq:c-y}, %y
    \eqref{eq:c7b},\eqref{eq:c8b}, %capacity
\end{split}
\end{equation}
where $F^1_n(\bm x, \bm z) = \hat t^\mathrm{tr}_n + \frac{1}{\sum_{c\in\mathcal{C}}p_{n,c}\eta^u_{n,c}\lambda_{n}} + \sum_{m\in\mathcal{M}} \pi_{n,m}(\bm z)$,
$\bm x = \{\hat t^\mathrm{tr}, \bm\eta^u\}$ and $\bm z = \bm\eta^s$.
The main challenge in solving \textbf{P3} lies in the non-convexity of  $\sum_{n\in\mathcal{N}} F^1_n(\bm x, \bm z)$ and the dependency of $\bm x$ and $\bm z$.
We further decompose the problem \textbf{P3}.

\begin{lemma}
    \label{lemma:Transform-DataMigration-ChannelAllocation}
Problem \textbf{P3} can be equivalently transformed into the \textbf{Channel Allocation} subproblem,  \textbf{P3-1}, and the \textbf{Server Collaboration} subproblem, \textbf{P3-2}, which are shown as follow:
\begin{align}
    \textbf{(P3-1)}~&\mathop{\min}_{\bm x}\frac{1}{N}\sum_{n\in\mathcal{N}} F^2_n(\bm x), \nonumber\\
    \text{s.t.}&~ \eqref{eq:c1},\eqref{eq:c2},\eqref{eq:c6},\eqref{eq:c7b},\eqref{eq:c8b},\nonumber\\
    \label{eq:DataMigration-AuxiliaryInequality}
    &\sum_{n \in \mathcal{N}^{\delta}_m}\sum_{c \in C} p_{n,c} \eta^u_{n,c}\lambda_{n} \leq \lambda^{s}_{\delta,m}.
\end{align}
\begin{equation}
\begin{split}
    \textbf{(P3-2)}~&\mathop{\min}_{\bm z}\frac{1}{N}\sum_{n\in\mathcal{N}} F^3_n(\bm z, \bm \lambda^s), \\
    \text{s.t.}&~ \eqref{eq:c3a},\eqref{eq:c5a},\eqref{eq:c4a},
\end{split}
\label{eq:P3-2}
\end{equation}
where we define $F^2_n(\bm x) \doteq \hat t^\mathrm{tr}_n + \frac{1}{\sum_{c\in\mathcal{C}}p_{n,c}\eta^u_{n,c}\lambda_{n}}$,  $F^3_n(\bm z, \bm \lambda^s) \doteq \sum_{m\in\mathcal{M}} \pi_{n,m}(\bm z)$,
and $\lambda^{s}_{\delta,m} \doteq  \sum_{m' \in \mathcal{M}}  \eta^s_{m,m'} \lambda^{s}_{\delta,m}$.
\end{lemma}
\begin{proof}
    \tr{The detailed proof can be found in Appendix \ref{app:lemma:Transform-DataMigration-ChannelAllocation} of technical report~\cite{qiao2023popec}.}\fp{The detailed proof can be found in Appendix I of technical report~\cite{qiao2023popec}.}
\end{proof}

{
Based on the lemma provided, we develop methods to address the channel allocation \textbf{P3-1} and server cooperation \textbf{P3-2} iteratively on the user and server sides, respectively.}

\subsubsection{Channel Allocation} 
The goal of \textbf{P3-1} is to allocate channels for each local server and the user it serves. Based on Lemma \ref{lemma:Transform-DataMigration-ChannelAllocation}, we can easily obtain the optimal solution of this sub-problem as follows:
     \begin{equation}
        \label{eq:opt-x}
        \begin{split}
            \bm x^{*}=\{\bm x^{*}_m\},
        \end{split}
     \end{equation}
where $\bm x^{*}_m$ can be derived from 
\begin{align}
    \textbf{(P3-3)} \mathop{\min}_{\bm x_m}  \sum_{n\in\mathcal{N}_m} F^2_n(\bm x_m, \bm \lambda^s),~\nonumber\\
    \text{s.t.} \eqref{eq:c1},\eqref{eq:c2}, \eqref{eq:c6},\eqref{eq:c7b},\eqref{eq:c8b}, \eqref{eq:DataMigration-AuxiliaryInequality}.
\end{align}
We further analyze \textbf{P3-3} to identify an optimal solution.

\begin{lemma}
\label{lemma:DataMigration-ChannelAllocation}
% \begin{align}
% \label{eq:P3-4}
%     \textbf{(P3-4)}~&\mathop{\min}_{\bm x_m}\sum_{n\in\mathcal{N}_m} F^2_n(\bm x_m) \nonumber\\
%     \text{s.t.}&~ \eqref{eq:c1},\eqref{eq:c2},\eqref{eq:c6},\eqref{eq:c7b},\eqref{eq:c8b},
%     \eqref{eq:DataMigration-AuxiliaryInequality},
% \end{align}
    Problem \textbf{P3-3} is a convex problem. 
    \end{lemma}
    \begin{proof}
        \tr{We can easily obtain the convexity of $F^1_n$ from Appendix \ref{app:lemma:paoi_mg1} of technical report~\cite{qiao2023popec}.}\fp{We can easily obtain the convexity of $F^1_n$ from Appendix C of technical report~\cite{qiao2023popec}.} Since both the sub-function and the constraints are convex, we yield the result. 
    \end{proof}
    Based on Lemma \ref{lemma:DataMigration-ChannelAllocation} and the solution of \ref{subsec:POSIP}, \textbf{P3-3} can be resolved by the existing method AC in Algorithm \tr{\ref{alg:admm}}\fp{5}.
    % Specifically, we can leverage the ADMM-Consnesus method to gain more insights into the structure of the solution.
    \tr{The details can be found in Section \ref{subsubsec:POSIP2} and Appendix \ref{app:det:admm-consensus} of technical report~\cite{qiao2023popec}.}\fp{The details can be found in Section D and Appendix D of technical report~\cite{qiao2023popec}.}
    Each local server can find an optimal solution $\bm x^{*}_m$ for all users it covers, which is in fact the optimal channel allocation for \textbf{P3-1}.
    
\subsubsection{Server Collaboration}
    Problem \textbf{P3-2} seeks to address the problem of multi-server collaboration between various servers in order to lessen the load on overhead servers and speed up task execution to decrease PAoI for multi-priority users, which has been shown to be an NP-Hard problem in section \ref{subsec:POMIP}.
    However, we develop an effective migration strategy,
    % which necessitates the waiting and execution of priority user task requests on various servers in each region as well as considering server-to-server communication.
    based on Lemma \ref{thm:single-or-multi}, the server collaboration strategy is
    $\bm z^* = \mathop{\arg \min}_{\bm z} \{\frac{1}{N} \sum_{n\in\mathcal{N}} \sum_{m\in\mathcal{M}} \pi_{n,m}(\bm z) ~ \text{s.t.}~\eqref{eq:c3a},\eqref{eq:c5a},\eqref{eq:c4a} \}$.

{\begin{proposition}
    \label{proposition:DataMigration-MigrationAllocation1}
    Given $\bm \lambda^s$, Problem \textbf{P3-2}  is equivalent to the following problem:
    \begin{align}
        \textbf{(P3-4)}~&\mathop{\min}_{\{\bm z_{n,m}\}} \sum_{n\in\mathcal{N}} \sum_{m\in\mathcal{M}} \{\pi^{u}_{n,m}(\bm z_{n,m}) - \vartheta^{*}_{n,m} \pi^{l}_{n,m}(\bm z_{n,m})\},\nonumber \\
        \text{s.t.}&~~\eqref{eq:c1},\eqref{eq:c2},
        \eqref{eq:c7a},
        \eqref{eq:c8a},
        \eqref{eq:c6}.
        \label{eq:P3-4}
    \end{align}
Given $\bm z_{n,m}$, $\vartheta^{*}_{n,m}$ is the minimum of $\frac{\pi^{u}_{n,m}(\bm z_{n,m})}{\pi^{l}_{n,m}(\bm z_{n,m})}$, 
where
% $\pi^{l}_{n,m}(\bm z)$ $=\Phi^{\pi}_{\delta(n),m}(\bm z)\Phi^{\pi}_{\delta(n)-1,m}(\bm z)\sum_{m'\in\mathcal{M}} \Psi_{n, m'}^{\pi}(\bm z)$ and
% $\pi^{u}_{n,m}(\bm z)=\Lambda_{n,m}\Psi_{n,m}^{\pi}(\bm z)\Phi^{\pi}_{\delta(n),m}(\bm z)\Phi^{\pi}_{\delta(n)-1,m}(\bm z)+\Psi_{n,m}^{\pi}(\bm z)\Upsilon^{\pi}(\bm z)$. Here, $\Upsilon^{\pi}_{n,m}(\bm z) = \frac{1}{2}\sum_{\delta\in\Delta} \sum_{m'\in\mathcal{M}} \eta^s_{m',m} \lambda^{s}_{\delta,m'} \nu_{n,m}$,
% % Additionally, it should be noted that
% $\Lambda_{n,m} = t^\mathrm{tr}_{m_n,m} + \frac{1}{\mu_{n,m}}$,
% $\Phi^{\pi}_{\delta(n),m}(\bm z) = 1-\sum_{\delta\in{\Delta(\delta(n))}}\sum_{m'\in\mathcal{M}} \eta^s_{m',m} \lambda^{s}_{\delta,m'} \frac{1}{\mu_{n,m}}$,
% and
% $\Psi_{n,m}^{\pi}(\bm z)=\sum_{\delta\in\Delta} \eta^s_{m_n,m} \lambda^{s}_{\delta,m_n}$.
\begin{align}
\left\{
    \begin{array}{l}
        \pi^{l}_{n,m}(\bm z)=\Phi^{\pi}_{\delta(n),m}(\bm z)\Phi^{\pi}_{\delta(n)-1,m}(\bm z)\sum_{m'\in\mathcal{M}} \Psi_{n, m'}^{\pi}(\bm z),\\
        \pi^{u}_{n,m}(\bm z)=\Lambda_{n,m}\Psi_{n,m}^{\pi}(\bm z)\Phi^{\pi}_{\delta(n),m}(\bm z)\Phi^{\pi}_{\delta(n)-1,m}(\bm z)\nonumber\\
        ~~~~~~~~~~+\Psi_{n,m}^{\pi}(\bm z)\Upsilon^{\pi}(\bm z),\\
        \Phi^{\pi}_{\delta(n),m}(\bm z) = 1-\sum_{\delta\in{\Delta(\delta(n))}}\sum_{m'\in\mathcal{M}} \eta^s_{m',m} \lambda^{s}_{\delta,m'} \frac{1}{\mu_{n,m}},\\
        \Upsilon^{\pi}_{n,m}(\bm z) = \frac{1}{2}\sum_{\delta\in\Delta} \sum_{m'\in\mathcal{M}} \eta^s_{m',m} \lambda^{s}_{\delta,m'} \nu_{n,m},\\
        \Psi_{n,m}^{\pi}(\bm z)=\sum_{\delta\in\Delta} \eta^s_{m_n,m} \lambda^{s}_{\delta,m_n},\\
        \Lambda_{n,m} = t^\mathrm{tr}_{m_n,m} + \frac{1}{\mu_{n,m}}.
    \end{array}
\right.
\end{align}
% As for \textbf{P3-4}, we can utilize the ADMM-Consensus method to derive an efficient solution $\{\bm z_{n,m}\}$.
\end{proposition}}
\begin{proof}
{${\pi^{u}_{n,m}(\bm z)}$ and ${\pi^{l}_{n,m}(\bm z)}$ are polynomial functions representing the numerator and denominator of the fraction $\pi_{n,m}(\bm z)$, respectively. Accordingly, we can reframe problem \textbf{P3-2} using nonlinear fractional programming, resulting in an equivalent problem denoted as \textbf{P3-4} with given $\bm \lambda^s$.
\tr{For more detailed problem transformation, please refer to Appendix \ref{app:proposition:DataMigration-MigrationAllocation1} of technical report~\cite{qiao2023popec}.}\fp{For more detailed problem transformation, please refer to Appendix J-1 of technical report~\cite{qiao2023popec}.}}
\end{proof}
Combining Proposition \ref{proposition:DataMigration-MigrationAllocation1} and Section \ref{subsubsec:NFP}, we can apply NFP in Algorithm \ref{alg:NFP} to convert \textbf{P3-2} to \textbf{P3-4}.
After transforming the problem in Proposition \ref{proposition:DataMigration-MigrationAllocation1}, we obtain problem \textbf{P3-4}, which has a cubic polynomial objective function. Nevertheless, deriving the closed-form solution for \textbf{P3-4} is still challenging, we instead provide an iterative algorithm as follows.
We first analyze the properties of \textbf{P3-4}.

\begin{lemma}
\label{lemma:DataMigration-MigrationAllocation2}
    In \textbf{P3-4}, the first-order derivative of $\pi^{u}_{n,m}(\bm z_{n,m}) - \vartheta^{*}_{n,m} \pi^{l}_{n,m}(\bm z_{n,m})$ is Lipschitz continuous.
\end{lemma}
\begin{proof}
    \tr{The detailed proof can be found in Appendix \ref{app:lemma:DataMigration-MigrationAllocation2} of technical report~\cite{qiao2023popec}.}\fp{The detailed proof can be found in Appendix J-2 of technical report~\cite{qiao2023popec}.}
\end{proof}
Based on Lemma \ref{lemma:DataMigration-MigrationAllocation2} and Section \ref{subsubsec:nonconvex-admm}, we use NAC and NFP in Algorithm \ref{alg:nonconvexadmm} to gain an efficient solution of \textbf{P3-4}. 
\tr{In this case, the complexity of the method is $O(1/\epsilon^{\mathrm{ac}})$, which can be proved by Appendix \ref{app:proof:DataMigration-MigrationAllocation3} of technical report~\cite{qiao2023popec}.}\fp{In this case, the complexity of the method is $O(1/\epsilon^{\mathrm{ac}})$, which can be proved by Appendix J-3 of technical report~\cite{qiao2023popec}.}

% Therefore, we can utilize the NFP and NAC method, shown by Section \ref{subsec:POMIP}, to gain an efficient solution of \textbf{P3-4}. 

% \begin{lemma}
%     \label{lemma:DataMigration-MigrationAllocation}
%          Given an fixed $\bm x$, we can solve the server collaboration strategy $\widetilde{\bm z}$ by Algorithm \ref{alg:NFP} and Algorithm \ref{alg:nonconvexadmm}, due to reliable communication and iteration between servers. Especially, if this has only a Priority-Free, we can obtain the optimal solution.
% \end{lemma}
% \begin{proof}
%     See Appendix \ref{app:lemma:DataMigration-MigrationAllocation}.
% \end{proof}

% \end{itemize}

% The function $F^2_n(\bm x)=\hat t^\mathrm{tr}_n + \frac{1}{\Psi_n(\bm x)} + \pi(\bm z_o)$ is a convex function, based on Lemma \ref{lemma:paoi_mg1}.

% \label{lemma:compare-convergence}    
%     Denote $\widetilde{s}_{m} = \sum_{n\in\mathcal{N}_m}\sum_{c\in\mathcal{C}} p_{n,c} \lambda_{n,c}$ as an auxiliary variable. We have $\widetilde{s}_{m} q^s_{m,m'} = \widetilde{s}_{m,m'}$. Denote $\widetilde{\bm z} = \{\widetilde{s}_{m,m'}\}$.  
%     We have $F_n(\bm x^*, \bm z_o) \geq F_n(\bm x^*, \widetilde{\bm z})$
% \end{lemma}

\subsubsection{Iterative Solution}
We design an iterative solution algorithm that first obtains the initial $\bm x_o$ inside each local server with a given $\bm y$. In the algorithm, $\bm x^t$ and $\bm z^t$ are solved alternately to obtain the solution. In the following, we establish the convergence guarantee for the proposed algorithm.

\begin{algorithm}[htbp]
	\caption{Iterative Solution (IS)}
	\label{alg:IS}
	\LinesNumbered
	\KwIn{$t=0$, $\lambda^{s,0}_{\delta,m} = \lambda^{s,\mathrm{max}}_m$}
    \While{not done}{
        {Compute $\bm x^{t+1}$ according to Eq.~\eqref{eq:opt-x}} by Lemma \ref{lemma:DataMigration-ChannelAllocation}\\
        {Compute $\bm \lambda^{s,t+1}$ according to Eq.~\eqref{eq:lambda-s}}\\
        {Compute $\bm z^{t+1}$ according to Eq.~\eqref{eq:P3-4}} by Lemma \ref{lemma:DataMigration-MigrationAllocation2}\\
        {Update $t = t + 1$}\\
	    }
	{Compute $\bm y^*$ by Eq.~\eqref{eq:opy-y}}\\
	\KwOut{$\bm x^*=\bm x^{t}$,$\bm y^*$,$\bm z^*=\bm z^{t}$}
\end{algorithm}

\begin{theorem}
\label{thm:DataMigration-Convergence}
If we solve \textbf{P3} by Algorithm \ref{alg:IS}, $F_n(\bm x^t, \bm z^t)$
monotonically decreases and converges to a unique point.
\end{theorem}
\begin{proof}
    \tr{The detailed proof can be found in Appendix \ref{app:thm:DataMigration-Convergence} of technical report~\cite{qiao2023popec}.}\fp{The detailed proof can be found in Appendix K of technical report~\cite{qiao2023popec}.}
\end{proof}

\section{Discussion}\label{sec:Discussion}
{In this section, we expand upon and analyze the performance of PoPeC.  Firstly, we introduce a communication-efficient asynchronous parallel algorithm and investigate both its convergence and convergence rate.  Following that, we delve into the distinctions and benefits of our approach when compared to non-priority and traditional multi-priority methods.}
\subsection{Asynchronous Parallel Algorithm}
\label{subsec:Discussion-AsynchronousADMM}
\begin{algorithm}[htbp]
    \label{alg:AsynchronousADMM}
	\caption{Asynchronous Non-Convex ADMM-Consensus (ANAC)}
	\LinesNumbered   
    \BlankLine
	{/* MEC Side */}\\
	\KwIn{$\epsilon^{\mathrm{ac}}$, $\tau_{o}$ and $\left(\{\bm x^0_n\}, \{\theta^*_n\}\right)$ from Algorithm \ref{alg:NFP}}
	{Convergence = False}\\
	\While{Convergence == False}{
        {Calculate $\bm x^{\tau_{o}+1}_{o} = \arg\min_{\bm x_{o} \in \Omega} L^p\left(\{\bm x^0_n\}, \bm x_{o}, \{\bm\sigma^0_n\}\right)$}
        % = \arg\min_{\bm x_{o} \in \Omega} \frac{\sum^{\mathcal{N}}\rho_n}{2} \| \bm x_{o} - \frac{\sum^{\mathcal{N}}\rho_n \bm x^{\tau_{o}}_n + \sum^{\mathcal{N}}\bm\sigma^{\tau_{o}}_n}{\sum^{\mathcal{N}}\rho_n}\|^2$.}\\
        \eIf{$\eta\left(\bm x^{\tau_{o}}, \bm\sigma^{\tau_{o}}\right) < \epsilon^{\mathrm{ac}}$}{
        {Convergence = True}\\
        {Send $\bm x^{\tau_{o}+1}_{o}$ and Convergence to all \textbf{users}}}
        {{Send $\bm x^{\tau_{o}+1}_{o}$ and Convergence to all \textbf{users}}\\
        {Wait for some fixed period of time}\\
        {Receive all the gradients $\{x^{\tau^r_n}_{n}\}$ and all the local time $\{\bm\tau^r_n\}$ from \textbf{users}}\\
        {Record the received users in $\mho^{\tau_{o}+1}$}\\
	    \eIf{$n \in \mho^{\tau_{o}+1}$}{
        {Select $\tau_n = \min \{\bm\tau^r_n\}$ and $\nabla G^{p,\tau_{o}+1}_n = x^{\tau^r_n}_{n}$}}
	{{Select $\nabla G^{p,\tau_{o}+1}_n = \nabla G^{p,\tau_{o}}_n$}}
	    
        {Calculate $\bm x^{\tau_{o}+1}_n = \bm x^{\tau_{o}+1}_o - \frac{1}{\rho_n} (\nabla G^{p,\tau_{o}+1}_n + \bm\sigma^{\tau_{o}}_n)$}\\
        {Calculate $\bm\sigma^{\tau_{o}+1}_n = \bm\sigma^{\tau_{o}}_n + \rho_n (\bm x^{\tau_{o}+1}_n + \bm x^{\tau_{o}+1}_o)$}\\
        {Update $\tau_{o} = \tau_{o} + 1$}}
	    }
	\KwOut{$\bm x^{\tau_{o}}_{o}$}
	    
	{/* USER Side */}\\
	\For{user $n\in\mathcal{N}$}{
	    {Initialize $\tau_n = 0$}\\
    	\While{Receive $\bm x^r$ and Convergence from \textbf{MEC}}{
    	\eIf{Convergence == False}{
    	{Select $\bm x^{\tau_{n}+1}_n = \bm x^r$}\\
    	{Calculate $\nabla g^{p}_n(\bm x^{\tau_{n}+1}_n)$}\\
        % according to Eq.\eqref{eq:FunctionGpnFirstOrder}
    	{Send $\nabla g^{p}_n(\bm x^{\tau_{n}+1}_n)$ and $(\tau_{n}+1)$} to \textbf{MEC} via the most reliable channel\\
    	{Update $\tau_{n} = \tau_{n} + 1$}\\}
    	{{Select $\bm x_n = \bm x^r$}\\
    	\KwOut{$\bm x_n$}}
    	}
	}
	 
\end{algorithm}
{In the previous sections, we propose several synchronous parallel algorithms. Nevertheless, the reliability and effectiveness of these algorithms can be severely affected by communication failures and chaos resulting from faulty communication networks. 
As Theorem \ref{theorem:innerConvergence1} has revealed, Algorithm \ref{alg:nonconvexadmm} has a slow convergence speed, highlighting the need to replace it with a more communication-efficient alternative.}
% In the work mentioned above, we developed numerous synchronous parallel algorithms. However, issues like chaos and communication loss brought on by faulty communication networks can lead to iterations in parallel algorithms failing, which has a significant impact on the robustness and effectiveness of the algorithms.
% According to the Lemma \ref{theorem:innerConvergence1}, the convergence speed of Algorithm \ref{alg:nonconvexadmm} is slow, and it is urgent to replace it with a communication-efficient algorithm.

% Instead of the Algorithm \ref{alg:nonconvexadmm}, Algorithm \ref{alg:AsynchronousADMM} guarantees that the iteration can still be successfully completed even if the server only collects updates from a limited number of users, which reduces complexity caused by unreliable channels.

% First, the channel with the highest reliability rate is chosen to make sure that the iteration achieves the highest success rate (\ie $p^\mathrm{max}_n = \max_c\{p_{n,c}\}$), in new Algorithm.
{Firstly, users opt to prioritize the selection of the channel with the highest reliability rate to maximize the success rate during iterations, represented as $p^\mathrm{max}_n = \max_c\{p_{n,c}\}$.
In doing so, the risk of communication iteration loss is intuitively reduced.
To ensure successful iterations in an asynchronous algorithm, there needs to be a limit on the number of communications required for each communication unit.
In order to satisfy the delay bound of iterations needed for successful communication $\Gamma_n$, we have $(1-p^\mathrm{max}_n)^{\Gamma_n}<\epsilon^{a}$, where  $\epsilon^{a}$ is the maximum tolerance for asynchronous iterative communication.
Both sides are logarithmic at the same time, it is $\Gamma_n \geq  \frac{ln(\epsilon^{a})}{ln(1-p^\mathrm{max}_n)}$.
 Thus, we have}
% Thus, we derive $\Gamma_n = \lceil \frac{ln(\epsilon^{a})}{ln(1-p^\mathrm{max}_n)} \rceil$.
    \begin{assumption}
	\label{assumption:finiteCommunication}
	     The upper bound on the number of communications to complete a successful iteration satisfies
	    \begin{equation}
	        \Gamma_n = \left\lceil \frac{ln(\epsilon^{a})}{ln(1-p^{\mathrm{max}}_n)} \right\rceil,
	    \end{equation}
        where $\lceil x \rceil$  is the ceiling function of $x$.
    \end{assumption}
{This is a standard assumption in the asynchronous ADMM literature~\cite{hong2016convergence,hong2017distributed}. In the worst case, it is a necessary condition to ensure that each user finishes one iteration within $\Gamma_n$ iterations.}
    % \item \textbf{Users with Limited Computing Resources}. 
    % Considering the possible computational resource constraints of users, we believe in assigning as few computational tasks as possible or designing them to be as simple as possible.
    % Since computational resources are limited, we aim to design as few tasks as possible while still taking into account the limitations of the users.
% \end{itemize}

% In addition, we give the asynchronous Algorithm \ref{alg:AsynchronousADMM} according to Algorithm \ref{alg:nonconvexadmm}, where for each iteration, each user computes the gradient based on the most recent information received from the server and sends it to the local server. The server collects all available iterations, finds the new value of $\bm x$, etc., and passes the latest information to the user. 
% Furthermore, considering the possible computational resource limitation of the user, we advocate assigning as few computational tasks as possible or designing as simple as possible.
Furthermore, we present an asynchronous variant of Algorithm \ref{alg:nonconvexadmm} as Algorithm \ref{alg:AsynchronousADMM}.  In each iteration, each user computes the gradient based on the most recently received information from the server and sends it to the local server.  
The server collects all available iterations, updates the value of $\bm x$, and passes the latest information back to the user.  
This asynchronous approach can help reduce communication overhead and improve convergence speed.
In addition, we recognize that users may have limited computational resources, and thus, we suggest assigning a minimal number of computational tasks or designing the tasks to be as simple as possible.

\textbf{Convergence Analysis:}
If Assumption \ref{assumption:finiteCommunication} is satisfied and we set $\rho_n>\max\{7\ell_n, \ell_n(\Gamma^2_n + \frac{3}{7}(\Gamma_n + 1)^2)\}$, the sequence $\{\{\bm x_n\}, \bm x\}$ in Algorithm \ref{alg:AsynchronousADMM} converges to the set of stationary solutions of the problem, based on \cite[Theorem 3.1]{hong2017distributed}.
Moreover, for $\epsilon^{\mathrm{ac}}>0$, we obtain
    \begin{equation}
        p^{\mathrm{asyn}}\Gamma^{\mathrm{asyn}}\epsilon^{\mathrm{ac}} <k^{\Gamma}(L^{p}\left(\left\{\bm x_{n}^{1}\right\}, \bm x_{o}^{1}, \bm \sigma^{1}\right)-\underline{G}^{p}),
    \end{equation}
    where  
    $\epsilon^{\mathrm{ac}}$ is a positive iteration factor, 
    $k^{\Gamma}$ is a constant, 
    % $p^\mathrm{syn} = \Pi_{n\in\mathcal{N}}(\frac{1}{C} \sum_{c\in\mathcal{C}}p_{n,c})$ probability of successfully completing a synchronous update,
    $\underline{G^{p}}$ is the lower bound of $\sum_{n\in\mathcal{N}} g^{p}_n(\bm x_n)$,
    % \ie $\underline{G^{p}} = \inf_{\bm x_n} \sum_{n\in\mathcal{N}} g^{p}_n(\bm x_n)$,
    $\Gamma^\mathrm{asyn}$ is number of iterations, \ie,  
    $\Gamma^\mathrm{asyn} = \min \left\{t\mid \eta\left(\bm x^t, \bm\sigma^t\right) \leq \epsilon, t \geq 0\right\}$, $p^{\mathrm{asyn}} = 1 - \Pi_{n\in\mathcal{N}}(1-p^{\mathrm{max}}_n)$ denote the probability that at least one communication unit communicates successfully in an iteration and
$p^{\mathrm{asyn}}\Gamma^{\mathrm{asyn}}$ represents the number of successful iterations.
{\tr{According to detailed analysis and proof in Appendix \ref{app:GapFunction-Convergence} of technical report~\cite{qiao2023popec},}\fp{According to detailed analysis and proof in Appendix L of technical report~\cite{qiao2023popec}, }
we have
    \begin{equation}
         \Gamma^{\mathrm{asyn}}<\frac{k^{\Gamma}(L^{p}\left(\left\{\bm x_{n}^{1}\right\}, \bm x_{o}^{1}, \bm \sigma^{1}\right)-\underline{G}^{p})}{\epsilon^{\mathrm{ac}}p^{\mathrm{asyn}}},
    \end{equation}
which means Algorithm \ref{alg:AsynchronousADMM} converges to an $\epsilon^{\mathrm{ac}}$-stationary point within $O(1/(p^\mathrm{asyn}\epsilon^\mathrm{ac}))$.
Therefore, the asynchronous parallel algorithm is approximately $p^\mathrm{asyn}/p^\mathrm{syn}$ times faster, in comparison to the synchronous parallel approach according to Theorem \ref{theorem:innerConvergence1}.
The value of $p^\mathrm{asyn}/p^\mathrm{syn}$ is greater than 1, and it increases as the channel quality declines.
    % If $\rho_n>\max\{7\ell_n, \ell_n\left(\Gamma^2_n + \frac{3}{7}(\Gamma_n + 1)^2\right)\}$, the non-convex ADMM-Consensus method in Algorithm \ref{alg:AsynchronousADMM} converges and achieves global sub-linear convergence rate
    % \begin{equation}
    %     \Gamma^{\mathrm{asyn}}<\frac{k^{\Gamma}(L^{p}\left(\left\{\bm x_{n}^{1}\right\}, \bm x_{o}^{1}, \bm \sigma^{1}\right)-\underline{G^{p}})}{\epsilon^{\mathrm{ac}} p^{\mathrm{asyn}}},
    % \end{equation}
    % where $p^{\mathrm{asyn}} = 1 - \Pi_{n\in\mathcal{N}}(1-p^{\mathrm{max}}_n)$.
    % For more detailed information, please refer to Appendix \ref{app:theorem:innerConvergence2}.
}

{Such convergence analysis shows that Algorithm \ref{alg:AsynchronousADMM} converges faster than Algorithm    \ref{alg:nonconvexadmm}, particularly when transmission reliability is low.
Furthermore, we demonstrate the effectiveness of the asynchronous algorithm through numerous simulation experiments.
In each iteration, we are required to compute $\bm{x}{o}^{\tau{o} + 1} = \arg \min_{\bm{x}{o} \in \Omega} \mathcal{L}^p \left({\bm{x}^0_n}, \bm{x}{o}, {\bm{\sigma}^0_n}\right)$, where $\mathcal{L}^p$ is a convex problem. 
We can exploit widely-used gradient descent or interior point methods to solve this problem with low computational cost.}
In particular, we introduce an asynchronous parallel communication algorithm tailored for the issue of unreliable channels in priority-free cases as well
\tr{(see Appendix \ref{app:det:admm-consensus-C} of technical report~\cite{qiao2023popec} for more details).}\fp{(see Appendix D-C of technical report~\cite{qiao2023popec} for more details).

% \textbf{Complexity}
% Here, we define $N^{ec}$ as the required number of iterations by the checkpointing algorithm to reach convergence. According to Lemma \ref{theorem:innerConvergence2}, $\frac{1}{\epsilon_\mathrm{min}}$ is the number of iterations in the inner convergence, where $\epsilon_\mathrm{min} = \mathop{\min}_{n} \{\epsilon_{n}\}$. It means that we can obtain the strategy via Algorithm \ref{alg:NFP} and Algorithm \ref{alg:AsynchronousADMM} with the worst complexity of $O(\frac{N^{ec}}{\epsilon_\mathrm{min}})$.
% The complexity of the overall algorithm discussed as follows. The iterative algorithm for sloving \textbf{P2} via the external (\ie Algorithm \ref{alg:NFP})  and internal circulation (\ie Algorithm \ref{alg:AsynchronousADMM}) .
% \subsubsection{Task Migration in Multi-Servers\label{subsubsec:Data-Migration}}
% To minimize the age of the information, we will migration task from the local MEC to other free MEC.

\subsection{Why Multi-Class Priority?}
% Many multi-priority efforts focus on scenarios where each user has a predetermined level of priority~\cite{huang2015optimizing,xu2020peak,maatouk2019age}. While this approach appears straightforward and easy to implement, it has certain limitations that may result in the unjust treatment of users with similar priorities.
% Consider, for instance, situations where users who should be treated equally are given lower priorities. This would go against the principles of fairness and equality that multi-priority systems are intended to uphold.

% To ensure fairness in allocation, our multi-priority user model can handle different scenarios, including cases where users have unequal priorities or where multiple users share the same priority level.  This type of model is known as a multi-class priority model.
% Specifically, when each user has a distinct level of priority, it can be considered a subset or special case of multi-class priority.  In such cases, our model aims to prevent unfair treatment of users with similar priorities by allocating available resources more equitably.
{Existing multi-priority methods focus on scenarios where each user has a unique priority level, which can lead to the unjust treatment of users who should have the same priority, as exemplified in previous studies such as~\cite{huang2015optimizing, xu2020peak, maatouk2019age}.
This may undermine the principles of fairness and equality that these systems are intended to uphold. However, the multi-priority model can extend to handle different scenarios to promote fair allocation, including cases where users have unequal priorities or where multiple users share the same priority level. In our research, this model is called multi-class priority, which aims to prevent unfair treatment of users with similar priorities by allocating resources equitably.
To the best of our knowledge, our study represents the first attempt to explore such a  priority model with a broader range of practical applications.}

{Next, we explore the benefits of our proposed multi-class priority mechanism in contrast to the commonly used multi-priority and no-priority mechanisms. Specifically, we delve into the effects of our approach on the performance of three distinct user groups:}

% The unique priority of each user in multi-priority events has been extensively studied in the existing literature~\cite{xu2020peak}. 
% But clearly, this is unfair to various users who should have been given equal priority but are given relatively low priorities. 
    % Each user in the system has a unique priority in multi-priority events, a situation that has been covered in existing literature~\cite{xu2020peak}. But clearly, this is unfair to various users who ought to be given equal priority but are given relatively low priorities. We, therefore, take this issue into account when talking about the multi-class priority case. In turn, multi-priority can be considered as a special case of multi-class priority.
% \end{remark}
% Specifically, we compare the information freshness of users with different priority levels and users with a single priority (\ie no-priority mechanism) using the same offloading strategy.
% \begin{lemma}
% \label{proof:sp-mp}
{\paragraph{The highest priority users} 
The PAoI of user $n*$, who are set to the highest priority ($\mathbb{E}[A^{p}_{n^*}]$), has more optimal information freshness than the priority-free case ($\mathbb{E}[A_{n^*}]$) under the same offloading strategy, which is
\begin{align}
    &\mathbb{E}[A_{n^*}] - \mathbb{E}[A^{p}_{n^*}]
    \nonumber\\
    = &\Upsilon(\bm\eta^u) \frac{\zeta_{\Delta}(\bm\eta^u)-\zeta_{\delta(n^*)}(\bm\eta^u)}{(1-\zeta_{\Delta}(\bm\eta^u))(1-\zeta_{\delta(n^*)}(\bm\eta^u))}\nonumber\\
    \geq& 0.
\end{align}
where 
$\zeta_{\delta(n^*)}(\bm\eta^u) =\sum_{\delta\in\Delta(\delta(n^*))}\sum_{n'\in\mathcal{N}^{\delta}}\sum_{c\in\mathcal{C}}p_{n',c}\frac{\eta^u_{n',c}\lambda_{n'}}{\mu_{n'}}$,  
$\zeta_{\Delta}(\bm\eta^u) =\sum_{\delta\in\Delta}\sum_{n'\in\mathcal{N}^{\delta}}\sum_{c\in\mathcal{C}}p_{n',c}\frac{\eta^u_{n',c}\lambda_{n'}}{\mu_{n'}}$,  and
$\Upsilon(\bm\eta^u) = \frac{1}{2}\sum_{\delta\in\Delta}\sum_{n'\in\mathcal{N}^{\delta}}\sum_{c\in\mathcal{C}}p_{n',c}\eta^u_{n',c}\lambda_{n'}\nu_{n'}$.
This formula indicates that prioritized offloading systems yield greater advantages for users with high real-time requirements.
\tr{(For more proof details, please refer to Appendix \ref{app:proof:sp-m:a} of technical report~\cite{qiao2023popec}.)}\fp{(For more proof details, please refer to Appendix N-A of technical report~\cite{qiao2023popec}.)}
}

{\paragraph{The lowest priority users} 
Similarly, the PAoI of user $n_*$, who is set to have the lowest priority ($\mathbb{E}[A^{p}_{n_*}]$), has a worse PAoI than the priority-free case ($\mathbb{E}[A_{n_*}]$). Thus, we gain
\begin{align}
    &\mathbb{E}[A_{n_*}] - \mathbb{E}[A^{p}_{n_*}] 
    \nonumber\\
    =& \frac{\Upsilon(\bm\eta^u)} {1-\zeta_{\Delta}(\bm\eta^u)}(1-\frac{1}{1-\zeta_{\delta(n)-1}(\bm\eta^u)})\nonumber\\
    \leq& 0.
\end{align}
This formula demonstrates that the potential drawback of a priority offloading system is the insufficient guarantee of information freshness for users with lower priority.
\tr{(For more proof details, please refer to Appendix \ref{app:proof:sp-m:b} of technical report~\cite{qiao2023popec}.)}\fp{(For more proof details, please refer to Appendix N-B of technical report~\cite{qiao2023popec}.)}
}
{
\paragraph{The higher priority users} 
Furthermore, multiple users are assumed to have the same priority level but are instead assigned priorities of $\delta^0$ and $\delta^0-\hat\delta$ in the multi-priority scenario, the resulting difference in freshness can be substantial and unfair:
\begin{align}
    &\mathbb{E}[A^{p}_{n}|\delta(n)=\delta^0] - \mathbb{E}[A^{p}_{n}|\delta(n)=\delta^0-\hat\delta]
    \nonumber\\
    \ge& \frac{\Upsilon(\bm\eta^u)(\zeta_{\delta^0}(\bm\eta^u)-\zeta_{\delta^0 - \hat\delta - 1}(\bm\eta^u))}{(1-\zeta_{\delta^0}(\bm\eta^u))(1-\zeta_{\delta^0 - 1}(\bm\eta^u))(1-\zeta_{\delta^0 - \hat\delta - 1}(\bm\eta^u))}\nonumber\\
    \geq& 0.
\end{align}
This formula reveals that the most straightforward and impact approach to getting fresher information is elevating the user's priority; otherwise, its priority would be lowered.
\tr{(For more proof details, please refer to Appendix \ref{app:proof:sp-m:c} of technical report~\cite{qiao2023popec}.)}\fp{(For more proof details, please refer to Appendix N-C of technical report~\cite{qiao2023popec}.)}
Based on the conclusions drawn, we have:
\begin{align}
    &\mathbb{E}[A^{p}_{n}|\delta(n)=\delta^h] - \mathbb{E}[A^{p}_{n}|\delta(n)=\delta^l]\nonumber\\
    =&\mathbb{E}[A^{p}_{n}|\delta(n)=\delta^l] - \mathbb{E}[A^{p}_{n}|\delta(n)=\delta^l-\hat\delta)]\nonumber\\
    \leq& 0,
\end{align}
where $\hat\delta=\delta^l-\delta^h>0$ represents the differential priority level, with $\delta^{h}$ and $\delta^{l}$ denoting high and low priorities, respectively.
This formula reveals that high-priority users are guaranteed to have lower PAoI values compared to low-priority users.
}

% \end{lemma}
% \end{proof}
{This subsection highlights that user $n$ can get superior performance as compared to priority-free by allocating it a high priority. 
Furthermore, allocating user $n$ to a higher level also enhances its performance and allows it to obtain more up-to-date information at the expense of users with lower priority.
We substantiate these claims with comprehensive simulations in the next section.}

\section{Simulation}\label{sec:simulation}
{In this section, we evaluate our proposed algorithm by answering the following questions:
\begin{enumerate}
    \item What is the overall utility of our method?
    \item Can it schedule multi-priority tasks effectively?
    % \item How does PoPeC perform with different quality of service (\S \ref{subsec:Exp-ServiceQuality})?
    % \item Can PoPeC obtain the scheduling policy efficiently and reliably in heterogeneous unreliable channels(\S \ref{subsec:Exp-ServiceQuality})?
    \item Can it deal with heterogeneous and unreliable channels effectively?
\end{enumerate}
}

\subsection{Experimental Setup}\label{subsec:Exp-Setup}
{\textit{\textbf{Parameters.}}
% We conducted extensive simulations to assess the efficacy, performance, and computational efficiency of our proposed method.  To this end, we considered a priority-free case consisting of $M = 10$ servers and $N = 200$ users within their respective coverage regions.  For the multi-priority case, we considered a minimum of three priorities, and users were allocated to different priority levels.  
In this section, we carry out simulations to evaluate the effectiveness, performance, and computational efficiency of our proposed method. In the priority-free case, we consider 10 servers and 200 users within their respective coverage areas. For the multi-priority case, we examine at least three priorities, allocating users to different priority levels. 
Moreover, we model the transmission success probability of the channel as a Gaussian distribution $N(0.5,1)$, following~\cite{scutari2008asynchronous}.  We set the number of available channels $C$ to be 30, with a bandwidth of $B = 5 MHz$, and channel gains were set to unity as in~\cite{balasubramanian2009energy}.  
% Given the heterogeneity of users and servers, the service time of tasks was assumed to follow a general distribution, where the mean and variance of the distribution were determined by the types of users and servers.  
To account for the heterogeneity of users and servers, the service time of tasks followed a general distribution, where the mean and variance of the distribution are determined by the types of users and servers. Specifically, we set the value of the mean and variance of the general distribution to follow the uniform distribution $U(1,5)$ and $U(1,25)$, respectively.}
% We use extensive simulations to evaluate the cost, performance, and computation time of our method. 
% We consider a priority-free case with $M = 10$ servers and $N = 200$ users in its coverage regions. In the multi-class priority case, there are at least three priorities, and users are divided into various priorities. Moreover, the transmission success probability of the channel is regarded as a Gauss distribution $N(0.5,1)$\cite{scutari2008asynchronous}.
% We set the number of channels $C$ as 30. Furthermore, the channel bandwidth is set to $B = 5 MHz$, and the channel gains are set to 1~\cite{balasubramanian2009energy}. Due to the heterogeneity of users and servers, the service time of tasks should follow a general distribution, and the mean and variance of the general distribution are related to the types of users and servers. We set the mean and variance of the general distribution to follow $U(1,5)$ and $U(1,25)$, respectively.

{\textit{\textbf{Performance metrics.}}
In the following experiments, we mainly use PAoI and throughput as performance metrics. A lower PAoI signifies fresher information for the user and a reduced number of outdated tasks. Conversely, a higher throughput indicates better utilization of communication and computation resources within the same experimental setup.
}

\begin{table}[htpb]
	\caption{Parameters in Simulation}
	\label{table:parameters}
	\centering
	\renewcommand\arraystretch{1.25}
	\resizebox{\columnwidth}{!}{
	\begin{tabular}{l|l}
		\hline
		\textbf{Parameter} & \textbf{Value}\\ 
		\hline
		\# users ($N$) & an integer varying between $[1,200]$\\
		\# channels ($C$) & an integer varying between $[1,30]$\\
		\# servers ($M$) & an integer varying between $[1,10]$\\
		the channel condition & 
		    \begin{tabular}[c]{@{}l@{}}
		        channel bandwidth ($B$) as 5MHz\\
		        channel gains ($R_c$) as 1\\
		    \end{tabular} \\
		the service time & 
		    \begin{tabular}[c]{@{}l@{}}
		        obey the general distribution\\
		        with the mean $\mu$ following $U(1,5)$, \\ 
		        and the variance $\nu$ following $U(1,25)$, respectively
		    \end{tabular} \\                  
		the transmission rates  & real numbers (Mbit/slot) varying between $[0,0.5]$\\
		task generation rates & real numbers (Mbit/slot) following $U(0.5,1.5)$ \\                                           \hline
	\end{tabular}
	}
\end{table}

\textit{\textbf{Baselines.}}
We compare our proposed algorithm with existing algorithms in the literature to perform a comprehensive analysis. Specifically, we compare our method with the Age-Aware Policy (AAP) algorithm which utilizes throughput constraints through the Lyapunov optimization method~\cite{sun2021age}. We also consider the Greedy Control Algorithm (GA), which selects the most reliable channel among the unreliable channels. Additionally, to account for the lack of priority mechanism in AAP and GA, we compare our algorithm with the Priority Scheduling method of Peak Age of Information in Priority Queueing Systems (PAUSE)~\cite{xu2020peak} and the Rate and Age of Information (RAI) method~\cite{abdollahi2022rate}.
To ensure a fair comparison, we assume that each user sends the maximum possible number of tasks to the edge server, and the edge server completes the tasks in a First-Come-First-Serve (FCFS) manner.
% We compare our proposed algorithm with the algorithm which develops an Age-Aware Policy (AAP) with throughput constraints based on the Lyapunov optimization method in~\cite{sun2021age}. We also consider a greedy control algorithm (GA) to select the most reliable channel in each unreliable channel. Consider the case where each user sends as many tasks as possible to the edge server, and the edge server completes the task with FCFS. Due to the lack of priority mechanism in AAP and GA, we also refer to the priority scheduling method of Peak Age of Information in Priority Queueing Systems (PAUSE)~\cite{xu2020peak} and Rate and Age of Information (RAI)\cite{abdollahi2022rate}.

\textit{\textbf{Implementation.}}
The simulation platform is Matlab R2019a and all the simulations are performed on a laptop with 2.5 GHz Intel Core i7 and 16 GB RAM.

\subsection{Overall Utility}\label{subsec:Exp-Utility}
{
In the assessment of the overall utility, we conducted an analysis of various metrics, such as PAoI and throughput, under different methods.
We further compare our PAoI-based approach with the latency-based method and the weight-based method.
Moreover, we examined their performance in various settings, including priority allocation and server collaboration.  
}

% Firstly, as depicted in Fig.\ref{fig:PTvsN} and \ref{fig:PTvsM}, the blue bars show throughput, and the black lines depict PAoI. Our evaluation shows that the number of users $N$ and channel capacity $M$ significantly impact the PAoI value and throughput. Our proposed algorithm (OUR) outperforms both the Greedy Control Algorithm (GA) and the Age-Aware Policy (AAP) algorithm, especially in terms of PAoI. However, in resource-constrained environments, specifically in Fig.\ref{fig:PTvsM}, our algorithm's throughput performance is slightly inferior since it does not have the same higher needs and throughput constraints as the (AAP) algorithm.
{Firstly, we compare the performance of our proposed algorithm (OUR) with the Greedy Control Algorithm (GA) and the Age-Aware Policy (AAP) algorithm in multi-user and multi-server cases.
As depicted in Fig.\ref{fig:PTvsN} and \ref{fig:PTvsM}, the blue bars represent throughput, while the black lines illustrate the Packet Age of Information (PAoI). 
Our evaluation underscores that the PAoI value and throughput are notably influenced by the number of users ($N$) and servers ($M$). 
% Remarkably, OUR surpasses both the GA and AAP algorithms, particularly in terms of PAoI.
% OUR algorithm's throughput performance is slightly inferior due to its specific priorities and throughput constraints not aligning with those of the AAP algorithm.
Overall, OUR's main strength lies in its consistently superior overall performance in comparison to GA and AAP.
Notably, OUR's PAoI excels across diverse parameter settings compared to GA and AAP.
However, in scenarios characterized by resource limitations, as illustrated in Fig.\ref{fig:PTvsM}, the throughput performance of our algorithm is slightly lower.
This is attributed to OUR's predominant emphasis on minimizing PAoI, whereas the AAP method inherently prioritizes throughput with its foundational constraint.
Although throughput is not OUR's primary focus, 
% OUR's throughput results remain competitive with those of AAP and GA. 
% While its throughput might not be the highest, 
it achieves an optimal or at least suboptimal level when contrasted with the other two algorithms.
}
\begin{figure}[ht!]
	\centering
	\subfigure[PAoI \& Throughout vs. $N$]{\label{fig:PTvsN}\includegraphics[width=0.225\textwidth]{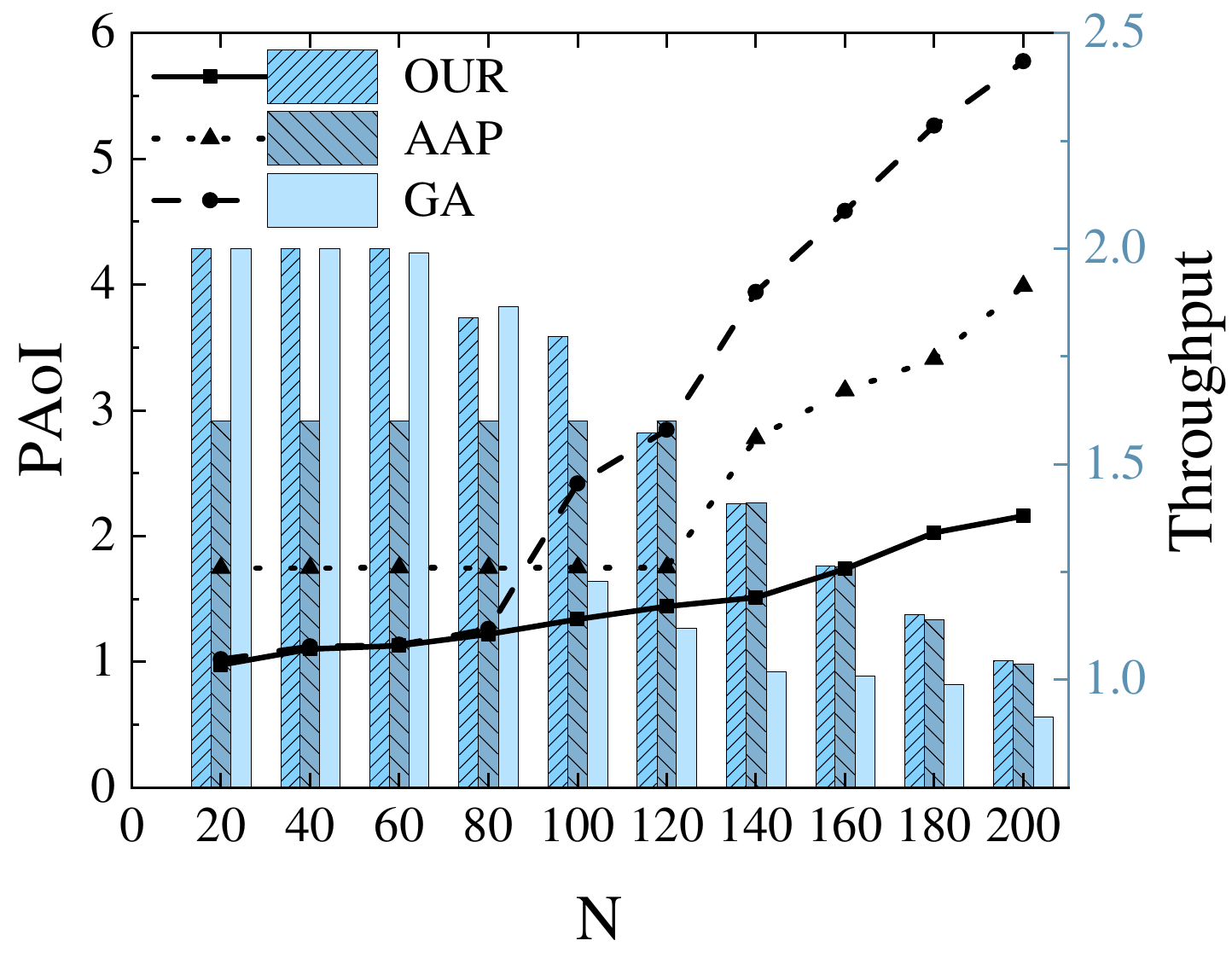}}
	\subfigure[PAoI \& Throughout vs. $M$]{\label{fig:PTvsM}\includegraphics[width=0.225\textwidth]{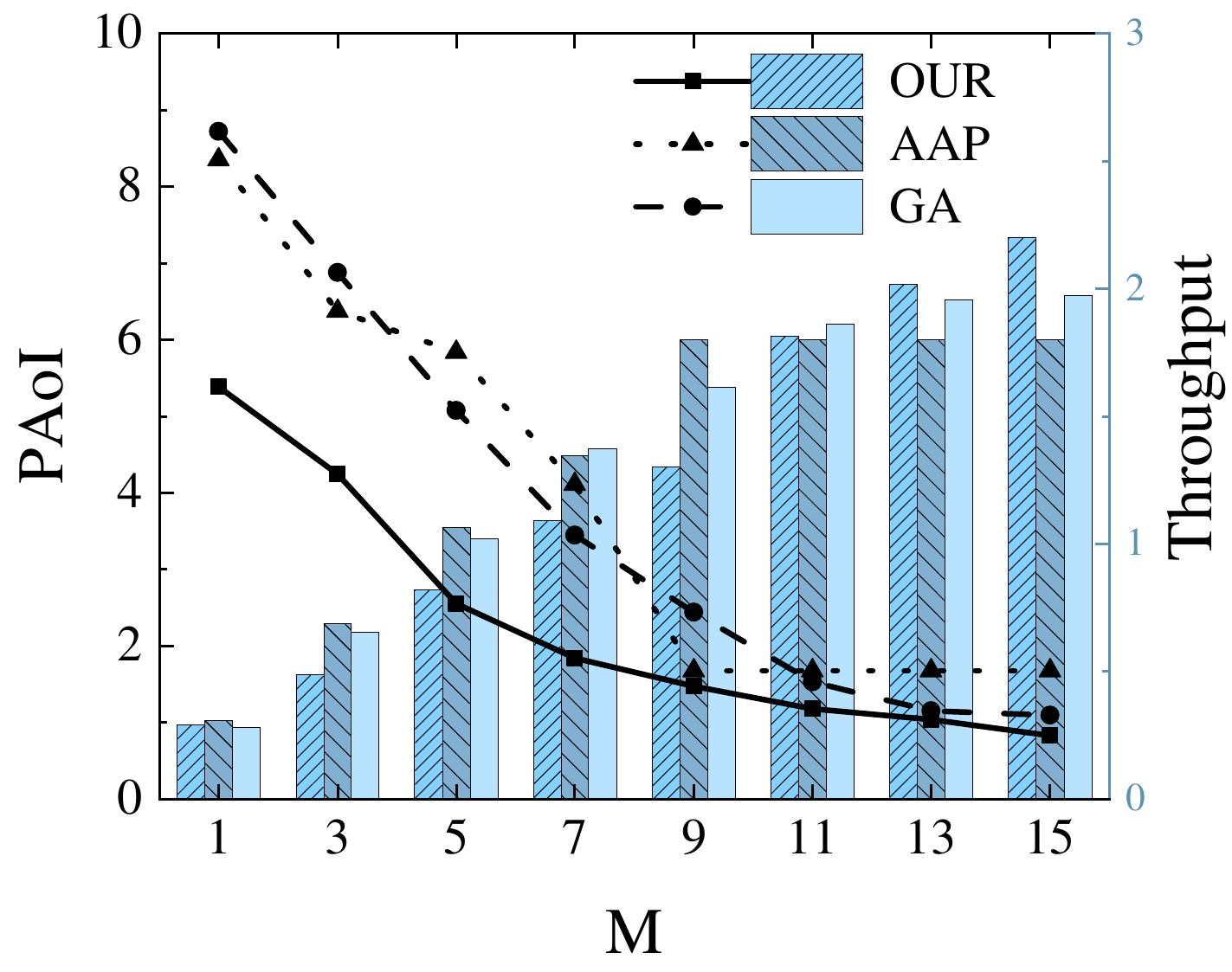}}
	\caption{Impact of methods on various metrics.}
    % \vspace{-4pt}
\end{figure}

{
In our second set of analyses, we compared our approach, which utilizes strict priority control, with latency-based methods and weight-based (flexible priority control).
We explore the latency-based method, focusing on the performance metric of arrival interval. As depicted in Fig.\ref{fig:OURvsLatency}, we found that the latency-based method underperforms in resource-constrained scenarios. The high arrival intervals, indicative of reduced frequency of updates, emerged as a notable limitation in latency-based methods. This shortfall makes them less adept for applications like AR/VR, where there is a critical need for swift information updates \cite{guo2021scheduling}.
We further compared our approach, which employs hard priority control, with the weight-based method that uses flexible priority control. The results, illustrated in Fig.\ref{fig:OURvsWeight}, showed that our approach effectively reduces the Peak Age of Information (PAoI) for high-priority users. This indicates a more efficient delivery of timely information compared to the weight-based method. The advantage of our approach stems from its server-side priority queuing system, which ensures high-priority tasks are processed more promptly. This experiment underlines the effectiveness of our PAoI-based method in scenarios where speed and accuracy of information processing are paramount.
% In our second analysis, we compared our approach, which utilizes strict priority control, with weight-based (flexible priority control) and latency-based methods. The experiments revealed that our approach significantly reduces the Peak Age of Information (PAoI) for high-priority users, as shown in Fig. \ref{fig:OURvsWeight}, indicating more timely information delivery compared to the Weight approach. This improvement is attributed to our method's priority queuing system on the server side, ensuring a hard-priority service for urgent tasks. Additionally, when comparing with latency-based methods, we found that these approaches underperform in resource-constrained scenarios, particularly unsuitable for time-sensitive applications like AR/VR, as highlighted in \cite{guo2021scheduling}. Our findings suggest a clear advantage of our PAoI-based method in environments where quick and accurate information processing is essential.
}
\begin{figure}[ht!]
	\centering
	\subfigure[Arrival Internal vs. $\mu$]{\label{fig:OURvsLatency}\includegraphics[width=0.212\textwidth]{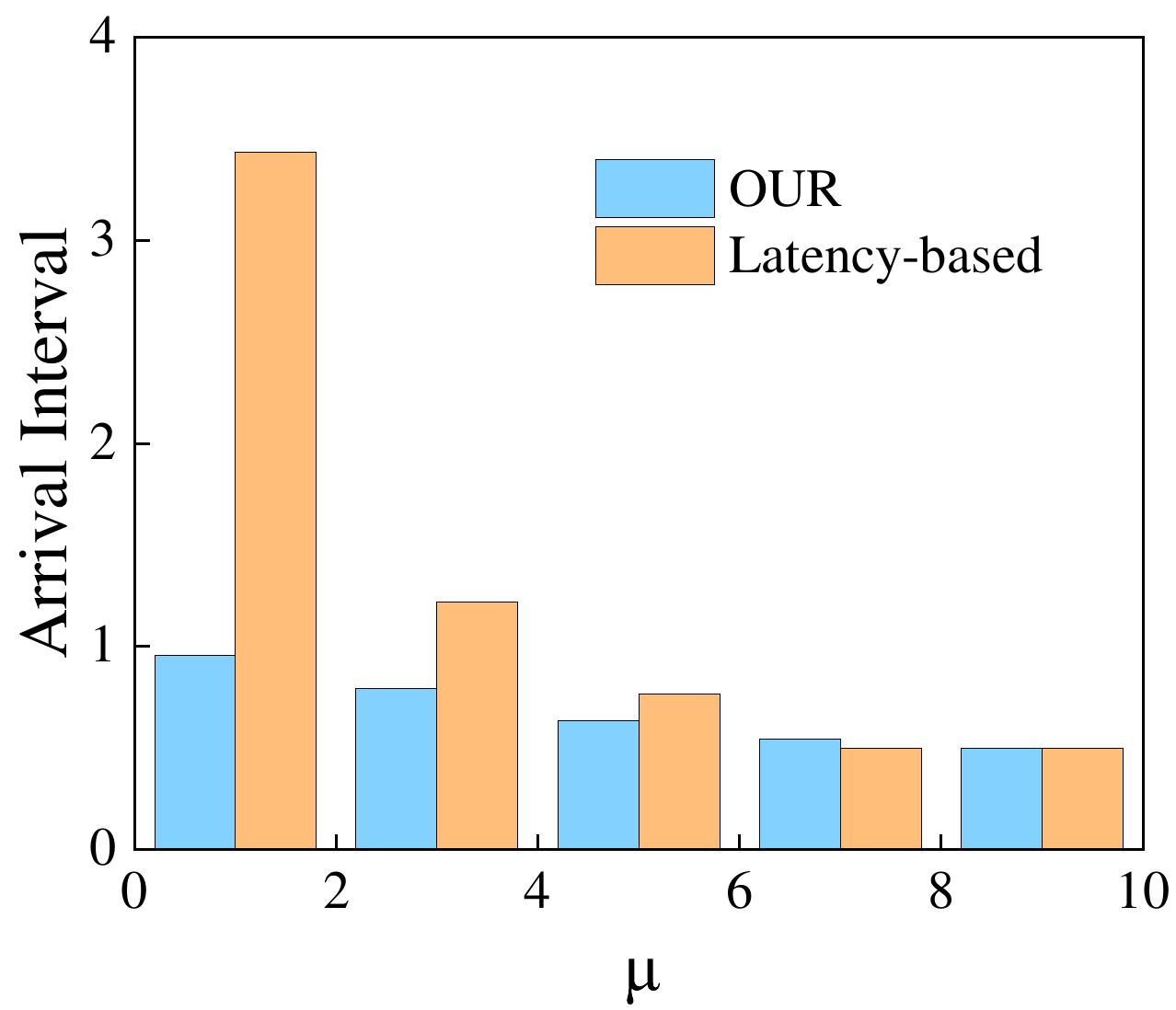}}
	\subfigure[PAoI vs. $\mu$]{\label{fig:OURvsWeight}\includegraphics[width=0.265\textwidth]{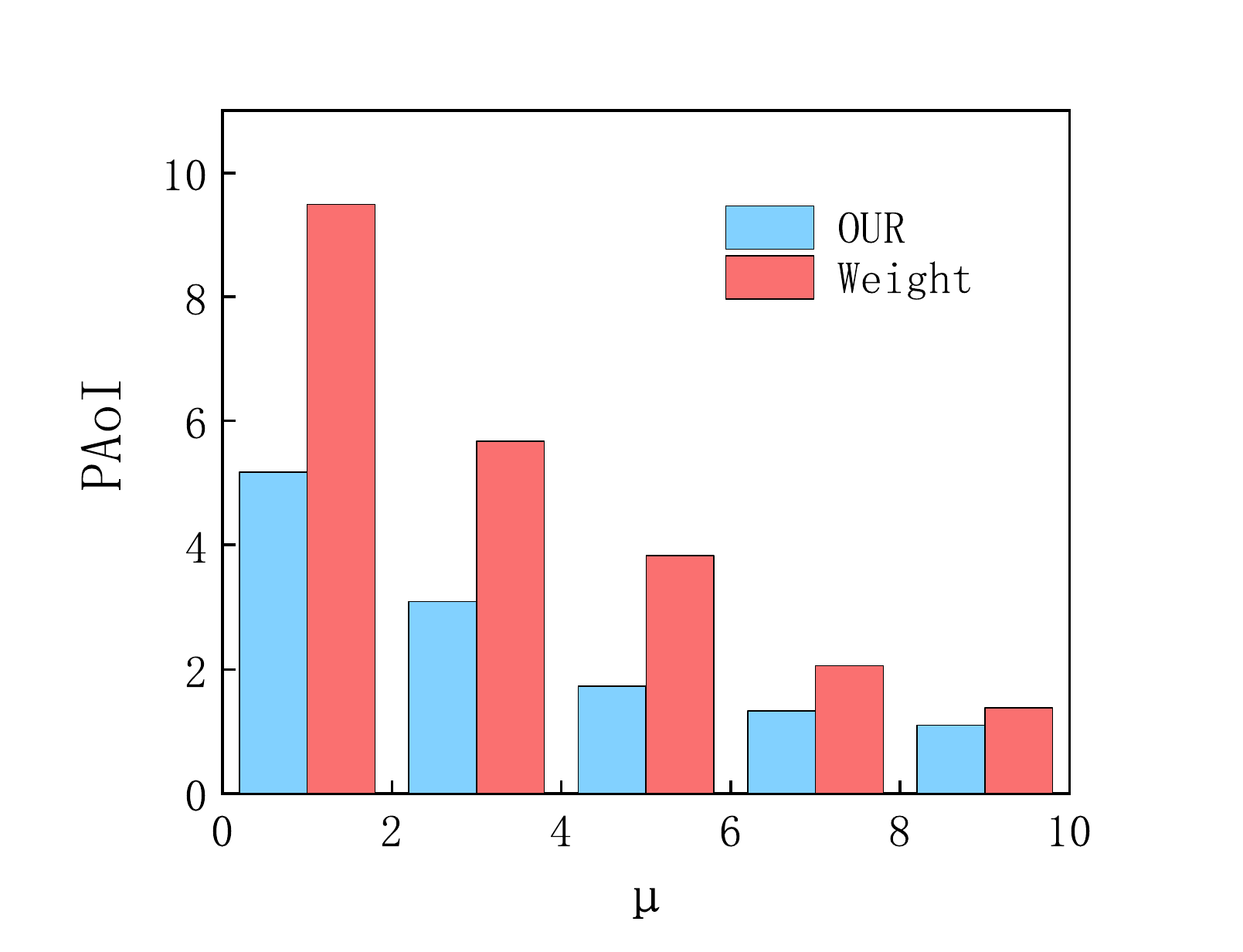}}
	\caption{Performance comparison between different methods.}
    % \vspace{-4pt}
\end{figure}

{Our third segment of analysis focuses on the impact of the priorities division and server collaboration on the overall PAoI performance.
We conducted a simulation for the multi-priority case, as illustrated in Fig.\ref{fig:PAoIvsMu}. Here, $\Delta$ denotes the degree of priority division, with $\Delta=1$ indicating no priority distinction and $\Delta=6$ indicating six priority classes. The simulation results show that the average PAoI of the system is mainly determined by the computing power of the server ($\mu$, task execution time), and the priority division level has little influence on it.
Furthermore, we investigate the cases of whether servers are collaborating.
% We find that server coordination is feasible and can significantly reduce the task execution time without affecting the update frequency. 
Fig.\ref{fig:PAoIvsN} presents the results for both the server collaboration case and the without-server collaboration case, showing that the PAoI performance of the former outperforms the latter regardless of the number of users. Notably, this feature is more prominent as the number of users $N$ increases.
Besides, the PAoI variance in the server collaboration scenario is notably lower, suggesting its stronger system stability.
}

% Secondly, for the multi-class priority scenario, the simulation in Fig.\ref{fig:PAoIvsMu} generalizes the algorithm to the multi-class first scenario. The value of $\Delta$ represents the degree of priority breakdown, with $\Delta=1$ indicating no priority distinction and $\Delta=6$ indicating six priority classes. The final results show that the average PAoI of the system is related to the computing power of the server ($\mu$, task execution time), and  has little to do with the priority division level.
% For multi-server collaboration scenarios, multi-server coordination is actually feasible because it significantly reduces the task execution time without affecting the update frequency.  
% In Fig.\ref{fig:PAoIvsN}, we consider two cases: single-server execution and multi-server collaboration, in which can be found that regardless of the number of users, the multi-server collaboration is always better than the single-server case, and the advantage is more obvious as the number of users increases.

\begin{figure}[ht!]
	\centering
	\subfigure[PAoI vs. $\mu$ in different priority allocation case]{\label{fig:PAoIvsMu}\includegraphics[width=0.220\textwidth]{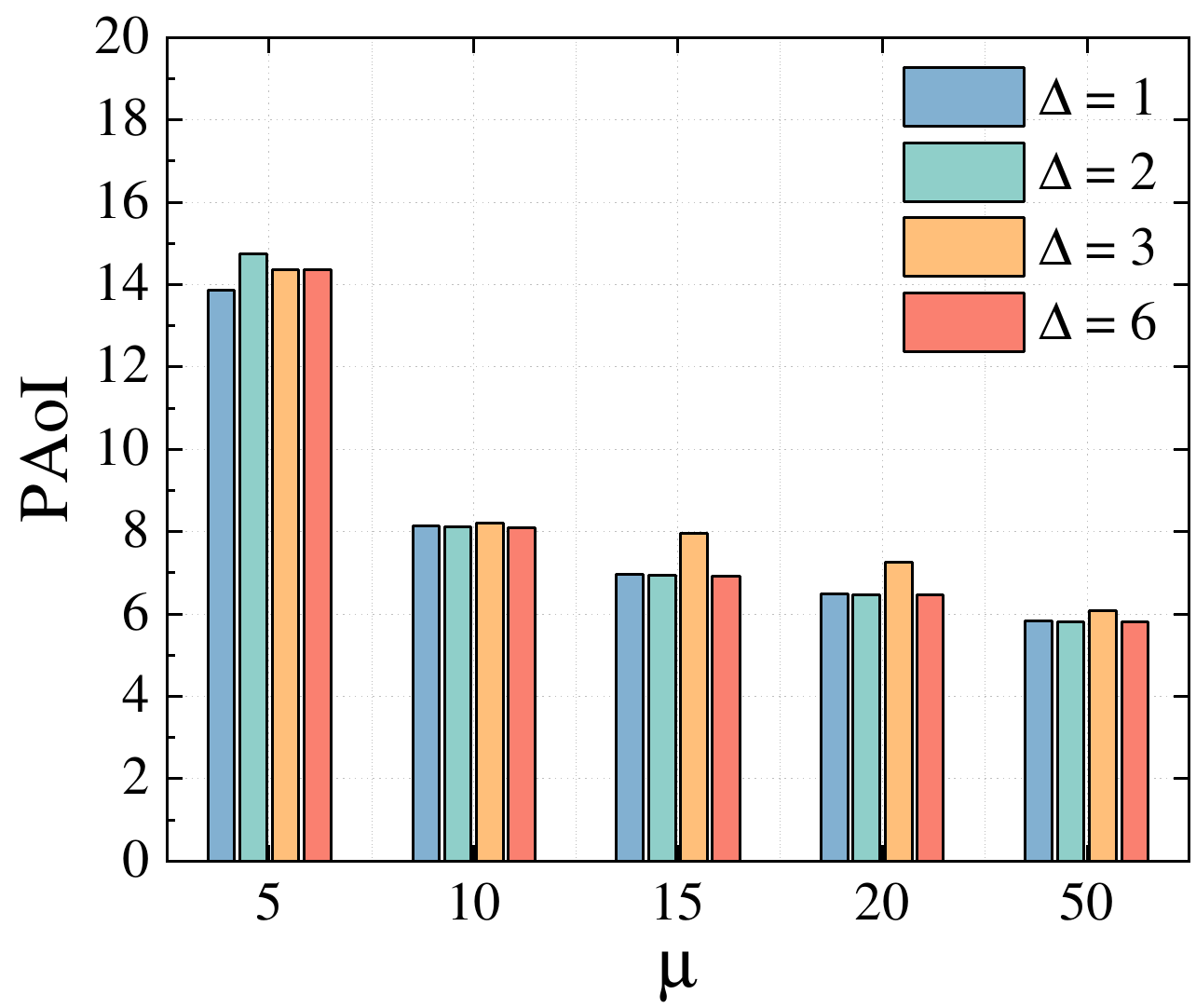}}
	\subfigure[PAoI vs. $N$ in the case of whether servers are collaborating]{\label{fig:PAoIvsN}\includegraphics[width=0.225\textwidth]{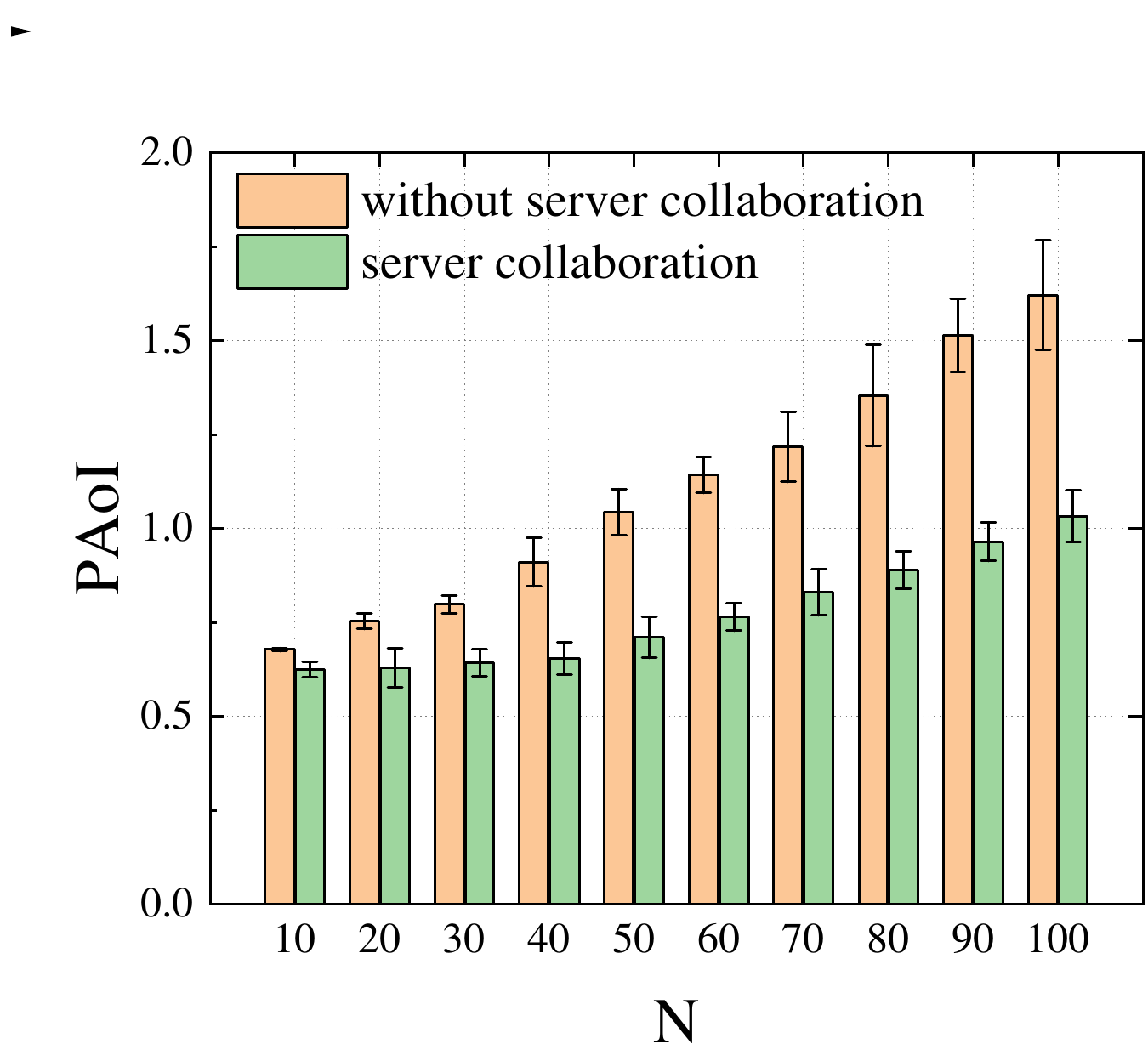}}
	\caption{The performance in various settings.}
    % \vspace{-4pt}
\end{figure}

% \begin{figure}[ht!]
%     \vspace{-0cm} 
%     \setlength{\abovecaptionskip}{-0cm} 
%     \setlength{\belowcaptionskip}{-0cm} 
%     \centering
%     \includegraphics[width=0.275\textwidth]{./figure/Graph2-2.pdf}
%     \vspace{-0pt}
%     \caption{PAoI vs. $N$}
%     \label{fig:2.2}
%     \vspace{-18pt}
% \end{figure}

\subsection{Performance of Priority Tasks}\label{subsec:Exp-Priority}
{Simulation results show that the proposed algorithm can be extended to multi-class priority scenarios.
As shown in Fig.\ref{fig:PAoIvsNinServers}, we set three priority classes ($\delta$ lower, higher priority) and then randomly allocated these priority levels to an equal number of users, each of whom had the same quantity of tasks.
As the iterations of our algorithm proceed, the average PAoI values of users with various priority levels steadily decrease and eventually tend to stabilize, and the offloading rate gradually increases and eventually tends to be stable.
We find that users with higher priority ($\delta=1$) always get higher offloading rates, which results from more channel resource allocations.
% Higher offloading rates tend to mean higher update frequencies.
% In addition to higher update frequencies, users are more likely to finish their tasks promptly.
In addition, high-priority users' PAoI values are lower since their tasks are scheduled and executed promptly.
The above observation implies that high-priority users receive a larger share of channel and computing resources and they are more likely to achieve superior performance.
}
\begin{figure}[ht!]
	\centering
	\subfigure[Performance of users with different priorities]{\label{fig:PAoIvsNinServers}\includegraphics[width=0.25\textwidth]{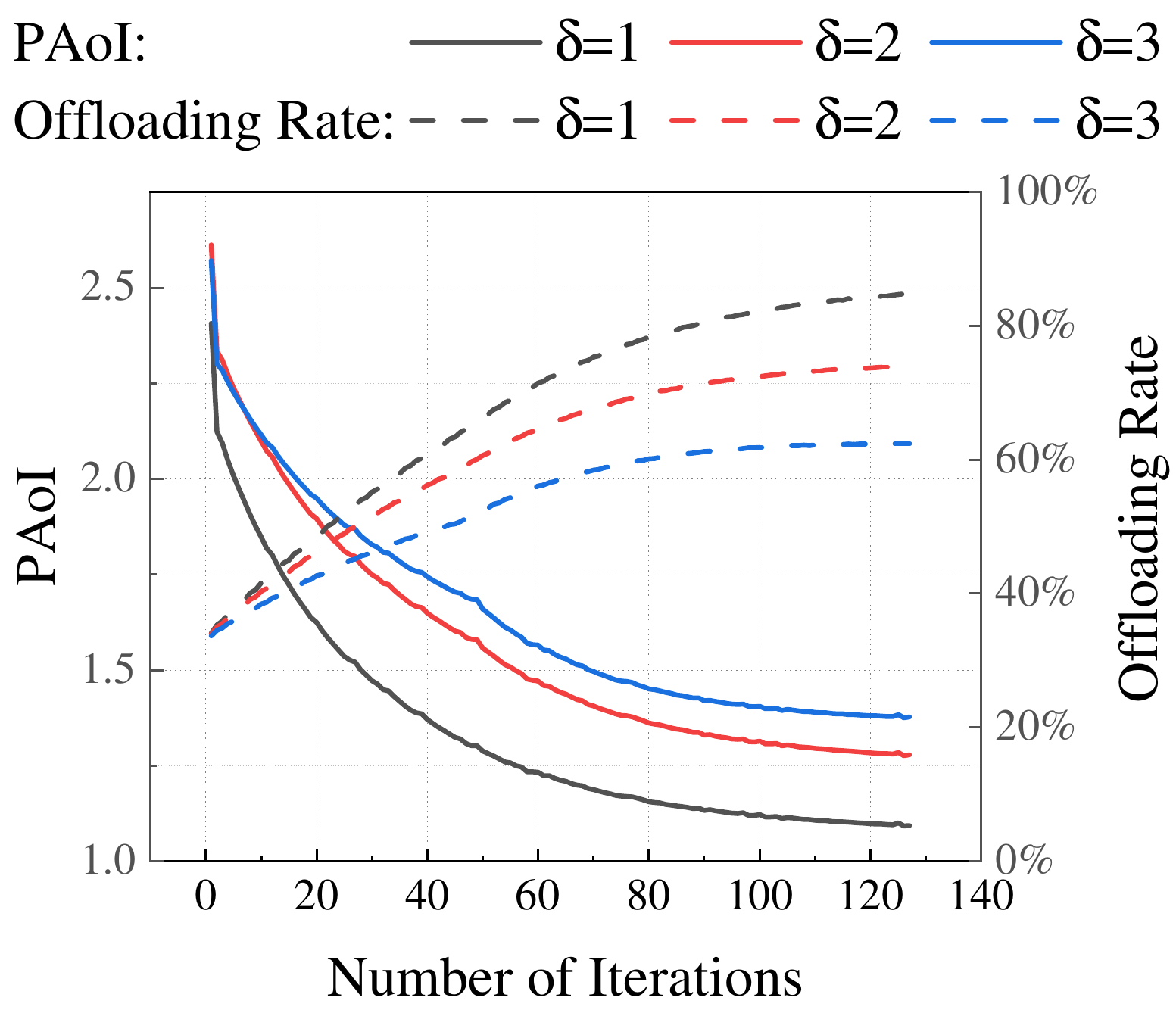}}
	\subfigure[Comparison of high priority users in different algorithms]{\label{fig:mp-compare}\includegraphics[width=0.205\textwidth]{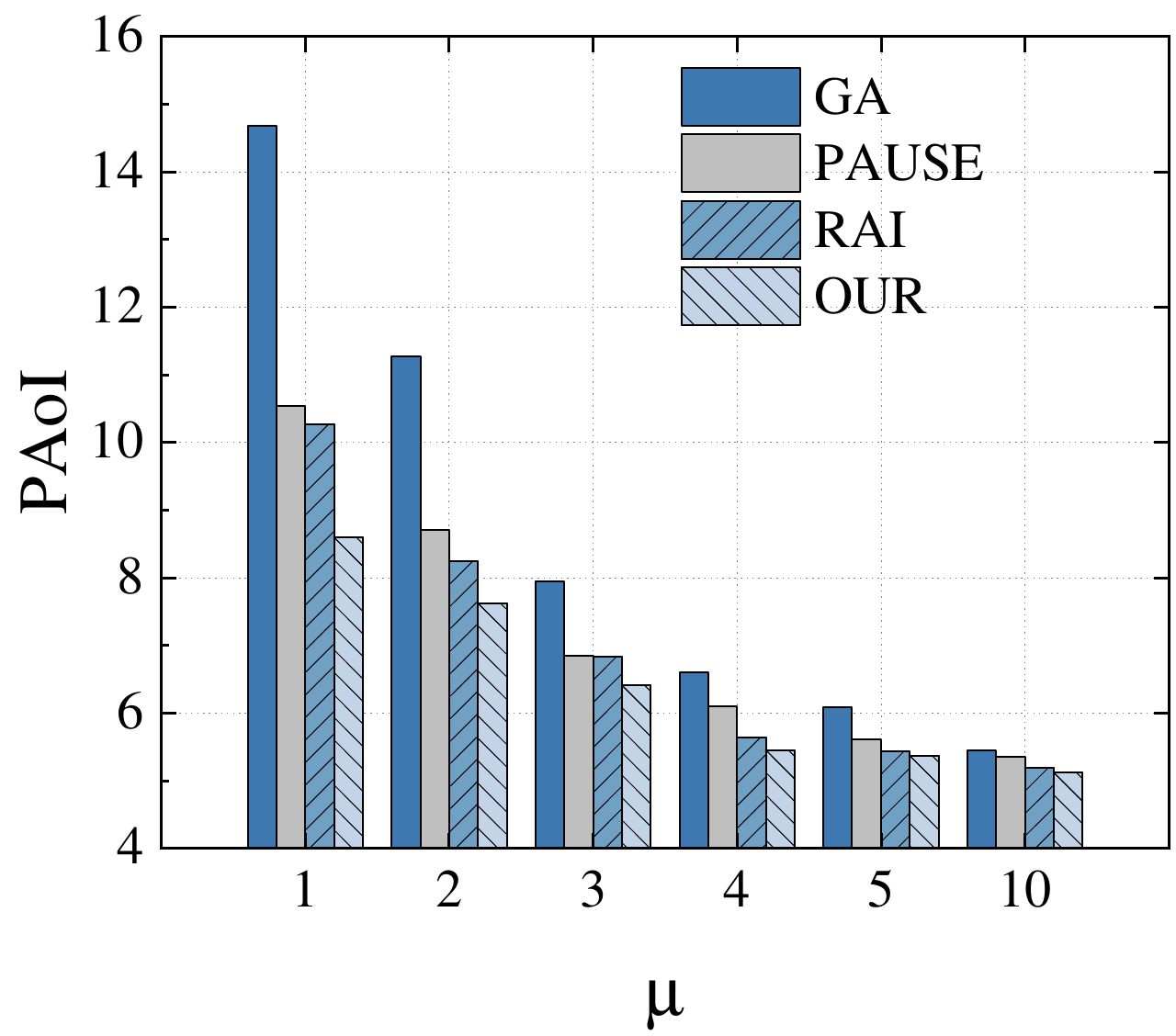}}
	\caption{Performance of Priority Users.}
    % \vspace{-4pt}
\end{figure}

% \begin{figure}[ht!]
%     \vspace{-0cm} 
%     \setlength{\abovecaptionskip}{-0cm} 
%     \setlength{\belowcaptionskip}{-0cm} 
%     \centering
% 	\includegraphics[width=0.275\textwidth]{./figure/Graph3-1.pdf}
% 	\caption{Performance of users with different priorities}
%     \vspace{-0pt}
%     \label{fig:PAoIvsNinServers}
%     \vspace{-9pt}
% \end{figure}
{With three types of priorities (high priority, medium priority, low priority), we compare the proposed multi-class priority method (OUR) with other algorithms (GA, PAUSE, RAI).
We pay special attention to the promotion effect of our algorithm for high-priority users, as shown in Fig.\ref{fig:mp-compare}.
% The PAoI values of the high-priority users all decrease in the simulation as the server processing rate $\mu$ rises.
% Moreover, it is clear that the proposed strategy OUR is always the best one with the lowest PAoI.
The GA algorithm is one that does not consider priorities, the PAUSE algorithm focuses on the discussion of multiple priorities rather than multiple classes of priorities, and the RAI algorithm only considers two classes of priorities.
The generality of these methods falls short, and they are unable to provide high-priority users with the lower PAoI that they require.
However, it is clear that the proposed strategy our approach is always the best one with the lowest PAoI and can be easily implemented across a range of scenarios.}

% \begin{figure}[ht!]
%     \vspace{-0cm} 
%     \setlength{\abovecaptionskip}{-0cm} 
%     \setlength{\belowcaptionskip}{-0cm} 
%     \centering
% 	\includegraphics[width=0.275\textwidth]{./figure/Graph3-2.pdf}
% 	\caption{Compare with different algorithms in multi-class priority case}
%     \vspace{-0pt}
%     \label{fig:mp-compare}
%     \vspace{-9pt}
% \end{figure}

% \subsection{The Impact of Service Quality}\label{subsec:Exp-ServiceQuality}
\begin{figure}[ht!]
	\centering
	\subfigure[High server processing rates]{\label{fig:Graph_highmu}\includegraphics[width=0.232\textwidth]{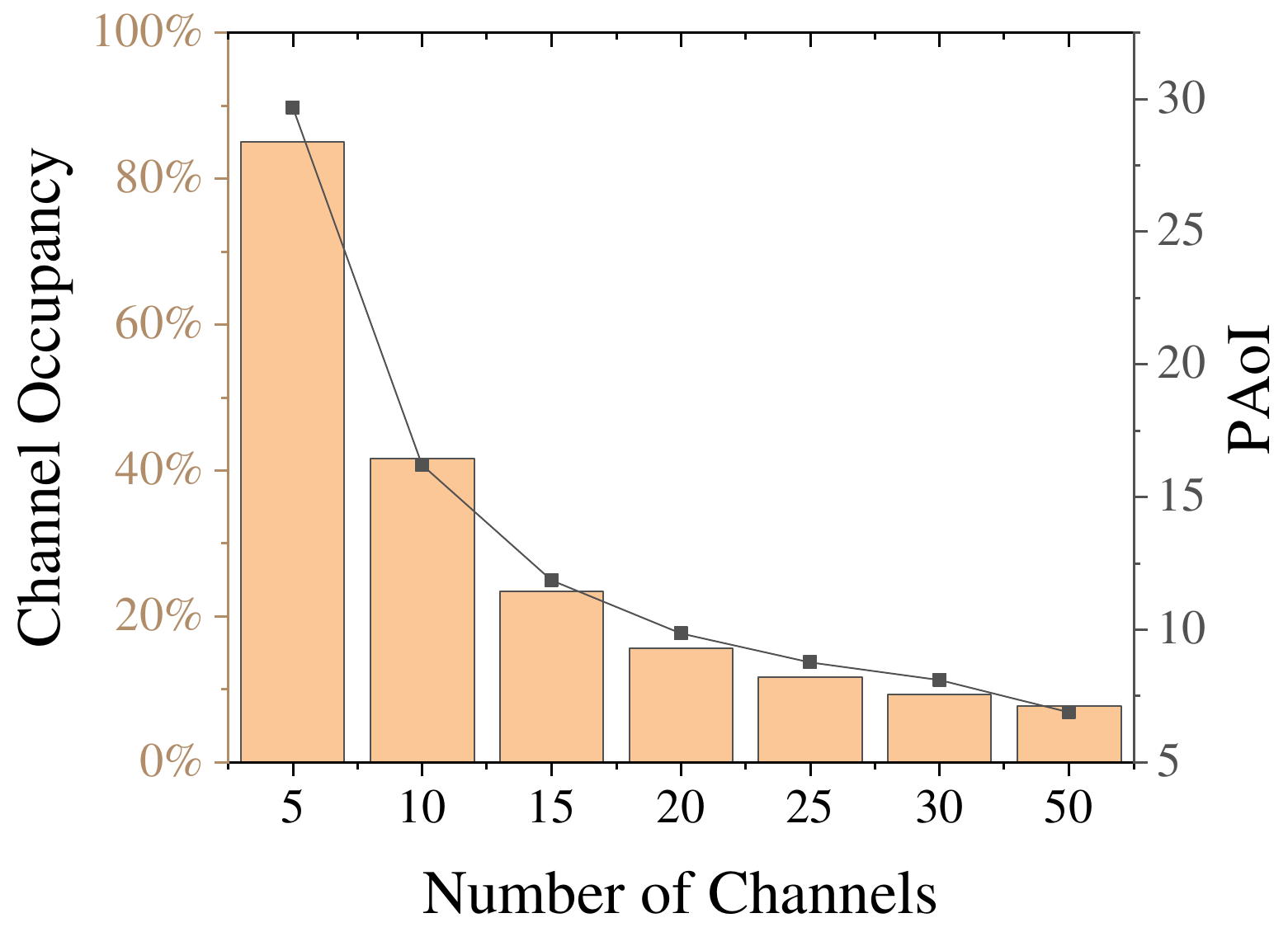}}
	\subfigure[Low server processing rates]{\label{fig:Graph_lowmu}\includegraphics[width=0.225\textwidth]{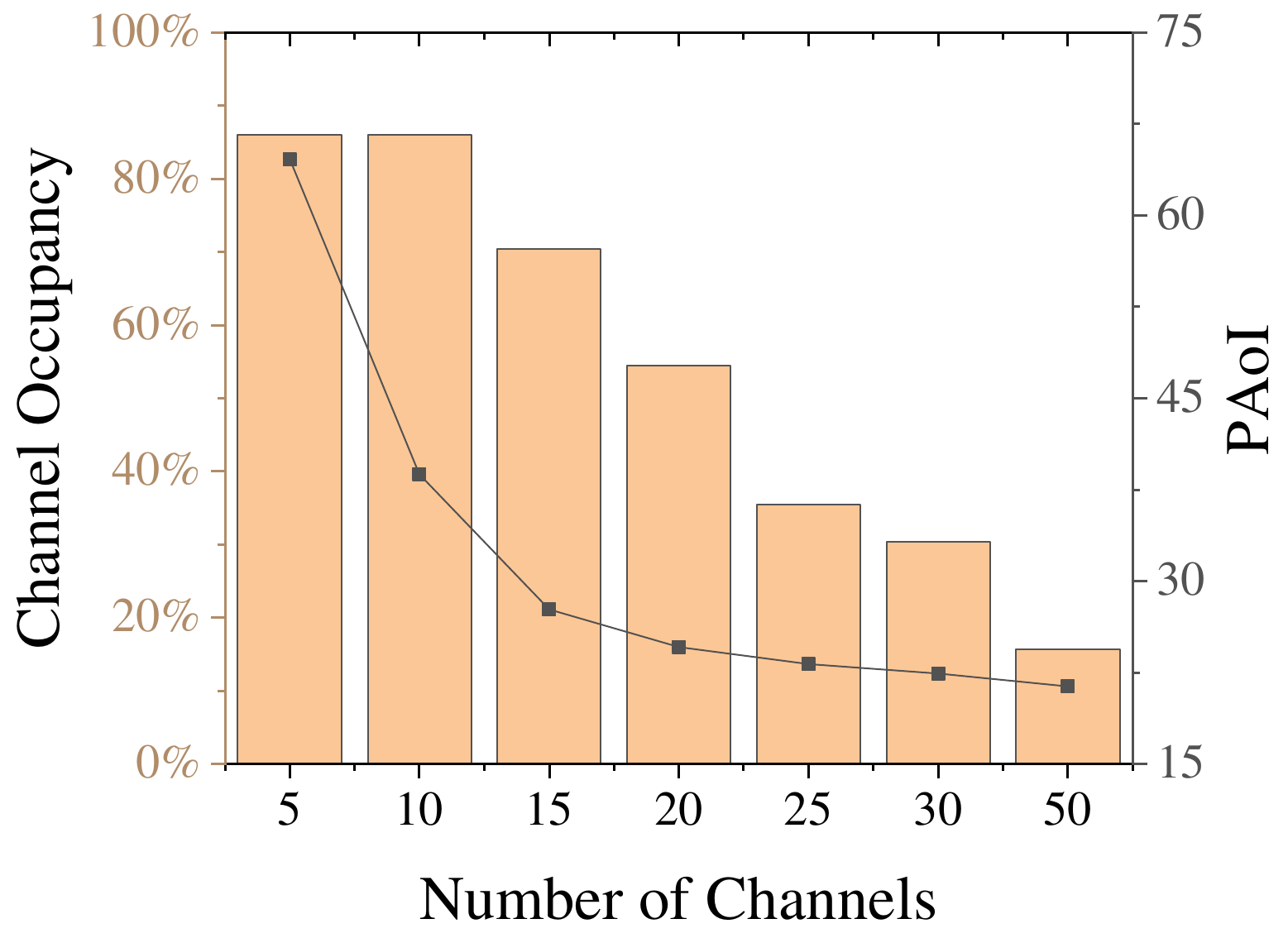}}
	\caption{PAoI \& Channel Occupancy vs. Number of Channels.}
    % \vspace{-4pt}
\end{figure}
\subsection{Performance and convergence of Algorithms in Channels}\label{subsec:Exp-Converge}
{
In this subsection, we empirically show the correlation between the algorithm performance and the number of channels, as well as the correlation between the algorithm convergence rate and the channel unreliability.}
% We further study the influence of the unreliable channel on the PAoI value and the convergence performance of different algorithms including (NFPA-NAC) and (NFPA-ANAC) in the unreliable channel.

{In order to investigate the potential impact of the number of channels on algorithm performance, we compared the impact of the number of channels on PAoI and average channel occupancy.
To this end, we run experiments in two different server processing rates scenarios, as seen in Fig.\ref{fig:Graph_highmu} and \ref{fig:Graph_lowmu}, respectively.
On the one hand, we notice a drop in PAoI as the number climbs, indicating that more tasks may be sent to the server across a more dependable and fast channel as overall channel capacity rises.  Particularly, in the scenario of high server processing rates (Fig.\ref{fig:Graph_highmu}), the server processing efficiency is high, resulting in a lower total PAoI value than in another scenario (Fig.\ref{fig:Graph_lowmu}).
On the other hand, as the overall channel capacity increases with the number of channels, the channel occupancy decreases.  This ensures that tasks can be transmitted to the server on a more reliable and faster path.  However, if the servers' service efficiency becomes lower than the channel transmission efficiency, there is an upper bound to the channel occupancy. 
This feature is shown in Fig.\ref{fig:Graph_lowmu} as the number of channels gets low and the reason why it does not reach 100\% is discussed in detail in Section \ref{sec:system_model}.
}

% As the number rises, the channel occupancy falls as the overall channel capacity rises, allowing for the delivery of more jobs to the server across a more dependable and fast channel.  Because the server processing efficiency is too poor in the low server processing rates scenario, the channel occupancy hits an upper bound when the number of channels is low.  In Sec.III.B. 1, the channel occupancy upper restriction is extensively discussed.
% The lack of communication resources, like unreliable channels, can prevent the development of more efficient distirbuted algorithms.
% Although, both NFPA-NAC in Fig.\ref{fig:NFPA-NAC} and NFPA-ANAC in Fig.\ref{fig:NFPA-ANAC} are advantageous for implementing powerful edge servers due to their ability to allocate channel and task scheduling resources efficiently. 
{The convergence rate of distributed algorithms can be hindered by a lack of reliable communication resources. However, both NFPA-NAC and NFPA-ANAC (as shown in Fig.\ref{fig:NFPA-NAC} and \ref{fig:NFPA-ANAC} respectively) are advantageous for the implementation of powerful edge servers. This is due to their ability to efficiently allocate channel and task-scheduling resources.
We show the performance of algorithms NFPA-NAC and NFPA-ANAC in different channel conditions ($p=0.3, 0.5, 0.7, 0.9$) in Fig.\ref{fig:NFPA-NAC} and Fig.\ref{fig:NFPA-ANAC} without repeated experiments.
We observe that the PAoI values gradually decline and eventually stabilize as the iterations proceed.
Moreover, the lower the value of channel condition $p$, the slower the convergence rate.
Additionally, it can be observed that PAoI can be more optimal with better channel conditions because users have more options for offloading.
We note that in our simulations the algorithm NFPA-NAC cannot successfully iterate in each iteration, while NFPA-ANAC can successfully iterate using the limited information in each iteration. }

\begin{figure}[ht!]
	\centering
	\subfigure[NFPA-NAC]{\label{fig:NFPA-NAC}\includegraphics[width=0.235\textwidth]{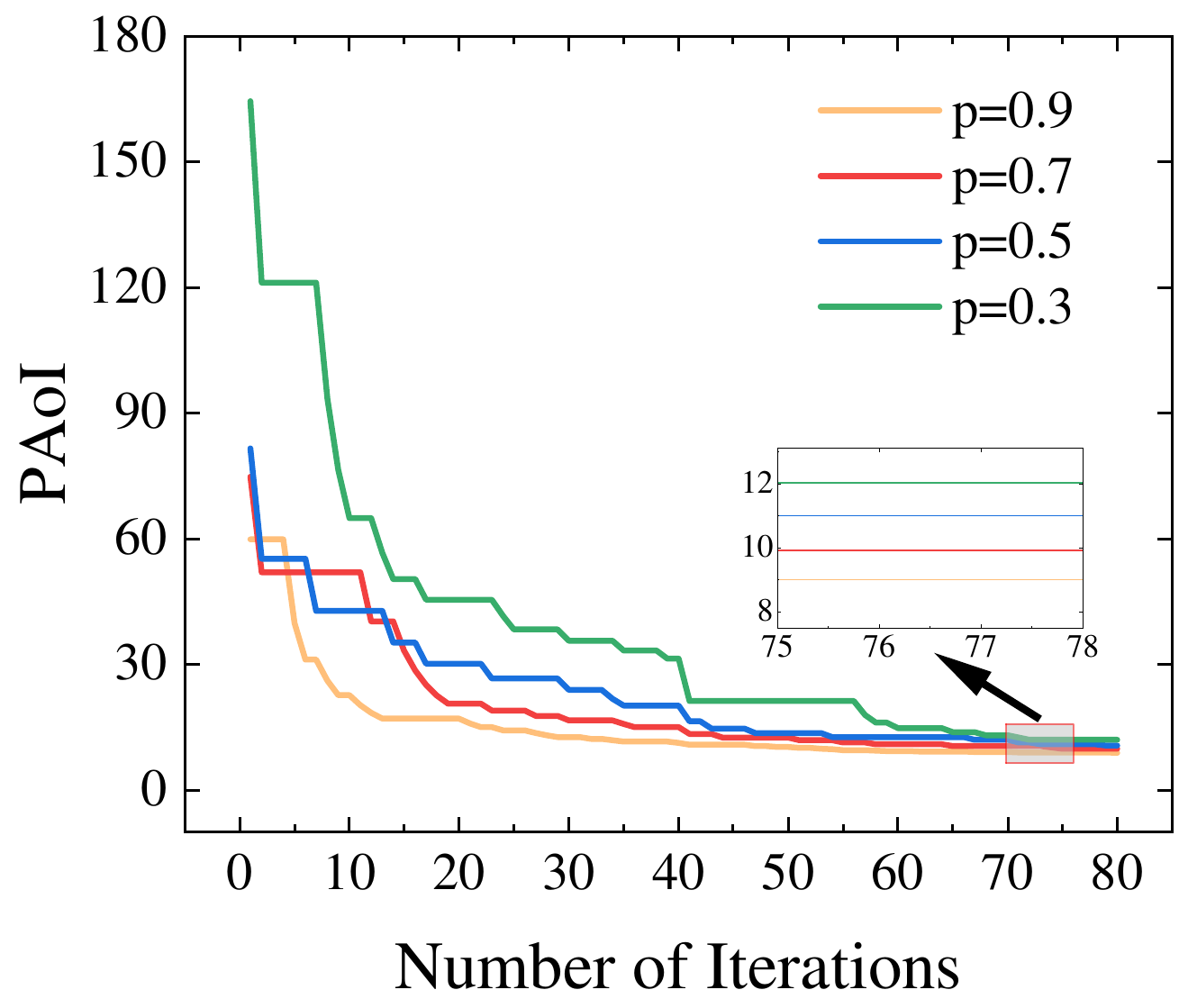}}
	\subfigure[NFPA-ANAC]{\label{fig:NFPA-ANAC}\includegraphics[width=0.235\textwidth]{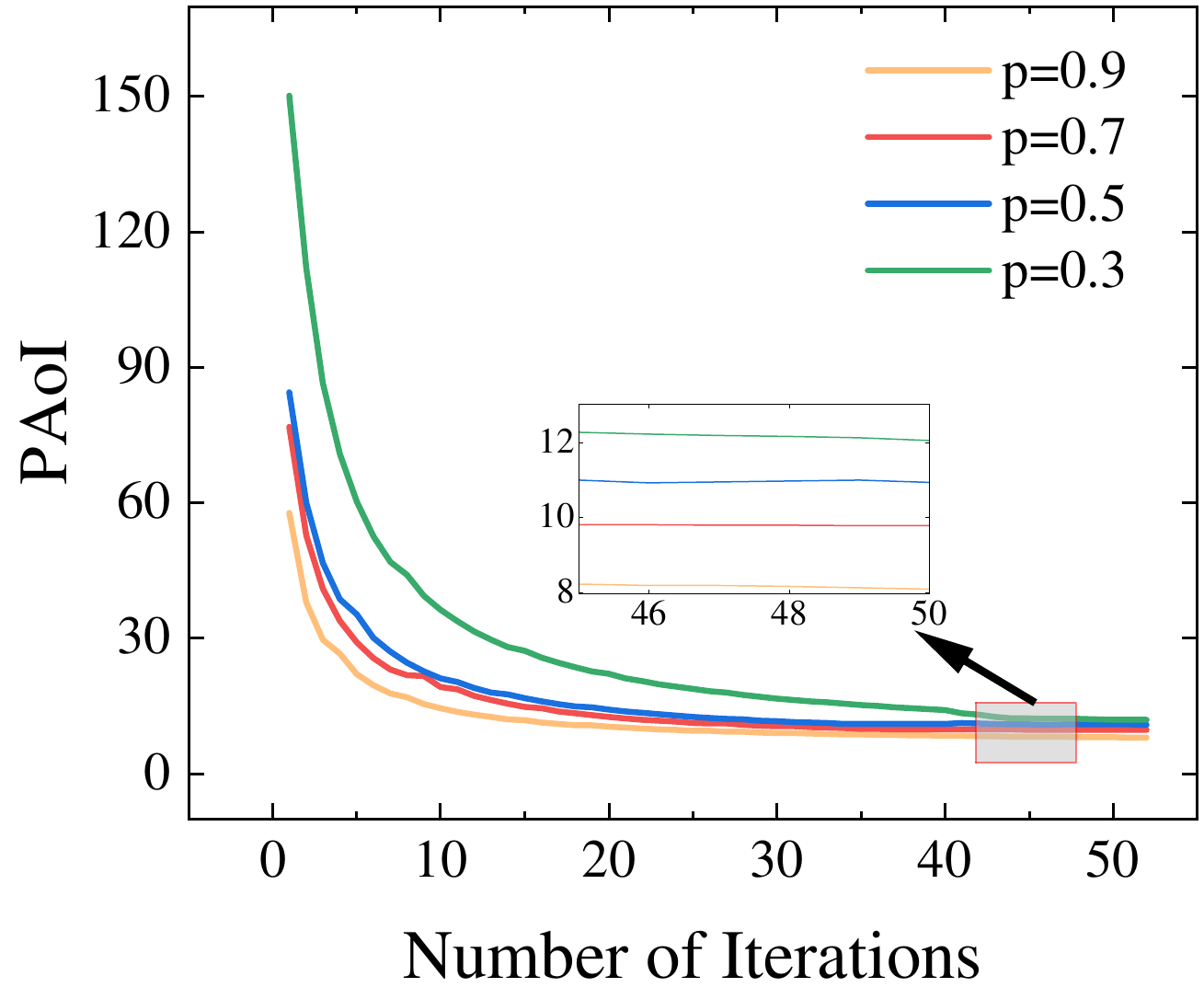}}
	\caption{The influence of the unreliable channel on the PAoI.}
	\label{fig:5}
\end{figure}

% To illustrate the difference between the synchronous and asynchronous parallel algorithms, NFPA-NAC and NFPA-ANAC more systematically,
% Fig.\ref{fig:compare} with various channel conditions ($p=0.7, 0.8, 0.9, 0.99$) compares the convergence of the two algorithms.
{To illustrate the differences between the synchronous and asynchronous parallel algorithms, specifically NFPA-NAC and NFPA-ANAC, we can refer to Fig.\ref{fig:compare}. This figure depicts the convergence of these two algorithms under various channel conditions ($p=0.7, 0.8, 0.9, 0.99$).
In contrast to Fig.\ref{fig:5}, we choose a better initial value, more users, and a large number of repeated experiments to help explain the convergence characteristics.
Both methods have similar rates of convergence in acceptable channels ($p=0.9, 0.99$).
However, the convergence rate advantage of algorithm NFPA-ANAC, however, is fairly apparent in the ($p=0.7, 0.8$) channel because of the poorer channel quality.
This observation verifies our discussion in Section \ref{subsec:Discussion-AsynchronousADMM}, as it clearly demonstrates that the asynchronous parallel algorithm (NFPA-ANAC) exhibits greater communication efficiency, particularly in unreliable channel conditions.}

\begin{figure}[ht!]
    \vspace{-0cm} 
    \setlength{\abovecaptionskip}{-0cm} 
    \setlength{\belowcaptionskip}{-0cm} 
    \centering
	\includegraphics[width=0.45\textwidth]{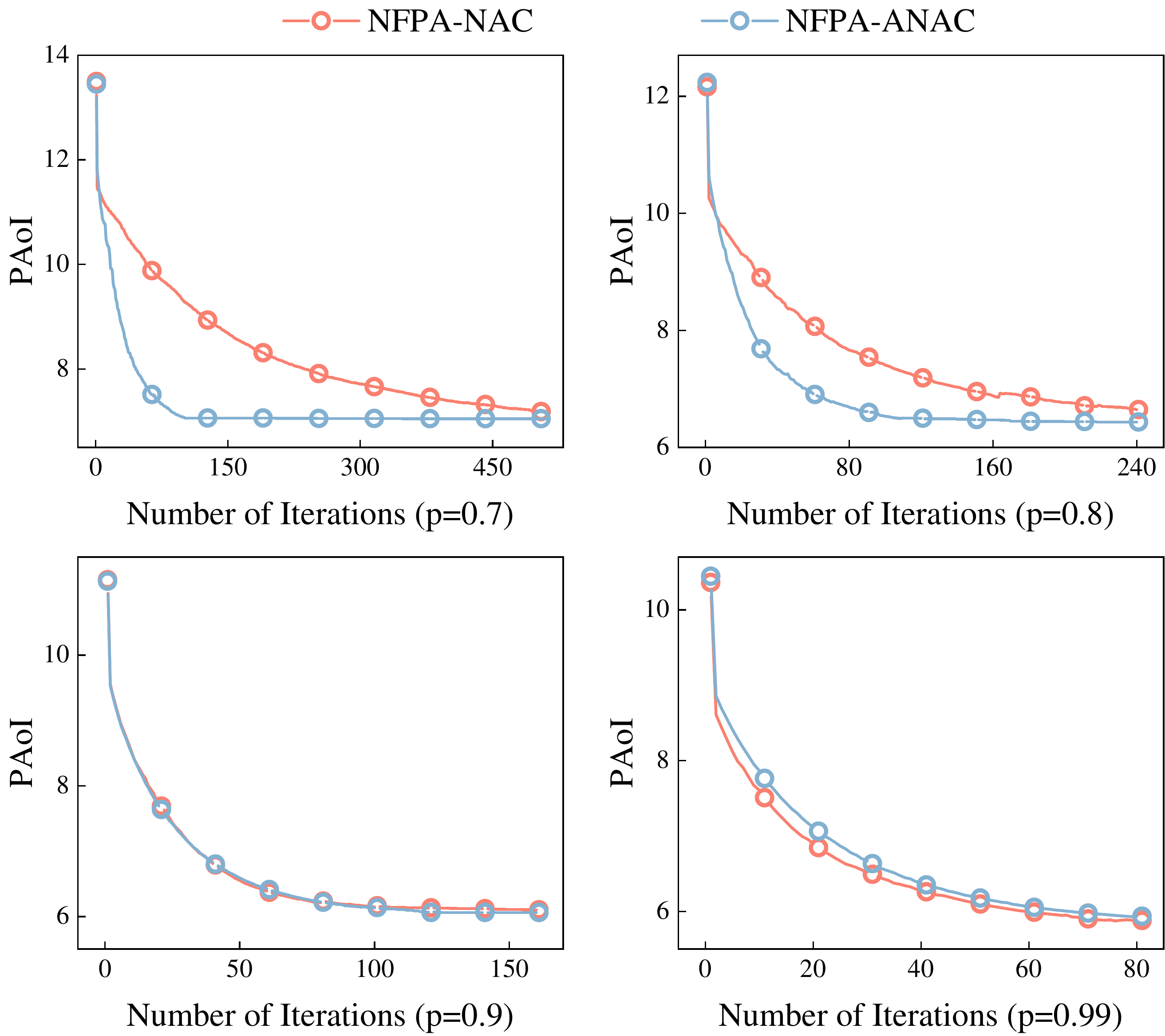}
    \vspace{-0pt}
	\caption{Comparison with different algorithms.}
	\label{fig:compare}
    % \vspace{-18pt}
\end{figure}

\section{Conclusion}\label{sec:conclusion}
In this paper, we proposed a task scheduling method that considered multi-priority users and multi-server collaboration to address the limitations of unreliable channels in current real-time systems and the individual needs of users. We derived the utility function of priority scheduling based on PAoI and designed a set of distributed optimization methods. Specifically, we first considered two simplified problems for the original problem and employed fractional programming as well as ADMM to obtain their solutions. Building upon these solutions and the conclusions drawn therein, we developed an iterative algorithm to solve the original problem. Furthermore, we proposed a distributed asynchronous approach with a sublinear rate of communication defects and discussed the theoretical performance improvement due to the multi-priority mechanism.  We implemented the method and conducted extensive simulations to compare it with the existing age-based scheduling strategies. Our results demonstrated the effectiveness and superiority of our method in addressing the requirements of freshness-sensitive users over unreliable channels.

% In this paper, we point out the limitations of unreliable channels in current real-time systems and introduce PoPeC, a scheduling method under multi-class priority and multi-server cooperation.
% We prove and derive the utility function of priority scheduling based on PAoI, and design a set of distributed optimization methods, that is, using a fractional programming transformation problem and ADMM response solution mechanism, to obtain the scheduling strategy. In order to reduce the loss caused by the unreliable channel iteration, we designed the asynchronous ADMM algorithm and proved that its convergence property is better than the conventional ADMM.
% We implemented PoPeC and conducted an extensive evaluation to show how PoPeC compares to other current age-based scheduling strategies.

% \twocolumn

% \clear
\normalem
\bibliographystyle{IEEEtran}
\bibliography{reference}

%\vskip -20pt plus -1fil

\begin{IEEEbiography}[{\includegraphics[width=1in,height=1.25in,clip,keepaspectratio]{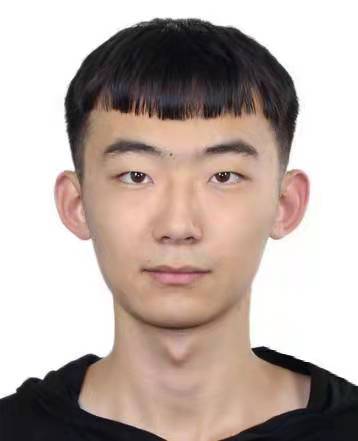}}]{Nan Qiao}
received the B.Sc. in Computer Science from Central South University, China. Since Sept. 2021, he has been pursuing the Ph.D. degree in Computer Science from Central South University, China. His research interests include wireless communications, distributed optimization, and Internet-of-Things.
\end{IEEEbiography}
\vskip -20pt plus -1fil

\begin{IEEEbiography}[{\includegraphics[width=1in,height=1.25in,clip,keepaspectratio]{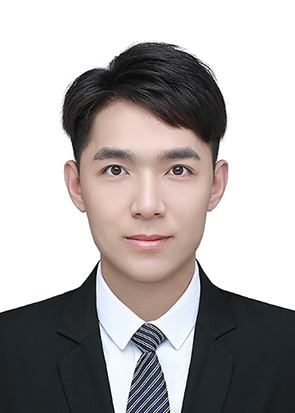}}]{Sheng Yue}
received his B.Sc. in mathematics (2017) and Ph.D. in computer science (2022), from Central South University, China. Currently, he is a postdoc with the Department of Computer Science and Technology, Tsinghua University, China. His research interests include network optimization, distributed learning, and reinforcement learning.
\end{IEEEbiography}
\vskip -20pt plus -1fil

\begin{IEEEbiography}[{\includegraphics[width=1in,height=1.25in,clip,keepaspectratio]{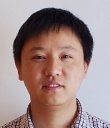}}]{Yongmin Zhang} 
(Senior Member, IEEE) received the PhD degree in control science and engineering in 2015, from Zhejiang University, Hangzhou, China. From 2015 to 2019, he was a postdoctoral research fellow in the Department of Electrical and Computer Engineering at the University of Victoria, BC, Canada. He is currently a Professor in the School of Computer Science and Engineering at the Central South University, Changsha, China. His research interests include IoTs, Smart Grid, and Mobile Computing. He won the best paper award of IEEE PIMRC 2012 and the IEEE Asia-Pacific (AP) outstanding paper award 2018.
\end{IEEEbiography}
\vskip -20pt plus -1fil

\begin{IEEEbiography}[{\includegraphics[width=1in,height=1.25in,clip,keepaspectratio]{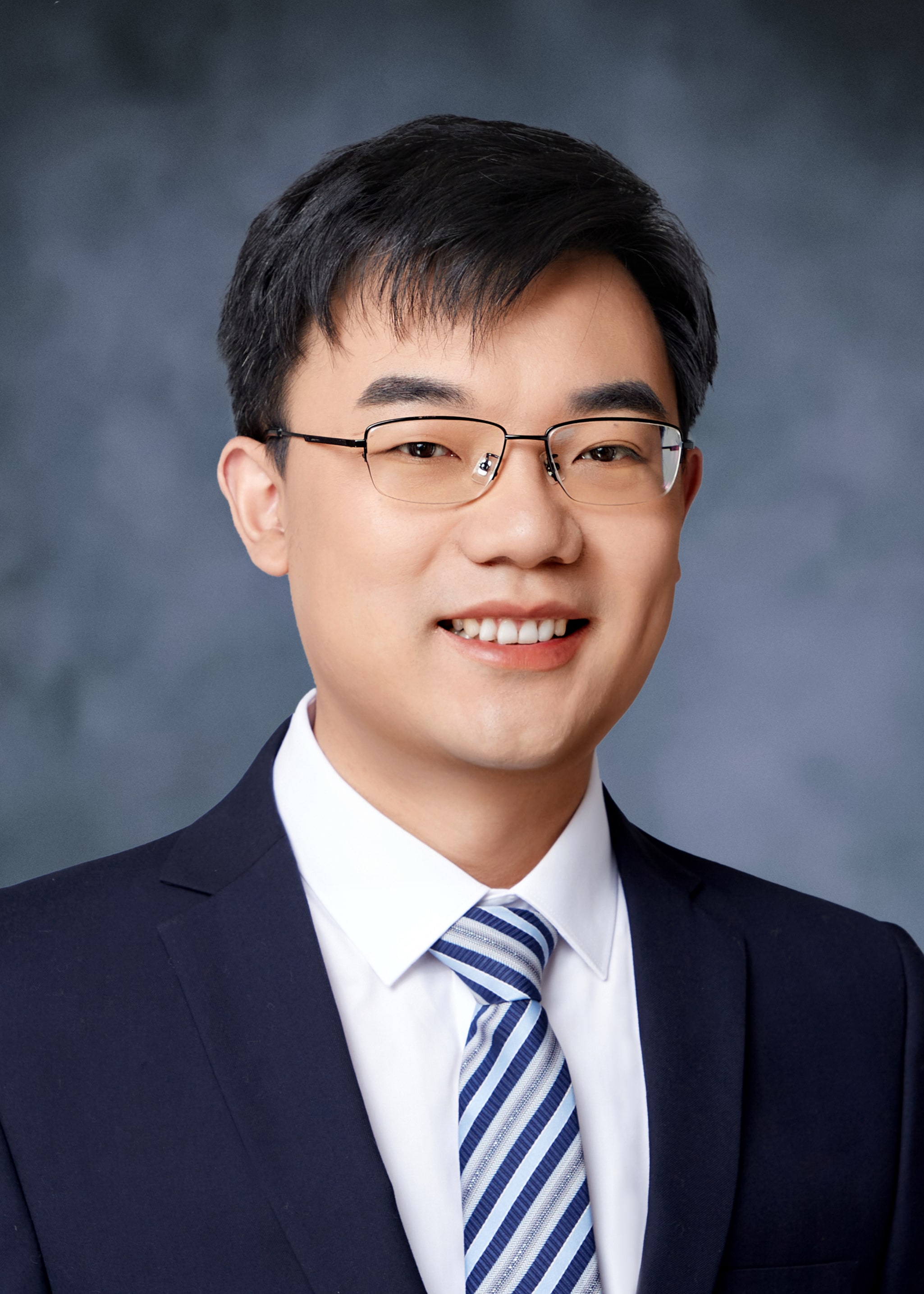}}]{Ju Ren}
Ju Ren [M’16, SM’21] (renju@tsinghua.edu.cn) received the B.Sc. (2009), M.Sc. (2012), Ph.D. (2016) degrees all in computer science, from Central South University, China. Currently, he is an associate professor with the Department of Computer Science and Technology, Tsinghua University, China. His research interests include Internet-of-Things, edge computing, edge intelligence, as well as security and privacy. 
    
He currently serves as an associate editor for many journals, including IEEE Transactions on Mobile Computing, IEEE Transactions on Cloud Computing and IEEE Transactions on Vehicular Technology, etc. He also served as the general co-chair for IEEE BigDataSE’20, the TPC co-chair for IEEE BigDataSE’19, the track co-chair for IEEE ICDCS’24, the poster co-chair for IEEE MASS’18, a symposium co-chair for IEEE/CIC ICCC’23\&19, I-SPAN’18 and IEEE VTC’17 Fall, etc. He received several best paper awards from IEEE flagship conferences, including IEEE ICC’19 and IEEE HPCC’19, etc., the IEEE TCSC Early Career Researcher Award (2019), and the IEEE ComSoc Asia-Pacific Best Young Researcher Award (2021). He was recognized as a highly cited researcher by Clarivate (2020-2022).

\end{IEEEbiography}
\vskip -20pt plus -1fil

% \end{document}

\clearpage
% \onecolumn
\begin{appendices}
% \section{A Proof of Lemma \ref{lemma:confidence_interval}}\label{app:lemma:confidence_interval}
% We prove Lemma \ref{lemma:confidence_interval} here.
% \begin{proof}
%     Since the number of packets offloaded to the channel obeys a Poisson distribution, it is not a stable value. But in order to guarantee that, at each slot t, there is a maximum capacity of the channel that is not smaller than the current moment, \ie $\chi_{c}(t) < M_c^\mathrm{max}$
%     According to the properties of the cumulative distribution function of the Poisson distribution~\cite{patil2012comparison}, we have $\sum_{n\in\mathcal{N}} \eta^u_{n,c}\lambda_{n} + z^2_{1}/2 + z_{2}\sqrt{\sum_{n\in\mathcal{N}} \eta^u_{n,c}\lambda_{n} + z^2_{2}/4} \le M_c^\mathrm{max}$. The proof is completed.
% \end{proof}

\section{}\label{app:pro:paoi_mg1_pri}
    % We prove Proposition \ref{pro:paoi_mg1_pri} here.

    The previous have investigated $\mathbb{E}[T^{p}_{n}]$, $\mathbb{E}[I^{p}_{n}]$ and $\mathbb{E}[Y^{p}_{n}]$ too much, 
    with $\mathbb{E}[Y^{p}_{n}]=1/\mu_n$ being the processing time, $\mathbb{E}[T^{p}_{n}]=T^\mathrm{tr}_{n,c}$  the maximum transmission time over all possible channels, 
    and $\mathbb{E}[I^{p}_{n}] = 1/\sum_{c\in\mathcal{C}}p_{n,c}\eta^u_{n,c}\lambda_{n}$ the arrival interval~\cite{yates2021age}.
    % remain the same values as \eqref{eq:paoi-expectation2}.
    Also, some studies discuss the waiting time $\mathbb{E}[W]$ in the condition multi-class M/G/1\cite{huang2015optimizing} or M/G/1 with priority\cite{xu2020peak}. However, we focus more on the value of $\mathbb{E}[W]$ in multi-class M/G/1 with priority, which means there are a number of different types of users in a unified priority in the M/G/1 system.
    According to Little's Law,
    \begin{align}
        \mathbb{E}[L_{i}] = \lambda_i \mathbb{E}[W_{i}].
    \end{align}
    For the highest priority users \ie $\delta=1$, it holds
    \begin{align}
        \mathbb{E}[W_{\delta=1}] 
        &= \sum_{i \in \mathcal{N}^{1}} \frac{\mathbb{E}[L_i]}{\mu_i} + \sum_{i \in \mathcal{N}} \rho_i \mathbb{E}[\nu^{2}_i]\nonumber\\
        &= \sum_{i \in \mathcal{N}^{1}} \rho_i \mathbb{E}[W_i] + \sum_{i \in \mathcal{N}} \rho_i \mathbb{E}[\nu^{2}_i]
    \end{align}
    where $\mathbb{E}[W_{\delta=1}] = \mathbb{E}[W_i], ~ \forall ~ \delta(i) = 1$. Accordingly, we obtain
    \begin{align}
    \label{eq:WaitingTime-HPU}
        \mathbb{E}[W_{\delta=1}] 
        = \frac{\sum_{i \in \mathcal{N}} \rho_i \mathbb{E}[\nu^{2}_i]}{1 - \sum_{i \in \mathcal{N}^{1}} \rho_i}
    \end{align}

     We separate $\mathbb{E}[W]$ into different components in order to break it down into smaller parts.  The first component consists of all high-priority or identically prioritized jobs that are in the queue when the current task arrives.  The waiting time for this component is $\mathbb{E}[S^1]$.  The second component, consisting of all high-priority users who arrived at the same time as the first component was executed, has a waiting time of $\mathbb{E}[S^2]$.  We continue to split the remaining portions in the same way.
    % We separate $\mathbb{E}[W]$ into several portions. The total number of high-priority or identically prioritized jobs in the queue at the moment the current task arrived makes up the first component, and its waiting time is $\mathbb{E}[S^1]$. The second portion has a waiting time of $\mathbb{E}[S^2]$ and is made up of all high-priority users who arrived at the same time as the first part was executed. The succeeding contents are also split in this manner.
    
    % We divide $\mathbb{E}[W]$ into several parts. The first part consists of the sum of all tasks of high priority or the same priority waiting in the queue at the time of arrival of the current task, and its waiting time is $\mathbb{E}[S^1]$. The second part consists of the sum of all users of high priority arriving at the time of execution of the first part, and its waiting time is $\mathbb{E}[S^2]$. The subsequent parts are divided according to similar way. Thus, we have
    \begin{align}
        \mathbb{E}[W_i] = \mathbb{E}[S^1_i + S^2_i + S^3_i + \cdots] = \sum^{\infty}_{k} \mathbb{E}[S^k_i]
    \end{align}
    
    Based on the above discussion, we derive the waiting time for the first part,
    \begin{align}
        \mathbb{E}[S^1_i] = \sum_{p=1}^{\delta(i)} \sum_{j  \in \mathcal{N}^{p}} \rho_j \mathbb{E}[W_j] + \sum_{j \in \mathcal{N}} \rho_j \mathbb{E}[\nu^{2}_j]
    \end{align}
    The waiting time for the other part is as follows,
    \begin{align}
        E[S^{k+1}_i] &=\int_{s=0}^{\infty} E[S^{k+1}_i \mid S^{k}_i=s] f_{k}(s) d s \nonumber\\
        &=\int_{s=0}^{\infty} \sum_{p=1}^{\delta(i)-1} \sum_{j  \in \mathcal{N}^{p}}  \rho_j s  f_{k}(s) d s \nonumber\\
        &=\left(\sum_{p=1}^{\delta(i)-1} \sum_{j  \in \mathcal{N}^{p}}  \rho_j\right) E[S^{k}_i] \nonumber\\
        &=\left(\sum_{p=1}^{\delta(i)-1} \sum_{j  \in \mathcal{N}^{p}}  \rho_j\right)^k E[S^{1}_i]
    \end{align}
    Thus, we derive
    \begin{align}
        \mathbb{E}[W_i] &= \frac{\mathbb{E}[S^1_i]}{1 - \sum_{p=1}^{\delta(i)-1} \sum_{j  \in \mathcal{N}^{p}}  \rho_j}\nonumber\\ 
        &= \frac{\sum_{p=1}^{\delta(i)} \sum_{j  \in \mathcal{N}^{p}} \rho_j \mathbb{E}[W_j] + \sum_{j \in \mathcal{N}} \rho_j \mathbb{E}[\nu^{2}_j]}{1 - \sum_{p=1}^{\delta(i)-1} \sum_{j  \in \mathcal{N}^{p}}  \rho_j}
    \end{align}
    which can be transformed into
    \begin{align}
        \mathbb{E}[W_i](1 - \sum_{p=1}^{\delta(i)-1} \sum_{j \in \mathcal{N}^{p}}  \rho_j)\nonumber\\
        = \sum_{p=1}^{\delta(i)} \sum_{j  \in \mathcal{N}^{p}} \rho_j \mathbb{E}[W_j] + \sum_{j \in \mathcal{N}} \rho_j \mathbb{E}[\nu^{2}_j]
    \end{align}
    and
    \begin{align}
        \mathbb{E}[W_{\delta=\delta(i)}](1 - \sum_{p=1}^{\delta(i)} \sum_{j  \in \mathcal{N}^{p}}  \rho_j)\nonumber\\
        = \sum_{p=1}^{\delta(i)-1} \sum_{j  \in \mathcal{N}^{p}} \rho_j \mathbb{E}[W_j] + \sum_{j \in \mathcal{N}} \rho_j \mathbb{E}[\nu^{2}_j]\nonumber\\
        = \mathbb{E}[W_{\delta=\delta(i)-1}](1 - \sum_{p=1}^{\delta(i)-2} \sum_{j  \in \mathcal{N}^{p}}  \rho_j)
    \end{align}
    
    In other words, it holds,
    \begin{align}
        &\mathbb{E}[W_{\delta=\delta(i)}](1 - \sum_{p=1}^{\delta(i)} \sum_{j  \in \mathcal{N}^{p}}  \rho_j)(1 - \sum_{p=1}^{\delta(i)-1} \sum_{j  \in \mathcal{N}^{p}}  \rho_j)\\
        =& \mathbb{E}[W_{\delta=\delta(i)-1}](1 - \sum_{p=1}^{\delta(i)-1} \sum_{j  \in \mathcal{N}^{p}}  \rho_j)(1 - \sum_{p=1}^{\delta(i)-2} \sum_{j  \in \mathcal{N}^{p}}  \rho_j)\nonumber
    \end{align}
    
    Combine with \eqref{eq:WaitingTime-HPU}, we obtain
    \begin{align}
        &\mathbb{E}[W_{\delta=\delta(i)}]
        = \frac{\mathbb{E}[W_{\delta=1}](1 - \sum_{j \in \mathcal{N}^{1}}  \rho_j)}{(1 - \sum_{p=1}^{\delta(i)} \sum_{j  \in \mathcal{N}^{p}}  \rho_j)(1 - \sum_{p=1}^{\delta(i)-1} \sum_{j  \in \mathcal{N}^{p}}  \rho_j)}\nonumber\\
        &= \frac{\sum_{i \in \mathcal{N}} \rho_j \mathbb{E}[\nu^{2}_j]}{(1 - \sum_{p=1}^{\delta(i)} \sum_{j  \in \mathcal{N}^{p}}  \rho_j)(1 - \sum_{p=1}^{\delta(i)-1} \sum_{j  \in \mathcal{N}^{p}}  \rho_j)}
    \end{align}
    which merges the value of $\mathbb{E}[T^{p}_{n}]$, $\mathbb{E}[I^{p}_{n}]$ and $\mathbb{E}[Y^{p}_{n}]$ to complete the proof.

\section{A Proof of Lemma \ref{lemma:P1-1_convex_problem}}\label{app:lemma:P1-1_convex_problem}
    First, we should prove that \eqref{eq:c7a} and \eqref{eq:c8a} are not convex sets. 
    Take \eqref{eq:c7a} as an example, 
    if it was a convex set, we should have obtained 
    \begin{equation}
        (1-r)h(\bm\eta^u_{c_1}) + r h(\bm\eta^u_{c_2}) \mathop{\ge}^{(a)} h((1-r)\bm\eta^u_{c_1} + r\bm\eta^u_{c_2}),
    \end{equation}
    where $h(\bm\eta^u_c) = \langle\bm\lambda, \bm\eta^u_c\rangle + z_{2}\sqrt{\langle\bm\lambda, \bm\eta^u_c\rangle + z^2_{2}/4} + z^2_{1}/2 - M_c^\mathrm{max}$, $\bm\lambda := \{\lambda_{n}\}_{n\in\mathcal{N}}$, $\bm\eta^u_c := \{\eta^u_{n,c}|\}_{n\in\mathcal{N}}$, $\forall r\in[0,1]$, 
    $\{c1, c2 \in \mathcal{C}\}$,
    $h(\bm\eta^u_{c_1})\leq 0$ and $h(\bm\eta^u_{c_2})\leq 0$.
    Moreover, inequality (a) is equivalent to
    \begin{equation}
    \label{eq:lemma-(a)}
    r (r-1) (\sqrt{\langle\bm\lambda, \bm\eta^u_{c_1}\rangle} + \sqrt{\langle\bm\lambda, \bm\eta^u_{c_2}\rangle})^2 \ge 2 r (1-r) z^2_{2}/4.
    \end{equation}
    However, if $r \ne 0~or~1$, inequality \eqref{eq:lemma-(a)} does not hold since the right term is higher than zero and the left term of the inequality is less than zero.
    Hence, \eqref{eq:c7a} is a nonconvex set.
     A comparable procedure in \eqref{eq:c8a} can lead to the same conclusion.

    % with $r\in[0,1]$, $\forall \eta^a_{n,c}, \eta^b_{n,c}\in\bm \eta^u$, $a = \sum_{n\in\mathcal{N}} \eta^a_{n,c}\lambda_{n}$ and $b = \sum_{n\in\mathcal{N}} \eta^b_{n,c}\lambda_{n}$, we obtain
    % if $\bm x^a$ and $\bm x^b$ satisfy. Let $\eta^a_{n,c}\lambda_{n}\in\bm x^a$ and $\eta^b_{n,c}\lambda_{n}\in\bm x^b$. We have 
    % \begin{equation}
    % \begin{split}
    % (1-r)a + (1-r)\mathcal{Z}_{\alpha}\sqrt{a+c} 
    %  + rb + r\mathcal{Z}_{\alpha}\sqrt{b+c} \\
    % \mathop{\leq}^{(a)} \left((1-r)a + rb\right)
    % + \mathcal{Z}_{\alpha}\sqrt{\left((1-r)a 
    % + rb + c\right)}\\
    % (1-r)\sqrt{a+c}+ r\sqrt{b+c} \leq  \sqrt{\left((1-r)a + rb + c\right)}\\
    % \left( (1-r)a + rb\right) - \left((1-r)^2 a + r^2 b + 2r(1-r) \sqrt{a b}\right)\\
    % \le c-(1-r)^2 c - r^2c,
    % \end{split}
    % \end{equation}
    Next, we will demonstrate that \eqref{eq:c7a} and \eqref{eq:c8a} are comparable to \eqref{eq:c7b} and \eqref{eq:c8b} and are convex sets.
    Again, let's take the example of \eqref{eq:c7a}, which is equivalent to
    \begin{align}
        % \begin{split}
            &\langle\bm\lambda, \bm\eta^u_{c}\rangle + z^2_{1}/2 + z_{2}\sqrt{\langle\bm\lambda, \bm\eta^u_{c}\rangle + z^2_{2}/4} \le M_c^\mathrm{max},\nonumber\\
            &(z_{2}/2 + \sqrt{\langle\bm\lambda, \bm\eta^u_{c}\rangle + z^2_{2}/4})^2 
            \le M_c^\mathrm{max} - z^2_{1}/2  + z^2_{2}/4,\nonumber\\
            &\langle\bm\lambda, \bm\eta^u_{c}\rangle \le M_c^\mathrm{max} + \frac{z^2_{2}}{2} - \frac{z^2_{1}}{2} - z_{2}\sqrt{M_c^\mathrm{max} + \frac{z^2_{2}}{2} - \frac{z^2_{1}}{2}},
        % \end{split}
    \end{align}
    which means that the constraint \eqref{eq:c7a} is transformed mathematically into \eqref{eq:c7b}, which is an affine set and a kind of convex set.
    \eqref{eq:c8b} can be obtained in a similar way.
    
    % \begin{align}
    % \textbf{C7-b}&~~
    % \label{eq:constraint-channelcapacity}
    %     \sum_{n\in\mathcal{N}} \eta^u_{n,c}\lambda_{n} \le M_c^\mathrm{max} + \frac{z^2_{2}}{2} - \frac{z^2_{1}}{2} - z_{2}\sqrt{M_c^\mathrm{max} + \frac{z^2_{2}}{2} - \frac{z^2_{1}}{2}}, \quad \forall c \\
    % \textbf{C8-b}&~~
    % \label{eq:constraint-calculatedcapacity}
    %     \sum_{n\in\mathcal{N}} \sum_{c\in\mathcal{C}} p_{n,c} \frac{\eta^u_{n,c}\lambda_{n}}{\mu} \leq 1 + \frac{z^2_{4}}{2} - \frac{z^2_{3}}{2} - z_{4}\sqrt{1 + \frac{z^2_{4}}{2} - \frac{z^2_{3}}{2}}
    % \end{align}

\section{A Proof of Theorem \ref{thm:admm-consensus}}\label{app:lemma:paoi_mg1}
% \noindent
% \textbf{Proof of Convex Function $F_n$}\\
    The objective function that needs to be proven convexity is
    \begin{equation}
        f_n(\hat t^\mathrm{tr},\bm\eta^u) = F^T_n(\hat t^\mathrm{tr}_n,\bm\eta^u) + F^I_n(\hat t^\mathrm{tr}_n,\bm\eta^u)+F^W_n(\hat t^\mathrm{tr}_n,\bm\eta^u)+F^Y_n(\hat t^\mathrm{tr}_n,\bm\eta^u),\nonumber
    \end{equation}
    where
        $F^I_n(\hat t^\mathrm{tr}_n,\bm\eta^u)=F^I_n(\bm\eta^u) = \frac{1}{\sum_{c\in\mathcal{C}}p_{n,c}\eta^u_{n,c}\lambda_{n}}$, 
        $F^T_n(\hat t^\mathrm{tr}_n,\bm\eta^u) = \hat t^\mathrm{tr}_n$,
        $F^Y_n(\hat t^\mathrm{tr}_n,\bm\eta^u) = \frac{1}{\mu}$,
        and
        $F^W_n(\hat t^\mathrm{tr}_n,\bm\eta^u)=F^W_n(\bm\eta^u) = (\frac{1}{2}\sum_{n'=1}^{\mathcal{N}}\sum_{c\in\mathcal{C}}p_{n',c}\eta^u_{n',c}\lambda_{n'}\nu)/(1-\sum_{n'=1}^{\mathcal{N}}\sum_{c\in\mathcal{C}}p_{n',c}\frac{\eta^u_{n',c}\lambda_{n'}}{\mu})$.
    \begin{enumerate}
        \item Due to properties of linear functions, it is obvious that \underline{$\nabla^{2} F^T_n(\hat t^\mathrm{tr}_n,\bm\eta^u) \succeq 0$}.

        \item Due to properties of constant functions, it holds that \underline{$\nabla^{2} F^Y_n(\hat t^\mathrm{tr}_n,\bm\eta^u) \succeq 0$}.
        
        \item Furthermore, we have $\{n, m, m_1, m_2 \in \mathcal{N} | n = m \& n \neq m_1 \& n \neq m_2\}$ , $\{c, d, d_1, d_2 \cdots d_\mathrm{max} \in \mathcal{C}\}$, $\boldsymbol p_n = \{p_{n,1} \cdots p_{n,C}\}$, and $\boldsymbol p = \{\boldsymbol p_n|n\in \mathcal{N}\}$.
        % We obtain some first-order derivatives.
        % \begin{align}
        % \label{eq:paoi-function-arrival-firstorderderivative}
        %     &\frac{\partial F^I_n(\bm\eta^u)}{\partial \lambda_{m,c}} = - \frac{p_{n,c}}{(\sum_{c\in\mathcal{C}}p_{n,c}\eta^u_{n,c}\lambda_{n})^2},
        %     &\frac{\partial F^I_n(\bm\eta^u)}{\partial \lambda_{m_1,c}} = 0.
        % \end{align}
        We list some key second-order derivatives.
        \begin{align}
        % \label{eq:paoi-function-arrival-secondorderderivative}
            \label{eq:sod-1}
            \frac{\partial^2 F^I_n(\bm\eta^u)}{\partial \lambda^2_{m,c}} &= \frac{2p^2_{m,c}}{(\sum_{c\in\mathcal{C}}p_{n,c}\eta^u_{n,c}\lambda_{n})^3},\\
            \label{eq:sod-2}
            \frac{\partial^2 F^I_n(\bm\eta^u)}{\partial \lambda_{m,d_1} \partial \lambda_{m,d_2}} &= \frac{2p_{n,d_1}p_{n,d_2}}{(\sum_{c\in\mathcal{C}}p_{n,c}\eta^u_{n,c}\lambda_{n})^3},\\
            \label{eq:sod-3}
            \frac{\partial^2 F^I_n(\bm\eta^u)}{\partial \lambda_{m,c} \partial \lambda_{m_1,c}} &= 
            \frac{\partial^2 F^I_n(\bm\eta^u)}{\partial \lambda_{m_1,c} \partial \lambda_{m_2,c}} = 
            \frac{\partial^2 F^I_n(\bm\eta^u)}{\partial \lambda^2_{n_1,c}} = 0.
        \end{align}
        Combining Eq.\eqref{eq:sod-1}, Eq.\eqref{eq:sod-2}, and Eq.\eqref{eq:sod-3}, the Hessian matrix of $F^I_n(\bm\eta^u)$ is described as
        \begin{align}
        \label{eq:paoi-function-arrival-hessian}
            \mathbf{H}^I_n(\bm\eta^u)           
           =\left[\begin{array}{c|c}
            \boldsymbol A & \boldsymbol B\\
            \hline
            \boldsymbol C & \boldsymbol D
           \end{array}\right]
        \end{align}
        where $\boldsymbol D = \boldsymbol{O}^{(N-1)C}$, $\boldsymbol B = \boldsymbol C^T = \boldsymbol{O}^{C \times (N-1)C}$ and 
        $\boldsymbol A=\left[\begin{array}{ccc}
            \frac{\partial^2 F^I_n(\bm\eta^u)}{\partial \lambda^2_{m,d_1}} & \cdots & \frac{\partial^2 F^I_n(\bm\eta^u)}{\partial\lambda_{m,d_1} \partial\lambda_{m,d_\mathrm{max}}}\\
            \vdots & \ddots & \vdots\\
            \frac{\partial^2 F^I_n(\bm\eta^u)}{\partial\lambda_{m,d_\mathrm{max}} \partial\lambda_{m,d_1}} & \cdots & \frac{\partial^2 F^I_n(\bm\eta^u)}{\partial\lambda^2_{m,d_\mathrm{max}}}
          \end{array}\right]^{C \times C}
         = \frac{\boldsymbol p^T_n \times \boldsymbol p_n}{(\sum_{c\in\mathcal{C}}p_{n,c}\eta^u_{n,c}\lambda_{n})^3}
          $ 
        denotes second-order derivative matrix at $m = n$, \ie $\boldsymbol A \succeq 0$.
        The leading Principle Submatrix of $\mathbf{H}^I_n(\bm\eta^u)$ are all greater than 0, \ie \underline{$\nabla^{2} F^I_n(\hat t^\mathrm{tr},\bm\eta^u) \succeq 0$}.
        
        \item As for $F^W_n(\bm\eta^u)$, it holds
        \begin{align}
	        \frac{\partial^2 F^W_n(\bm\eta^u)}{\partial \lambda_{m_1,d_1}\partial \lambda_{m_2,d_2}} = \boldsymbol p_{m_1,d_1} \boldsymbol p_{m_2,d_2}\frac{\nu/\mu}{(1-\sum_{m=1}^{\mathcal{N}}\sum_{c\in\mathcal{C}}p_{m,c}\frac{\lambda_{m,c}}{\mu})^3}
	    \end{align}
        The hessian matrix is described as 
        \begin{align}
        \nabla^{2}F^W_n(\bm\eta^u)=\boldsymbol p^T \times \boldsymbol p \frac{\nu/\mu}{(1-\sum_{m=1}^{\mathcal{N}}\sum_{c\in\mathcal{C}}p_{m,c}\frac{\lambda_{m,c}}{\mu})^3}
        \end{align}
        Thus, we derive \underline{$\nabla^{2} F^W_n(\hat t^\mathrm{tr},\bm\eta^u) \succeq 0$}.
        \end{enumerate}
        
        Above all, $\nabla^{2} f_n(\hat t^\mathrm{tr},\bm\eta^u) \succeq 0$. The proof is completed.
\section{Details of Optimal Solution to \textbf{P1}}\label{app:det:admm-consensus}
{
In this section, we delve into the intricacies of the optimal solution for problem \textbf{P1} in detail. We elucidate the Algorithm AC, tailored for solving \textbf{P1}, and expound on its comprehensive procedure in Appendix \ref{app:det:admm-consensus-A}. The convergence properties and proofs pertaining to Algorithm AC are methodically analyzed in Appendix \ref{app:det:admm-consensus-B}. Furthermore, we expand upon Algorithm AC in Appendix \ref{app:det:admm-consensus-C} by architecting its asynchronous variant, thereby enhancing its communication efficacy in the context of unreliable channels. This enhancement paves the way for deeming the algorithm as adept in managing communication challenges within such environments.}

{
\subsection{Algorithm AC for the solution to \textbf{P1}}\label{app:det:admm-consensus-A}
According to Problem \textbf{P1-3}, the augmented Lagrangian for Problem \textbf{P1-3} can be expressed as:
\begin{align}
% \label{eq:PLadmm-lagrange}
    L_{\rho}(\{\bm x_n\}, \bm x_{o}, \{\bm\sigma_n\})
   &= \sum_{n\in\mathcal{N}}(g_n(\bm x_n) + \langle\bm\sigma_n, \bm x_n - \bm x_{o}\rangle\nonumber\\
   &+ (\rho/2)\|\bm x_n - \bm x_{o}\|^2_2,
\end{align}
where $\bm\sigma_n := \{\sigma^n_{m,c}|c\in \mathcal{C}, m\in \mathcal{N}\}$ denotes the Lagrangian multipliers \textbf{P1-3} and $\rho$ is a positive penalty parameter.
Based on the analysis of Theorem 1 and Section IV-A2, the ADMM-Consensus based offloading algorithm is summarized as Algorithm \ref{alg:admm}. The procedure is detailed subsequently:
\paragraph{User side}  In iteration $t$ of the loop, given $\bm x_n$, the primal and dual variables are updated according to Eq.~\eqref{eq:PLadmm1-relax-i} and Eq.~\eqref{eq:PLadmm1-relax-sigma}, respectively. Then, each user sends the $\bm x^{t+1}_n$ and $\bm\sigma^{t+1}_n$ to the local server.
\begin{align}
    \label{eq:PLadmm1-relax-i} 
    \bm x^{t+1}_n & = \arg\min_{\bm x_n} (g_n(\bm x_n) + \langle\bm\sigma^t, \bm x_n - \bm x_{o}\rangle\nonumber\\
    &+ (\rho/2)\|\bm x_n - \bm x^t_{o}\|^2_2), \\ 
    \label{eq:PLadmm1-relax-sigma} 
    \bm\sigma^{t+1}_n & = \bm\sigma^{t}_n + \rho(\bm x^{t+1}_n - \bm x^{t}_{o}).
\end{align}
\paragraph{Server side} 
% The server collects all available iterations, updates the value of $\bm x^{t+1}_{o}$, and passes the latest information back to the user.
The server aggregates all the iterations it receives, updates the value of $\bm x^{t+1}_{o}$ according to Eq.~\eqref{eq:PLadmm1-relax-o}, and circulates the updated information back to the users.
\begin{align}
    \label{eq:PLadmm1-relax-o} 
    \bm x^{t+1}_{o} & = \frac{1}{N}\sum_{n\in\mathcal{N}}(\bm x^{t+1}_n + (\frac{1}{\rho}\bm\sigma^{t+1}_n)).
\end{align}
\paragraph{Termination criteria}
The iterative process is terminated when primal residual \(\|\bm x^{t+1}_n - \bm x^{t+1}_{o}\|^2_2\) exceeds a predefined threshold \(\epsilon^\mathit{pr}\), or when dual residual \(\rho\|\bm x^{t+1}_{o} - \bm x^{t}_{o}\|^2_2\) surpasses the threshold \(\epsilon^\mathrm{dl}\). 
% The iteration stops when $\|\bm x^{t+1}_n - \bm x^{t+1}_{o}\|^2_2 > \epsilon^\mathit{pr}$ or $\rho\|(\bm x^{t+1}_{o}) - \bm x^{t}_{o})\|^2_2 > \epsilon^\mathrm{dl}$.
% \end{enumerate}
\begin{algorithm}[ht]
	\caption{ADMM-Consensus (AC)}
        \label{alg:admm}
	% \label{AOHC}
	\LinesNumbered
	\For{$n = 1$ \KwTo $N$}{
	\KwIn{$\epsilon^\mathit{pr}$, $\epsilon^\mathrm{dl}$, $\theta$, $\bm x_n$, $t$}
    \While{$\|\bm x^{t+1}_n - \bm x^{t+1}_{o}\|^2_2 > \epsilon^\mathit{pr}$ or $\rho\|(\bm x^{t+1}_{o}) - \bm x^{t}_{o})\|^2_2 > \epsilon^\mathrm{dl}$}{
    	{Calculate $\bm x^{t+1}_n$, simultaneously according to \eqref{eq:PLadmm1-relax-i}}\\
        {Calculate $\sigma^{t+1}_n$, according to \eqref{eq:PLadmm1-relax-sigma}}\\
        {Calculate $\bm x^{t+1}_{o}$, according to \eqref{eq:PLadmm1-relax-o}}\\
        {Update $t = t + 1$}\\
	}
	{Set $\bm x^{*}_n = \bm x^{t+1}_n$}\\
	\KwOut{$\bm x^{*}_n$}
}	
% \label{alg:admm}
\end{algorithm}
\subsection{The performance analysis of Algorithm AC}\label{app:det:admm-consensus-B}
In this subsection, we further analyze the convergence and computational cost of the Algorithm AC.
According to Theorem 1, the objective function of Problem \textbf{P1-3} is defined as a closed, proper, and convex function. Its associated domain is also a well-defined closed, non-empty convex set. Moreover, the Lagrangian $L_n\left(\left\{\bm x_{n}^{t+1}\right\}, \bm x_{o}^{t+1} ; y^{t+1}\right)$ is endowed with a saddle point. Consequently, based on \cite{boyd2011distributed}, the iteration is demonstrated to satisfy three types of convergence: dual, consensus, and objective function. These convergences are explicated as follows:
\begin{itemize}
        \item \textit{Dual variable convergence.} We have $\sigma^{t}_n \rightarrow \sigma^{*}_n$ as $t \rightarrow \infty$, where $\sigma^{*}_n$ is a dual optimal vector.
        \item \textit{Consensus convergence.} The consensus constraint is satisfied eventually $\lim_{t \rightarrow \infty}\left\|\bm x_{n}^{t+1}-\bm x_{o}^{t+1}\right\|=0, \forall n$.
        \item \textit{Objective function convergence.} The objective function ultimately stabilizes at the optimal value:
    \begin{align}
        &\sum_{n\in\mathcal{N}} g_n(\bm x^t) \rightarrow \sum_{n\in\mathcal{N}} g_n(\bm x^*_n),~ t \rightarrow \infty
    \end{align}
    % according to~\cite{hong2016convergence}.
    % which indicates that the objective function eventually converges to the optimal value.
    % Thus, $L_n(x;y)$ will decrease after each dual.
    % Furthermore, according to property 3), we obtain $L^{p}_n\left(\left\{\bm x_{n}^{t+1}\right\}, \bm x_{o}^{t+1} ; y^{t+1}\right)$ is limited.
    \end{itemize}
Within each iteration of the Algorithm \ref{alg:admm}, Eq.~\eqref{eq:PLadmm1-relax-i}, Eq.~\eqref{eq:PLadmm1-relax-sigma}, and Eq.~\eqref{eq:PLadmm1-relax-o} are updated in a sequential manner to get $\bm x^t$.
% corresponding to a given $g_n$.
Accordingly, each individual subproblem can be efficiently tackled using prevalent optimization techniques such as gradient descent or interior point methods with low computational costs.
\subsection{Detail of Asynchronous Solution to \textbf{P1-3}}\label{app:det:admm-consensus-C}
In the preceding subsections, we introduced a suite of synchronous parallel algorithms. However, the robustness and performance of these algorithms are significantly susceptible to disruptions caused by communication errors and the inherent unreliability of wireless communication networks \cite{moon2015minimax}. This emphasizes the imperative to transition towards an alternative that is more resilient and economizes on communication costs.
Consequently, we introduce the distributed asynchronous ADMM algorithm tailored for consensus optimization of the global variable.
}

{
To guarantee the progression of an asynchronous algorithm, it is essential to establish a cap on the requisite number of communications per communicative entity. To adhere to the iteration delay bound necessary for effective communication, denoted as \( \Gamma_n \), the inequality \( (1-p^{\text{max}}_n)^{\Gamma_n} < \epsilon^a \) holds, where \( \epsilon^a \) represents the maximal allowable delay for asynchronous iteration exchanges. By taking the logarithm of both sides, we arrive at the inequality \( \Gamma_n \geq \frac{\ln(\epsilon^a)}{\ln(1-p^{\text{max}}_n)} \). Consequently, we assume the following:
\begin{assumption}
\label{assumption:finiteCommunication-inappendix}
    The communication rounds required to conclude a successful iteration are bounded above by
    \begin{equation}
        \Gamma_n = \left\lceil \frac{\ln(\epsilon^a)}{\ln(1-p^{\text{max}}_n)} \right\rceil,
    \end{equation}
    where $ \lceil x \rceil $ denotes the smallest integer greater than or equal to $ x $.
\end{assumption}
This assumption is a conventional one in studies of asynchronous ADMM, as referenced in works such as~\cite{hong2016convergence, hong2017distributed}. In the worst case, this condition is pivotal to ensuring that each participant completes at least one iteration within $ \Gamma_n $ cycles.
\paragraph{Updating $\bm x_n$ and $\bm \sigma_n$ in Users}
We first update the variables $\bm x_n$ and $\bm \sigma_n$ as follows:
\begin{align}
    \label{eq:asyn-PLadmm1-relax-i} 
    \bm x^{t+1}_n & = \arg\min_{\bm x_n} (g_n(\bm x_n) + \langle\bm\sigma^t, \bm x_n - \bm x_{o}\rangle\nonumber\\
    &+ (\rho/2)\|\bm x_n - \bm x^t_{o}\|^2_2), \\ 
    \label{eq:asyn-PLadmm1-relax-sigma} 
    \bm\sigma^{t+1}_n & = \bm\sigma^{t}_n + \rho(\bm x^{t+1}_n - \bm x^{t+1}_{o}),
\end{align}
where the most recent consensus value $\bm x_o$ is received from the server\footnote{Due to the different update speeds of users, the values of $\bm x_o$ of different users are not the same.}.
Furthermore, the termination criteria are the same as in Algorithm \ref{alg:admm}.
Accordingly, the user-side procedure is delineated in Algorithm \ref{alg:async-admm-user}:
\begin{algorithm}[ht]
\caption{Asynchronous ADMM-Consensus in Users (AACU)}
\label{alg:async-admm-user}
\LinesNumbered
% \KwIn{Convergence threshold $ \tau $, update frequency $N_\mathrm{min}$, number of users $ N $}
\KwIn{$\bm  \sigma_n^0, t$ for all users $ n $}
\KwOut{Optimized parameter $ x_o^* $}
\While{$\|\bm x^{t+1}_n - \bm x^{t+1}_{o}\|^2_2 > \epsilon^\mathit{pr}$ or $\rho\|(\bm x^{t+1}_{o}) - \bm x^{t}_{o})\|^2_2 > \epsilon^\mathrm{dl}$}{
    Update $\bm x_n^{t+1} $ using Eq.~\eqref{eq:asyn-PLadmm1-relax-i}\\
    Update $\bm\sigma^{t+1}_n$ using Eq.~\eqref{eq:asyn-PLadmm1-relax-sigma}\\
    Send $\bm  \sigma_n^{t} $ and $\bm  x_n^{t+1} $ to the \textbf{Server}\\
    Wait for $\bm x^n_o$ from the \textbf{Server}\\
    Update $\bm  \sigma_n^{t+1} $ using (8)\\
    $ t \leftarrow t + 1 $\\
}
\end{algorithm}
\paragraph{Updating $\bm x_o$ in Server}
As shown in Algorithm \ref{alg:async-admm-server}, the server continuously accumulates updates from users until it has gathered a minimum of $N_\mathrm{min}$ updates or until the maximum delay surpasses a threshold, designated as $\Gamma = \max \{\Gamma_1 \cdots \Gamma_N\}$. 
In this context, the set of users who have submitted updates in phase $t$ is denoted as $\Phi^t$. For each user $n$ within this group, their corresponding $\bm x_n$ and $\bm \sigma_n$ values are updated.
Conversely, for those not in $\Phi^t$, their values remain unchanged. Upon completion of this process, the server proceeds to update the value of $\bm x_o$ in accordance with Eq.~\eqref{eq:asyn-PLadmm1-relax-o} and subsequently broadcasts this revised value.
\begin{align}
    \label{eq:asyn-PLadmm1-relax-o} 
    \bm x^{t+1}_{o} & = \frac{1}{N}\sum_{n\in\mathcal{N}}(\bm x^{t+1}_n + (\frac{1}{\rho}\bm\sigma^{t}_n)),
\end{align}
\begin{algorithm}[ht]
\caption{Asynchronous ADMM-Consensus by Server (AACS)}
\label{alg:async-admm-server}
\LinesNumbered
\KwIn{$ t, \bm x_n, \bm \sigma_n$ for $ i = 1, 2, \ldots, N $}
\KwOut{Optimized parameter $ \bm x_o^* $}
Initialize \\
\While{$\|\bm x^{t+1}_n - \bm x^{t+1}_{o}\|^2_2 > \epsilon^\mathit{pr}$ or $\rho\|(\bm x^{t+1}_{o}) - \bm x^{t}_{o})\|^2_2 > \epsilon^\mathrm{dl}$}{
    $ \tau_n \leftarrow 0 $ for $ n = 1, 2, \ldots, N $\\
    \While{less than $N_\mathrm{min}$ updates received or $ \max(\tau_1, \tau_2, \ldots, \tau_N) > \Gamma $}{
        Wait for updates\\
    }
    \ForEach{user $ n \in \Phi^t $}{
        $ \tau_n \leftarrow 1 $\\
        $ \bm x_n \leftarrow $ newly received $\bm x_n $ from user $ n $\\
        $ \bm \sigma_n \leftarrow $ newly received $ \sigma_n $ from user $ n $\\
    }
    \ForEach{user $ n \notin \Phi^t $}{
        $ \tau_n \leftarrow \tau_n + 1 $\\
    }
    Update $\bm x_o^{t+1} $ by Eq.~\eqref{eq:asyn-PLadmm1-relax-o}\\
    Broadcast $\bm x_o^{t+1} $ to all the users in $ \Phi^t $\\
    $ k \leftarrow t + 1 $\\
}
\end{algorithm}
% \begin{remark}
% If $N_{min} = N$ and $\Gamma = 1$, the combination of Algorithms \ref{alg:async-admm-user} and \ref{alg:async-admm-server} is equivalent to Algorithm \ref{alg:admm}.
% \end{remark}
\begin{theorem}
\label{thm:convergence-asysn-admm}
Let $ (x^*, \bm x_o^*) $ represent the optimal primal solution of problem \textbf{P1-3}, and $ \{\bm\sigma_n^*\}_{n=1}^N $ the corresponding optimal dual solution. Then,
\begin{align}
&\mathbb{E} \left[ \sum_{n=1}^N g_n(\bm{x}_n) - g_n(\bm x^*) + \langle \bm\sigma_n^*, \bm{x}_n - \bm x_o^* \rangle \right] \nonumber\\
\leq &\frac{N\Gamma}{2T N_\mathrm{min}} \left\{ \sum_{n=1}^N \rho \|\bm{x}_o^0 - \bm x_o^*\|^2 + \frac{1}{\rho} \|\bm\sigma_n^0 - \bm\sigma_n^*\|^2 \right\},
\end{align}
where $ \bm{x}_o^0 $ and $ \bm\sigma_n^0 $ denote the initial values of $\bm x_o $ and $ \bm\sigma_n $, respectively, for user $ n $.
\end{theorem}
\begin{proof}
The result of Theorem \ref{thm:convergence-asysn-admm} can be obtained via the proof of combining \cite[Theorem 4.2]{zhang2014asynchronous} and \cite[Corollary 4.3]{zhang2014asynchronous}.
\end{proof}
Consequently, the theorem \ref{thm:convergence-asysn-admm} implies that the combination of Algorithms \ref{alg:async-admm-user} and \ref{alg:async-admm-server} can sublinearly converge.
Specifically, the asynchronous algorithm convergence rate is $O(\frac{N\Gamma}{T N_\mathrm{min}})$.  
}
    
\section{A Proof of Proposition \ref{pro:npl-mp}}
\label{app:pro:npl-mp}
    A brief idea of the proof is given as follows. 
    According to nonlinear fractional programming~\cite{dinkelbach1967nonlinear}, $\theta^{*}_n$ is achieved if and only if
	\begin{align}
	\mathop{\min}_{\bm x_n} f^{p,u}_n(\bm x_n) - \theta^{*}_n f^{p,l}_n(\bm x_n)
	\nonumber\\ 
       = f^{p,u}_n(\bm x_n^*) - \theta^{*}_nf^{p,l}_n(\bm x_n^*) = 0.
    \label{eq:NFP-optimal-theta}
    \end{align}
which outlines the necessary and sufficient criteria in order to reach $\theta^{*}_n$.
In light of this, $\{\bm x_n\}$ can be acquired by resolving the following transformation problem.

\section{A Proof of Proposition \ref{theorem:NFP_Problem}}\label{app:theorem:NFP_Problem}
Here is a concise proof followed by a comprehensive one.
\subsection{A Concise Proof}
\begin{enumerate}  
\item Given the complexity of the proof, we will only provide a brief overview of the main idea here. We begin by computing the Hessian matrix of the entire function. Due to differing priority attributes ($\delta(n_1)> \delta(n_2)$, $\delta(n_1)=\delta(n_2)$, n1=n2,$\delta(n_1)< \delta(n_2)$), we need to examine the second derivative of the function in at least 16 cases.

\item Through analyzing the properties of the Hessian matrix, we observe that the function is Lipschitz smooth or gradient Lipschitz continuous as long as the values in the Hessian matrix are finite. 

\item We can also obtain the value of $\ell_n$in $\ell_n I \succeq \nabla^{2}F$
 based on the characteristics of the Hessian matrix.
\end{enumerate}
\subsection{A Comprehensive Proof}
    For computational convenience, we present some functions, where
    $\Phi_{\delta(n)}(\bm x) = 1-\sum_{\delta=1}^{\delta(n)}\sum_{n'=1}^{\mathcal{N}^{\delta}}\sum_{c\in\mathcal{C}}p_{n',c}\frac{\eta^u_{n',c}\lambda_{n'}}{\mu_{n'}}$,
    $\Phi_0(\bm x) = 1$,
    $\Upsilon(\bm x) = \frac{1}{2}\sum_{\delta=1}^{\Delta}\sum_{n'=1}^{\mathcal{N}^{\delta}}\sum_{c\in\mathcal{C}}p_{n',c}\eta^u_{n',c}\lambda_{n'}\nu_{n'}$,
    and $\Psi_n(\bm x)=\sum_{c\in\mathcal{C}}p_{n,c}\eta^u_{n,c}\lambda_{n}$.
    We also give some useful parameters in advance, where
    $\{n_1, n_2 \in \mathcal{N}\}$, $\{c_1, c_2\in \mathcal{C}\}$,
    $\kappa_{1}=\frac{p_{n_1,c_1}}{\mu_{n_1}}>0$, 
    $\kappa^1_\mathrm{max}=\frac{1}{\mu_\mathrm{min}}>0$,
    and $p_{n,c}\in[0,1]$.
    % and $\kappa_{1}\kappa_{2}=\kappa_{1}\kappa_{2}$
    % $\kappa^2_\mathrm{max}=\frac{1}{\mu^2_\mathrm{min}}>0$.
    Additionally, since $\nabla_{\{t^\mathrm{tr} \in \bm x\}} g^{p}_n(\bm x) = 0$, we just ignore it.
    Due to $\nabla^2_{\eta_{n_1,c_1}\eta_{n_2,c_2}} g^{p}_n(\bm x) = \lambda_{n_1}\lambda_{n_2}\nabla^2_{\lambda_{n_1,c_1}\lambda_{n_2,c_2}} g^{p}_n(\bm x)$, we have
\subsubsection{Hessian Matrix}
    We first analyze the function $f^{p,u}_n(\bm x)$,
    \begin{align}
        f^{p,u}_n(\bm x)  & = \Phi_{\delta(n)}(\bm x)\Phi_{\delta(n)-1}(\bm x) + \Upsilon(\bm x)\Psi_n(\bm x).
    \end{align}

    \noindent{\bf Case I)} If $\delta(n_1)<\delta(n)$, we obtain first-order derivative.
    \begin{align}
    \frac{\partial f^{p,u}_n(\bm x)}{\partial \lambda_{n_1,c_1}} &= -\frac{p_{n_1,c_1}}{\mu_{n_1}}(\Phi_{\delta(n)}(\bm x) + \Phi_{\delta(n)-1}(\bm x))\nonumber\\
    &+ \frac{1}{2}p_{n_1,c_1}\nu_{n_1}\Psi_n(\bm x)
    \end{align}
    We obtain some second-order derivatives.
    \begin{equation}  
		\frac{\partial^2 f^{p,u}_n(\bm x)}{\partial\lambda_{n_1,c_1} \partial\lambda_{n_2,c_2}}\\
		= \left\{
		\begin{array}{ll}
		 2\frac{p_{n_1,c_1}p_{n_2,c_2}}{\mu_{n_1}\mu_{n_2}}, \quad \delta(n_2)<\delta(n)\\
		 \frac{p_{n_1,c_1}p_{n_2,c_2}}{\mu_{n_1}\mu_{n_2}}+\frac{1}{2}p_{n_1,c_1}p_{n_2,c_2}\nu_{n_1},
		\quad n_2 = n\\
		 \frac{p_{n_1,c_1}p_{n_2,c_2}}{\mu_{n_1}\mu_{n_2}},
		\quad \delta(n_2) = \delta(n), \quad n_2 \ne n\\
		 0, \quad \delta(n_2)>\delta(n)\\
		\end{array}
		\right..
	\end{equation}
    
    \noindent{\bf Case II)} If $\delta(n_1)=\delta(n)~\&~n_1 = n$, we obtain first-order derivative.
    \begin{align}
    \frac{\partial f^{p,u}_n(\bm x)}{\partial \lambda_{n_1,c_1}} &= -\frac{p_{n_1,c_1}}{\mu_{n_1}}\Phi_{\delta(n)-1}(\bm x) + p_{n_1,c_1}\Upsilon(\bm x)\nonumber\\
    &+\frac{1}{2}p_{n_1,c_1}\nu_{n_1}\Psi_n(\bm x)
    \end{align}
    We obtain some second-order derivatives.
    \begin{equation}  
		\frac{\partial^2 f^{p,u}_n(\bm x)}{\partial\lambda_{n_1,c_1} \partial\lambda_{n_2,c_2}} = \left\{
		\begin{array}{ll}
		& \frac{p_{n_1,c_1}p_{n_2,c_2}}{\mu_{n_1}\mu_{n_2}}+\frac{1}{2}p_{n_1,c_1}p_{n_2,c_2}\nu_{n_2},\nonumber\\
		&\quad \delta(n_2)<\delta(n)\\
		& \frac{1}{2}p_{n_1,c_1}p_{n_2,c_2}\nu_{n_1}+\frac{1}{2}p_{n_1,c_1}p_{n_2,c_2}\nu_{n_2},\nonumber\\  
		&\quad n_2 = n\\
		& \frac{1}{2}p_{n_1,c_1}p_{n_2,c_2}\nu_{n_2},  \quad \delta(n_2) = \delta(n), \nonumber\\ 
		&\quad n_2 \ne n\\
		& \frac{1}{2}p_{n_1,c_1}p_{n_2,c_2}\nu_{n_2}, \quad \delta(n_2)>\delta(n)\\
		\end{array}
		\right..
	\end{equation}
	
	\noindent{\bf Case III)} If $\delta(n_1)=\delta(n)~\&~n_1 \ne n$, we obtain first-order derivative.
    \begin{align}
    \frac{\partial f^{p,u}_n(\bm x)}{\partial \lambda_{n_1,c_1}} = -\frac{p_{n_1,c_1}}{\mu_{n_1}}\Phi_{\delta(n)-1}(\bm x) +\frac{1}{2}p_{n_1,c_1}\nu_{n_1}\Psi_n(\bm x)
    \end{align}
    We obtain some second-order derivatives.
    \begin{equation}  
		\frac{\partial^2 f^{p,u}_n(\bm x)}{\partial\lambda_{n_1,c_1} \partial\lambda_{n_2,c_2}} = \left\{
		\begin{array}{ll}
		& \frac{p_{n_1,c_1}p_{n_2,c_2}}{\mu_{n_1}\mu_{n_2}}, \quad \delta(n_2)<\delta(n)\\
		& \frac{1}{2}p_{n_1,c_1}p_{n_2,c_2}\nu_{n_1},  \quad n_2 = n\\
		& 0,  \quad \delta(n_2) = \delta(n), \quad n_2 \ne n\\
		& 0, \quad \delta(n_2)>\delta(n)\\
		\end{array}
		\right..
	\end{equation}
	
    \noindent{\bf Case IV)} If $\delta(n_1)>\delta(n)$, we obtain first-order derivative.
    \begin{align}
    \frac{\partial f^{p,u}_n(\bm x)}{\partial \lambda_{n_1,c_1}} =\frac{1}{2}p_{n_1,c_1}\nu_{n_1}\Psi_n(\bm x)
    \end{align}
    We obtain some second-order derivatives.
    \begin{equation}  
		\frac{\partial^2 f^{p,u}_n(\bm x)}{\partial\lambda_{n_1,c_1} \partial\lambda_{n_2,c_2}} = \left\{
		\begin{array}{ll}
		& 0, \quad \delta(n_2)<\delta(n)\\
		& \frac{1}{2}p_{n_1,c_1}p_{n_2,c_2}\nu_{n_1},  \quad n_2 = n\\
		& 0,  \quad \delta(n_2) = \delta(n), \quad n_2 \ne n\\
		& 0, \quad \delta(n_2)>\delta(n)\\
		\end{array}
		\right..
	\end{equation}

	As for $f^{p,l}_n(\bm x)$, we have some different cases, as follows.
	\begin{align}
    f^{p,l}_n(\bm x)  & = \Psi_n(\bm x)\Phi_{\delta(n)}(\bm x)\Phi_{\delta(n)-1}(\bm x).
    \end{align}   
    
    \noindent{\bf Case I)} If $\delta(n_1)<\delta(n)$, we obtain first-order derivative.
    \begin{align}
    \frac{\partial f^{p,l}_n(\bm x)}{\partial \lambda_{n_1,c_1}} = -\frac{p_{n_1,c_1}}{\mu_{n_1}}(\Phi_{\delta(n)}(\bm x)+\Phi_{\delta(n)-1}(\bm x))\Psi_n(\bm x)
    \end{align}
    We obtain some second-order derivatives.
    \begin{equation}  
		\noindent\frac{\partial^2 f^{p,l}_n(\bm x)}{\partial\lambda_{n_1,c_1} \partial\lambda_{n_2,c_2}}\nonumber
	\end{equation}
    \begin{equation}  
		= \left\{
		\begin{array}{ll}
		& 2\frac{p_{n_1,c_1}p_{n_2,c_2}}{\mu_{n_1}\mu_{n_2}}\Psi_n(\bm x), \quad \delta(n_2)<\delta(n)\\
		& \frac{p_{n_1,c_1}p_{n_2,c_2}}{\mu_{n_1}\mu_{n_2}}(\Psi_n(\bm x)-\frac{\nu_{n_2}}{2\mu_{n_2}}\Phi_{\delta(n)}(\bm x)-\frac{\nu_{n_2}}{2\mu_{n_2}}\Phi_{\delta(n)-1}(\bm x)),\\
		& ~~~\quad n_2 = n\\
		& \frac{p_{n_1,c_1}p_{n_2,c_2}}{\mu_{n_1}\mu_{n_2}}\Psi_n(\bm x),\quad \delta(n_2) = \delta(n), \quad n_2 \ne n\\
		& 0, \quad \delta(n_2)>\delta(n)\\
		\end{array}
		\right..
	\end{equation}
    
    \noindent{\bf Case II)} If $\delta(n_1)=\delta(n)~\&~n_1 = n$, we obtain first-order derivative.
    \begin{align}
    \frac{\partial f^{p,l}_n(\bm x)}{\partial \lambda_{n_1,c_1}} &= p_{n_1,c_1}\Phi_{\delta(n)}(\bm x)\Phi_{\delta(n)-1}(\bm x)\nonumber\\
    &-\frac{p_{n_1,c_1}}{\mu_{n_1}}\Psi_n(\bm x)\Phi_{\delta(n)-1}(\bm x)
    \end{align}
    We obtain some second-order derivatives.
    \begin{equation}  
		\frac{\partial^2 f^{p,l}_n(\bm x)}{\partial\lambda_{n_1,c_1} \partial\lambda_{n_2,c_2}}\nonumber
		\end{equation}
    \begin{equation}  
     = \left\{
		\begin{array}{ll}
		& \frac{p_{n_1,c_1}p_{n_2,c_2}}{\mu_{n_1}\mu_{n_2}}(\Psi_n(\bm x)-\frac{\nu_{n_2}}{2\mu_{n_2}}\Phi_{\delta(n)}(\bm x)-\frac{\nu_{n_2}}{2\mu_{n_2}}\Phi_{\delta(n)-1}(\bm x)),\\
		&\quad \delta(n_2)<\delta(n)\\
		& -\frac{p_{n_1,c_1}p_{n_2,c_2}}{\mu_{n_1}\mu_{n_2}}(\frac{\nu_{n_1}}{2\mu_{n_1}}+\frac{\nu_{n_2}}{2\mu_{n_2}})\Phi_{\delta(n)-1}(\bm x),  \quad n_2 = n\\
		& -\frac{p_{n_1,c_1}p_{n_2,c_2}}{\mu_{n_1}\mu_{n_2}}\frac{\nu_{n_2}}{2\mu_{n_2}}\Phi_{\delta(n)-1}(\bm x),  \quad \delta(n_2) = \delta(n), \quad n_2 \ne n\\
		& 0, \quad \delta(n_2)>\delta(n)\\
		\end{array}
		\right..
	\end{equation}
	
	\noindent{\bf Case III)} If $\delta(n_1)=\delta(n)~\&~n_1 \ne n$, we obtain first-order derivative.
    \begin{align}
    \frac{\partial f^{p,l}_n(\bm x)}{\partial \lambda_{n_1,c_1}} = -\frac{p_{n_1,c_1}}{\mu_{n_1}}\Psi_n(\bm x)\Phi_{\delta(n)-1}(\bm x)
    \end{align}
    We obtain some second-order derivatives.
    \begin{equation}  
		\frac{\partial^2 f^{p,l}_n(\bm x)}{\partial\lambda_{n_1,c_1} \partial\lambda_{n_2,c_2}} = \left\{
		\begin{array}{ll}
		& \frac{p_{n_1,c_1}p_{n_2,c_2}}{\mu_{n_1}\mu_{n_2}}\Psi_n(\bm x), \quad \delta(n_2)<\delta(n)\\
		& -\frac{p_{n_1,c_1}p_{n_2,c_2}}{\mu_{n_1}\mu_{n_2}}\frac{\nu_{n_2}}{2\mu_{n_2}}\Phi_{\delta(n)-1}(\bm x),  \quad n_2 = n\\
		& 0, \quad n_2 \ne n\\
		& 0, \quad \delta(n_2)>\delta(n)\\
		\end{array}
		\right..
	\end{equation}
	
    \noindent{\bf Case IV)} If $\delta(n_1)>\delta(n)$, we obtain first and second-order derivative.
    \begin{align}
    &\frac{\partial f^{p,l}_n(\bm x)}{\partial \lambda_{n_1,c_1}} =0
    &\frac{\partial^2 f^{p,l}_n(\bm x)}{\partial\lambda_{n_1,c_1} \partial\lambda_{n_2,c_2}} = 0.
    \end{align}
    
    Above all, we have the first-order derivative and the second-order derivative information of $g^{p}_n(\bm x) = f^{p,u}_n(\bm x) - \theta^{*}_n f^{p,l}_n(\bm x)$.

\subsubsection{The Relationship between Hessian Matrix and Lipschitz Smooth}
	Obviously, for any $n$, $g^{p}_n(\bm x)$ is a smooth function.
	
	On the other hand, the twice differentiable function $g^{p}_n(\bm x)$ has a Lipschitz continuous gradient with modulus $\ell_n$ if and only if its Hessian satisfies $\ell_n I \succeq \nabla^{2} g^{p}_n(\bm x)$. We have $\ell_n$ is equal to the maximum eigenvalue of $\nabla^{2} g^{p}_n(\bm x)$,\cite{zhang2020fedpd}.
	
	In order to obtain the maximum eigenvalue, we ought to get the maximum value of the sum of the absolute values of the elements of the column of the matrix $\nabla^{2} g^{p}_n(\bm x)$, which is because that if $\xi$ is the eigenvector of matrix $\mathcal{A}$ with respect to eigenvalue $\omega$. 
	
	According to the basic property, we have $\sum_{j=1}^{n} A_{i j} \xi_{j}=\omega \xi_{i}$, \ie $\sum_{j=1}^{n}\left|A_{i j} \right|\left|\xi_{j}\right| \geq|\omega |\left|\xi_{i}\right|$.
	
	Summing over both sides, it is 
	\begin{align}
	    \sum_{i=1}^{n}\sum_{j=1}^{n}\left|A_{i j} \right|\left|\xi_{j}\right| \geq\sum_{i=1}^{n}|\omega |\left|\xi_{i}\right|
	\end{align}
	
	Let $M=\max_{1 \leq j \leq n} \sum_{i=1}^{n}\left|A_{i j}\right|$, we have
	\begin{align}
	    M \sum_{j=1}^{n}\left|\xi_{j}\right| \geq \sum_{j=1}^{n}\left|\xi_{j}\right| \sum_{i=1}^{n}\left|A_{i j}\right| \geq \sum_{i=1}^{n}|\omega |\left|\xi_{i}\right| \geq|\omega | \sum_{i=1}^{n}\left|\xi_{i}\right|
	\end{align}

	Since $\xi$ is not a zero vector, we obtain 
	\begin{align}
	    |\omega | \leq \max_{1 \leq j \leq n} \sum_{i=1}^{n}\left|A_{i j}\right|
	\end{align}

\subsubsection{Solve for the value of $\ell_n$}
    Upper bound on the sum of the absolute values of the elements in the columns of the matrix $\nabla^{2}_{\bm\eta^u} g^{p}_n(\bm x)$ is expressed as $\varphi(\bm x)\lambda^2_{max}$ with $\lambda_{max}=\mathop{\max}_n {\lambda_n}$:
	
	\begin{align}
        \label{eq:hessian}
	    \varphi(\bm x) &\leq \mathop{\max} \Big\{\inf |\varphi_1(\bm x)|, \inf |\varphi_2(\bm x)|, \inf |\varphi_3(\bm x)|, \inf |\varphi_4(\bm x)| \Big\} \nonumber\\
	    &\mathop{\leq}^{(a)} \mathop{\max} \Big\{\inf |\varphi_1(\bm x)|, \inf |\varphi_2(\bm x)|, \inf |\varphi_3(\bm x)| \Big\}\nonumber\\
	    &\mathop{\leq}^{(b)} \mathop{\max} \Big\{\inf |\varphi_1(\bm x)|, \inf |\varphi_3(\bm x)| \Big\}\nonumber\\
	    &\mathop{\leq}^{(c)} \mathop{\max} \Big\{\inf |\varphi_3(\bm x)| \Big\}\nonumber\\
	    &\mathop{\leq}^{(d)} \frac{1}{\mu^2_\mathrm{min}}\Big(
	        \sum_{n_1\in \mathcal{N}} |\frac{\nu_{n}}{2\mu_{n}\mu_{n_1}}|
	        + \sum_{n_1\in \{<\delta(n)\}\bigcup \{n\}}|1+\theta^{*}_n \frac{\nu_{n}}{2\mu_{n}}|\nonumber\\
	        &~~~~+ \sum_{n_1\in \{\leq\delta(n)\}}|\theta^{*}_n \frac{\nu_{n}}{2\mu_{n}}|\Big),
	\end{align}
    where $\varphi_1(\bm x)$, $\varphi_2(\bm x)$, $\varphi_3(\bm x)$, and $\varphi_4(\bm x)$ are
    
	\noindent{\bf Case I)} If $\delta(n_2)<\delta(n)$, we have
	\begin{align}
	    \varphi_1(\bm x)
	   &= \overbrace{\sum_{n_1\in \{<\delta(n)\}}\mid 2\kappa_{1}\kappa_{2}(1-\theta^{*}_n\Psi_n(\bm x))\mid}^{\varphi_{1a}(\bm x)}\nonumber\\
	    &+ \overbrace{\sum_{n_1\in \{n\}}\mid \kappa_{1}\kappa_{2}(1+\frac{\nu_{n}}{2\mu_{n}\mu_{n_2}}+\theta^{*}_n(\frac{\nu_{n}}{2\mu_{n}}\Phi_{\delta(n)}(\bm x)\mid }^{\varphi_{1b}(\bm x)} \nonumber\\
	    &+\frac{\nu_{n}}{2\mu_{n}}\Phi_{\delta(n)-1}(\bm x))-\Psi_n(\bm x))\nonumber\\
	    &+ \overbrace{\sum_{n_1\in \{=\delta(n)\}/n}\mid  \kappa_{1}\kappa_{2}(1-\theta^{*}_n\Psi_n(\bm x))\mid }^{\varphi_{1c}(\bm x)} \nonumber\\
	    &+ \overbrace{\sum_{n_1\in \{>\delta(n)\}} 0 }^{\varphi_{1d}(\bm x)} .
	\end{align}
	
	\noindent{\bf Case II)} If $\delta(n_2)=\delta(n)~\&~n_2\ne n$, we derive
	\begin{align}
	    \varphi_2(\bm x)
	   &= \overbrace{\sum_{n_1\in \{<\delta(n)\}}\mid \kappa_{1}\kappa_{2}(1-\theta^{*}_n\Psi_n(\bm x))\mid }^{\varphi_{2a}(\bm x)}\nonumber\\
	    &+ \overbrace{\sum_{n_1\in \{n\}}\mid \kappa_{1}\kappa_{2}(\frac{\nu_{n}}{2\mu_{n}\mu_{n_2}}+\theta^{*}_n \frac{\nu_{n}}{2\mu_{n}}\Phi_{\delta(n)}(\bm x))\mid }^{\varphi_{2b}(\bm x)} \nonumber\\
	    &+ \overbrace{\sum_{n_1\in \{=\delta(n)\}/n}\mid \kappa_{1}\kappa_{2}(1-\theta^{*}_n\Psi_n(\bm x))\mid }^{\varphi_{2c}(\bm x)}\nonumber\\
	    &+ \overbrace{\sum_{n_1\in \{>\delta(n)\}} 0 }^{\varphi_{2d}(\bm x)}.
	\end{align}
	
	\noindent{\bf Case III)} If $n_2 = n$, we obtain
	\begin{align}
	    \varphi_3(\bm x)
	   &= \overbrace{\sum_{n_1\in \{<\delta(n)\}}\mid \kappa_{1}\kappa_{2}(1+\frac{\nu_{n}}{2\mu_{n}\mu_{n_1}}+\theta^{*}_n(\frac{\nu_{n}}{2\mu_{n}}\Phi_{\delta(n)}(\bm x)}^{\varphi_{3a}}\nonumber\\
          &+\frac{\nu_{n}}{2\mu_{n}}\Phi_{\delta(n)-1}(\bm x))(\bm x)
	   -\Psi_n(\bm x))\mid \nonumber\\
	    &+ \overbrace{\sum_{n_1\in \{n\}}\mid 2\kappa_{1}\kappa_{2}(1+\theta^{*}_n \frac{\nu_{n}}{2\mu_{n}}\Phi_{\delta(n)-1}(\bm x))\mid }^{\varphi_{3b}(\bm x)}\nonumber\\
	    &+ \overbrace{\sum_{n_1\in \{=\delta(n)\}/n}\mid  \kappa_{1}\kappa_{2}(\frac{\nu_{n}}{2\mu_{n}\mu_{n_1}}+\theta^{*}_n \frac{\nu_{n}}{2\mu_{n}}\Phi_{\delta(n)}(\bm x))\mid }^{\varphi_{3c}(\bm x)}\nonumber\\
	    &+ \overbrace{\sum_{n_1\in \{>\delta(n)\}}\mid \kappa_{1}\kappa_{2}\frac{\nu_{n}}{2\mu_{n}\mu_{n_1}}\mid }^{\varphi_{3d}(\bm x)}.
	\end{align}
	
	\noindent{\bf Case IV)} If $\delta(n_2)>\delta(n)$, we gain
	\begin{align}
	    \varphi_4(\bm x)
	   &= \overbrace{\sum_{n_1\in \{<\delta(n)\}} 0 }^{\varphi_{4a}(\bm x)}+ \overbrace{\sum_{n_1\in \{n\}}\mid \kappa_{1}\kappa_{2}\frac{\nu_{n}}{2\mu_{n}\mu_{n_2}}\mid }^{\varphi_{4b}(\bm x)}\nonumber\\
	    &+ \overbrace{\sum_{n_1\in \{=\delta(n)\}/n} 0 }^{\varphi_{4c}(\bm x)}
	    + \overbrace{\sum_{n_1\in \{>\delta(n)\}} 0 }^{\varphi_{4d}(\bm x)}.
	\end{align}

    Eq.\eqref{eq:hessian} holds because:
	\begin{itemize}
	    \item  (a) holds because $\varphi_2(\bm x)>\varphi_4(\bm x)$.
	
    	\item  (b) holds because $\inf | \varphi_{2}(\bm x) | < \inf | \varphi_{1}(\bm x) |$
	    
	    \item  (c) holds because $\inf | \varphi_{1}(\bm x) | < \inf | \varphi_{3}(\bm x) |$
	    
	    \item  (d) holds because
	    
	    \begin{align}
         	\inf |\varphi_3(\bm x)| 
	        \leq |\kappa^2_\mathrm{max}|\Big(\sum_{n_1\in \{<\delta(n)\}}| 1+\frac{\nu_{n}}{2\mu_{n}\mu_{n_1}}+2\theta^{*}_n \frac{\nu_{n}}{2\mu_{n}}|\nonumber\\
	        + \sum_{n_1\in \{n\}}| 2(1+\theta^{*}_n \frac{\nu_{n}}{2\mu_{n}})|
	        + \sum_{n_1\in \{=\delta(n)\}/n}| \frac{\nu_{n}}{2\mu_{n}\mu_{n_1}}+\theta^{*}_n \frac{\nu_{n}}{2\mu_{n}}|\nonumber\\
	        + \sum_{n_1\in \{>\delta(n)\}}| \frac{\nu_{n}}{2\mu_{n}\mu_{n_1}}| \Big)\nonumber\\
	        = \frac{1}{\mu^2_\mathrm{min}}\Big(
	        \sum_{n_1\in \mathcal{N}} |\frac{\nu_{n}}{2\mu_{n}\mu_{n_1}}|
	        + \sum_{n_1\in \{<\delta(n)\}\bigcup \{n\}}|1+\theta^{*}_n \frac{\nu_{n}}{2\mu_{n}}|\nonumber\\
	        + \sum_{n_1\in \{\leq\delta(n)\}}|\theta^{*}_n \frac{\nu_{n}}{2\mu_{n}}|\Big)
	    \end{align}
	\end{itemize}
    
    Let $\omega^{G}$ denote the the eigenvalue of $\nabla^{2} g^{p}_n$, 
    \begin{align}
        &|\omega^{G}|\nonumber\\
        \leq& \frac{\lambda^2_\mathrm{max}}{\mu^2_\mathrm{min}}\Big(
	        \sum_{n_1\in \mathcal{N}} |\frac{\nu_{n}}{2\mu_{n}\mu_{n_1}}|
	        + \sum_{n_1\in \{<\delta(n)\}\bigcup \{n\}}|1+\theta^{*}_n \frac{\nu_{n}}{2\mu_{n}}|\nonumber\\
	        &+ \sum_{n_1\in \{\leq\delta(n)\}}|\theta^{*}_n \frac{\nu_{n}}{2\mu_{n}}|\Big)\nonumber\\
	=& K_n,
    \end{align}
    where $\mu_{min} = \mathop{\min}_n \mu_{n}$ and $\lambda_{max} = \mathop{\max}_n \lambda_{n}$.
    Thus, if $\ell_n = K_n$ are limited, $\ell_n I \succeq \nabla^{2} g^{p}_n(\bm x)$ and $g^{p}_n(\bm x)$ is Lipschiz smooth.
    
    % In the meanwhile, since $\frac{1}{\mu_\mathrm{min}}$ and $\mu_n$ are limited, $\nabla^{2} g^{p}_n(\bm x) \succeq \gamma_n I$, \ie $\gamma_n = -K_n$.

        % \onecolumn
\section{A Proof of Theorem \ref{theorem:innerConvergence1}}\label{app:theorem:innerConvergence1}    
\subsection{Convergence Analysis}
    % We prove Lemma \ref{theorem:innerConvergence1} here.
    First, we prove the convergence of Algorithm \ref{alg:nonconvexadmm}.
    There are some properties in $g^{p}_n\left(\bm\lambda \right)$, as follows,
    \begin{enumerate}
        \item there exists a positive constant $\ell_n>0$, it has
        $\left\|\nabla g^{p}_n\left(\bm\lambda \right)-\nabla g^{p}_n\left(\bm x^{'} \right)\right\| \leq \ell_n\left\|\bm\lambda-\bm x^{'} \right\| \quad \forall \bm\lambda, \bm x^{'}$
        \item for any $n$, the $\rho_{n}$ chosen is large enough, 
        $\rho_{n}>\max \left\{\frac{2 \ell_n^{2}}{\varepsilon_{n}}, \ell_n\right\}$, 
        where $\varepsilon_{n}$ satisfies $\varepsilon_{n}I \preceq \nabla^2 L^{p}_n(\bm x)$
        \item $\mathop{\min}_{\bm x\in \Omega} g^{p}_n(\bm x)>-\infty$ 
    \end{enumerate}
    
    Proof of Property 1), Proposition \ref{theorem:NFP_Problem} has proven there exists a positive constant $\ell_n>0$ satisfies $\ell_n I \succeq \nabla^{2} g^{p}_n(\bm x)$, \ie $\left\|\nabla g^{p}_n\left(\bm\lambda \right)-\nabla g^{p}_n\left(\bm x^{'} \right)\right\| \leq \ell_n\left\|\bm\lambda-\bm x^{'} \right\| \quad \forall \bm\lambda, \bm x^{'}$.
    
    Proof of Property 2), obviously, since $\varepsilon_{n}I \preceq \nabla^2 L^{p}_n(\bm x)$, $\varepsilon_{n} = -\ell_n+\rho_n$. Due to $\rho_n>2\ell_n$ and $\rho_n>0$, $(\rho_n-2\ell_n)(\rho_n+\ell_n)>0$. Thus, $\rho_{n}>\max \left\{\frac{2 \ell_n^{2}}{\varepsilon_{n}}, \ell_n\right\}$.
    
    Proof of Property 3), when $\bm x\in \Omega$, we have,
    \begin{align}
        & g^{p}_n(\bm x) = f^{p,u}_n(\bm x) - \theta^{*}_n f^{p,l}_n(\bm x) \nonumber\\
        &= \Phi_{\delta(n)}(\bm x)\Phi_{\delta(n)-1}(\bm x) + \Upsilon(\bm x)\Psi_n(\bm x)\nonumber\\
        &- \theta^{*}_n\Psi_n(\bm x)\Phi_{\delta(n)}(\bm x)\Phi_{\delta(n)-1}(\bm x) \nonumber\\
        & \geq - \theta^{*}_n\lambda_n >-\infty
    \end{align}
    where $\theta^{*}_n$ and $\lambda_n$ are constant, $\Phi_{\delta(n)}(\bm x), \Phi_{\delta(n)-1}(\bm x) \in [0,1]$, $\Psi_n(\bm x) \in [0,\lambda_n]$ and $\Upsilon(\bm x) > 0$
    
    \begin{itemize}
        \item \textbf{Dual variable convergence} $L^{p}_n(x;y)$ will increase after each dual.\cite{boyd2011distributed}
        \item \textbf{Consensus convergence} The consensus constraint is satisfied eventually $\lim_{t \rightarrow \infty}\left\|\bm x_{n}^{t+1}-\bm x_{o}^{t+1}\right\|=0, \forall n$.
        \item \textbf{Objective Function convergence} Above the property 1) and 2), we have    
    \begin{align}
        &L^{p}_n\left(\left\{\bm x_{n}^{t+1}\right\}, \bm x_{o}^{t+1} ; y^{t+1}\right)-L^{p}_n\left(\left\{\bm x^t_{n}\right\}, \bm x^t_{o} ; \bm\sigma^t\right)\nonumber\\
        &\leq \sum_{n\in\mathcal{N}}\left(\frac{L_{n}^{2}}{\rho_{n}}-\frac{\varepsilon_{n}}{2}\right)\left\|\bm x_{n}^{t+1}-\bm x^t_{n}\right\|^{2}\nonumber\\
        &-\frac{\sum_{n\in\mathcal{N}}\rho_{n}}{2}\left\|\bm x_{o}^{t+1}-\bm x^t_{o}\right\|^{2}<0
    \end{align}
    according to~\cite{hong2016convergence}.
    Thus, $L^{p}_n(x;y)$ will decrease after each dual.
    Furthermore, according to property 3), we obtain $L^{p}_n\left(\left\{\bm x_{n}^{t+1}\right\}, \bm x_{o}^{t+1} ; y^{t+1}\right)$ is limited.
    \end{itemize}
    
    Thus, the Algorithm \ref{alg:nonconvexadmm} will converge to the set of stationary solutions.
    
    \subsection{Convergence Rate}
    We have proven some properties in $g^{p}_n\left(\bm\lambda \right)$ and the convergence of Algorithm \ref{alg:nonconvexadmm}.
    Based on \cite[Theorem 2.5]{hong2016convergence},
    % [\cite{hong2016convergence}, Theorem 2.5.], 
    we have
    \begin{equation}
        \epsilon^{\mathrm{ac}} <\frac{k^{\Gamma}(L^{p}\left(\left\{\bm x_{n}^{1}\right\}, \bm x_{o}^{1}, \bm \sigma^{1}\right)-\underline{G^{p}})}{p^\mathrm{syn}\Gamma^\mathrm{syn}},
    \end{equation}
    where $p^\mathrm{syn}\Gamma^\mathrm{syn}$ means the number of successful iterations.
    Therefore, we have
    \begin{equation}
        \Gamma^\mathrm{syn}<\frac{k^{\Gamma}(L^{p}\left(\left\{\bm x_{n}^{1}\right\}, \bm x_{o}^{1}, \bm \sigma^{1}\right)-\underline{G^{p}})}{\epsilon^{\mathrm{ac}} p^\mathrm{syn}},
    \end{equation}
    where  
    $\epsilon^{\mathrm{ac}}$ is a positive iteration factor, 
    $k^{\Gamma}$ is a constant, 
    $p^\mathrm{syn} = \Pi_{n\in\mathcal{N}}(\frac{1}{C} \sum_{c\in\mathcal{C}}p_{n,c})$ probability of successfully completing a synchronous update,
    $\underline{G^{p}}$ is the lower bound of $\sum_{n\in\mathcal{N}} g^{p}_n(\bm x_n)$,
    % \ie $\underline{G^{p}} = \inf_{\bm x_n} \sum_{n\in\mathcal{N}} g^{p}_n(\bm x_n)$,
    and $\Gamma^\mathrm{syn}$ is number of iterations, \ie 
    $\Gamma^\mathrm{syn} = \min \left\{t\mid \eta\left(\bm x^t, \bm\sigma^t\right) \leq \epsilon, t \geq 0\right\}$.

    Above all, we derive that Algorithm \ref{alg:nonconvexadmm} converge to an $\epsilon^{\mathrm{ac}}$-stationary point within $O(1/(p^\mathrm{syn}\epsilon^\mathrm{ac}))$.

\section{}\label{app:thm:single-or-multi}    
\subsection{A Proof of Theorem \ref{thm:single-or-multi}}\label{app:thm:single-or-multi-A}
If the amount of migration tasks from the local server to other servers is 0, \ie $\phi^{out}(\bm \eta^s) = \sum_{m'\in\{\mathcal{M}/m_n\}} \eta^s_{m_n,m'} \sum_{\delta \in \Delta} \lambda^{s}_{\delta,m_n}$, 
and the amount of migration tasks from other servers to the local server is 0, \ie $\phi^{in}(\bm \eta^s) = \sum_{m'\in\{\mathcal{M}/m_n\}} \eta^s_{m',m_n} \sum_{\delta \in \Delta} \lambda^{s}_{\delta,m'}$,
 the migration decision variableis $y_{m_n} = 0$.
Otherwise, vice versa.
Thus, we obtain
    \begin{equation}  
	% \label{eq:opy-y}
		y_{m_n} = \left\{
		\begin{array}{ll}
		 1, &\quad  \phi^{in}(\bm \eta^s)+\phi^{out}(\bm \eta^s)> 0\\
		 0 , &\quad otherwise,
		\end{array}
		\right..	
	\end{equation}

\subsection{An Alternative Problem for \textbf{(PoPeC)}}\label{app:thm:single-or-multi-B}
If  $\phi^{in}(\bm \eta^s)+\phi^{out}(\bm \eta^s) = 0$,
we derive $f^p_n(\hat t^\mathrm{tr}_n,\bm\eta^u) = f^s_n(\hat t^\mathrm{tr}_n,\bm\eta^u, \bm\eta^s)$.
In other words, if  $\phi^{in}(\bm \eta^s)+\phi^{out}(\bm \eta^s) = 0$, we can derive the obtained results can cover all the results of Server Collaboration, no matter what $y_{m_n}$ is.
That is $\frac{1}{N}\sum_{n\in\mathcal{N}} F_n(\bm x, \bm y, \bm z) \geq \frac{1}{N}\sum_{n\in\mathcal{N}} F_n(\bm x, \bm 1, \bm z)$, where $F_n(\bm x, \bm y, \bm z) = \mathbb{E}[A_n]$.

Combining Theorem \ref{thm:single-or-multi} and discussion of constraints, there is a comparable solution
    $\hat{\bm y^*} = \bm 1$.
    which holds $\sum_{n\in\mathcal{N}} F_n(\bm x, \bm y^{*}, \bm z)=\sum_{n\in\mathcal{N}} F_n(\bm x, \hat{\bm y^{*}}, \bm z)$.
Hence, we have 
    \begin{equation}
    \label{eq:fn=f1n}
        \sum_{n\in\mathcal{N}} F_n(\bm x, \hat{\bm y^{*}}, \bm z) \leq \sum_{n\in\mathcal{N}} F_n(\bm x, \bm y, \bm z)\leq \sum_{n\in\mathcal{N}} F_n(\bm x, \bm y, \bm z).
    \end{equation}

% Hence, we have $\hat{\bm y^{*}} = \bm 1$, which represents $$\sum_{n\in\mathcal{N}} F_n(\bm x, \hat{\bm y^{*}}, \bm z) = \sum_{n\in\mathcal{N}} F_n(\bm x, \bm y^{*}, \bm z) \leq \sum_{n\in\mathcal{N}} F_n(\bm x, \bm y, \bm z).$$

% Another benefit is that if we set $y_{m_n} = 1, \forall n$, we ignore servers that only perform local task offloading without server coordination.
Similar to Problem \textbf{P1} and \textbf{P2}, we aim to minimize a highly accurate upper bound for the average expected PAoI of multi-priority users through the following approach.
Thus, one alternative for \textbf{(PoPeC)} is
\begin{equation}
\label{eq:P3-proof}
\begin{split}
    \textbf{(P3)}~&\mathop{\min}_{\bm x, \bm z}
    \frac{1}{N}\sum_{n\in\mathcal{N}} F^1_n(\bm x, \bm z) \\
    \text{s.t.}&~ 
    \eqref{eq:c1},\eqref{eq:c2},%to
    \eqref{eq:c3a},\eqref{eq:c5a},\eqref{eq:c4a},%tm
    \eqref{eq:c6},%transmission
    % \eqref{eq:c-y}, %y
    \eqref{eq:c7b},\eqref{eq:c8b}, %capacity
\end{split}
\end{equation}
where $F^1_n(\bm x, \bm z) = \hat t^\mathrm{tr}_n + \frac{1}{\sum_{c\in\mathcal{C}}p_{n,c}\eta^u_{n,c}\lambda_{n}} + \sum_{m\in\mathcal{M}} \pi_{n,m}(\bm z)$,
$\bm x = \{\hat t^\mathrm{tr}, \bm\eta^u\}$ and $\bm z = \bm\eta^s$.

\section{A Proof of Lemma \ref{lemma:Transform-DataMigration-ChannelAllocation}}\label{app:lemma:Transform-DataMigration-ChannelAllocation}  
\textbf{P3} is
\begin{equation}
\label{eq:P3-copy}
\begin{split}
    &\mathop{\min}_{\bm x, \bm z}
    \frac{1}{N}\sum_{n\in\mathcal{N}} F^1_n(\bm x, \bm z) \\
    =&\mathop{\min}_{\bm x, \bm z}
    \frac{1}{N}\sum_{n\in\mathcal{N}}( F^2_n(\bm x) +F^3_n(\bm z, \bm \lambda^s))\\
    \text{s.t.}&~ 
    \eqref{eq:c1},\eqref{eq:c2},%to
    \eqref{eq:c3a},\eqref{eq:c4a},%tm
    \eqref{eq:c6},%transmission
    % \eqref{eq:c-y}, %y
    \eqref{eq:c7b},\eqref{eq:c8b},\\ %capacity\\
    \sum_{m' \in \mathcal{M}} &  \eta^s_{m,m'} \lambda^{s}_{\delta,m}=\sum_{n \in \mathcal{N}^{\delta}_m}\sum_{c \in C} p_{n,c} \eta^u_{n,c}\lambda_{n},
\end{split}
\end{equation}
If you decompose the problem in terms of different variables,
$\sum_{m' \in \mathcal{M}}  \eta^s_{m,m'} \lambda^{s}_{\delta,m}$ and $\sum_{n \in \mathcal{N}^{\delta}_m}\sum_{c \in C} p_{n,c} \eta^u_{n,c}\lambda_{n}$ will always remain consistent due to the presence of the constraint \eqref{eq:c5a}
which is equivalent to
\begin{equation}
\label{eq:c-to&tm}
    \lambda^{s}_{\delta,m} \sum_{m' \in \mathcal{M}} \eta^s_{m,m'}  = \sum_{n \in \mathcal{N}^{\delta}_m}\sum_{c \in C} p_{n,c} \eta^u_{n,c}\lambda_{n},
    % = \sum_{\delta \in \Delta} \lambda^{s}_{\delta,m}
    , \quad \forall m \in \mathcal{M}.
\end{equation}
This makes it simpler to solve for the best $\bm x$ and $\bm z$ because their values don't change when they are solved iteratively.
Therefore, in addition to decomposing the objective functions and constraints as
\begin{align}
    \textbf{(sp1)}~&\mathop{\min}_{\bm x}\frac{1}{N}\sum_{n\in\mathcal{N}} F^2_n(\bm x) \nonumber\\
    \text{s.t.}&~ \eqref{eq:c1},\eqref{eq:c2},\eqref{eq:c5a},\eqref{eq:c6},\eqref{eq:c7b},\eqref{eq:c8b},
\end{align}

\begin{equation}
\begin{split}
    \textbf{(sp2)}~&\mathop{\min}_{\bm z}\frac{1}{N}\sum_{n\in\mathcal{N}} F^3_n(\bm z, \bm \lambda^s) \\
    \text{s.t.}&~ \eqref{eq:c3a},\eqref{eq:c5a},\eqref{eq:c4a},
\end{split}
\end{equation}
constraint \eqref{eq:c5a} in \textbf{sp1} is replaced by the auxiliary inequality \eqref{eq:DataMigration-AuxiliaryInequality2} as shown in \textbf{P3-1},
% to guarantee a better solution of \textbf{P3} 
    \begin{align}
        \label{eq:DataMigration-AuxiliaryInequality2}
             \sum_{n \in \mathcal{N}^{\delta}_m}\sum_{c \in C} p_{n,c} \eta^u_{n,c}\lambda_{n} &\leq \sum_{m' \in \mathcal{M}}  \eta^s_{m,m'} \lambda^{s}_{\delta,m}
    \end{align}
which leads to a better solution for $\bm x$, the reason is as follows
 
    Denote $\bm x_1$, $\bm x_2$, $\bm x_3$ as
    \begin{align}
        \bm x_1 = \mathop{\arg \min}_{\bm x} \Big\{ \sum_{n\in\mathcal{N}} F^2_n(\bm x)  ~ \text{s.t.} ~ \eqref{eq:c1},\eqref{eq:c2},\eqref{eq:c6},\eqref{eq:c7b},\eqref{eq:c8b},\nonumber\\
        \{\sum_{n \in \mathcal{N}^{\delta}_m}\sum_{c \in C} p_{n,c} \eta^u_{n,c}\lambda_{n} \leq \sum_{m' \in \mathcal{M}}  \eta^s_{m,m'} \lambda^{s}_{\delta,m}\} \Big\}\nonumber
    \end{align}
    \begin{align}
        \bm x_2 = \mathop{\arg \min}_{\bm x} \Big\{ \sum_{n\in\mathcal{N}} F^2_n(\bm x)  ~ \text{s.t.} ~ \eqref{eq:c1},\eqref{eq:c2},\eqref{eq:c6},\eqref{eq:c7b},\eqref{eq:c8b},\nonumber\\
        \{\sum_{n \in \mathcal{N}^{\delta}_m}\sum_{c \in C} p_{n,c} \eta^u_{n,c}\lambda_{n} = \sum_{m' \in \mathcal{M}}  \eta^s_{m,m'} \lambda^{s}_{\delta,m}\} \Big\}\nonumber
    \end{align}
    \begin{align}
        \bm x_3 = \mathop{\arg \min}_{\bm x} \Big\{ \sum_{n\in\mathcal{N}} F^2_n(\bm x)  ~ \text{s.t.} ~ \eqref{eq:c1},\eqref{eq:c2},\eqref{eq:c6},\eqref{eq:c7b},\eqref{eq:c8b},\nonumber\\
        \{\sum_{n \in \mathcal{N}^{\delta}_m}\sum_{c \in C} p_{n,c} \eta^u_{n,c}\lambda_{n} < \sum_{m' \in \mathcal{M}}  \eta^s_{m,m'} \lambda^{s}_{\delta,m}\} \Big\}\nonumber
    \end{align}
    
    Obviously, 
    \begin{equation}
        \sum_{n\in\mathcal{N}} F_n(\bm x_1) = \min \{\sum_{n\in\mathcal{N}} F_n(\bm x_2), \sum_{n\in\mathcal{N}} F_n(\bm x_3)\}.\nonumber
    \end{equation}
Meanwhile, the constraint \eqref{eq:c5a} in \textbf{sp2} is unchanged as shown in \textbf{P3-2}, in order to ensure that constraint \eqref{eq:c5a} of \textbf{P3} is satisfied.

\section{}\label{app:DataMigration-MigrationAllocation}
\subsubsection{Problem Transform}\label{app:proposition:DataMigration-MigrationAllocation1}
First, $\bm \lambda^s$ can be obtained by $\lambda^{s}_{\delta,m}=\sum_{n\in\mathcal{N}^{\delta}_m}\sum_{c\in\mathcal{C}} p_{n,c} \lambda_{n,c}$ according to system model.
With given $\bm \lambda^s$, transforming the problem from \textbf{P3-2} to \textbf{3-5}, a similar proof is already discussed in Proposition \ref{pro:npl-mp}.
$\theta^{*}_n$ is achieved if and only if
	\begin{align}
	\mathop{\min}_{\{\bm z_{n,m}\}} \{\pi^{u}_{n,m}(\bm z_{n,m}) - \vartheta^{*}_{n,m} \pi^{l}_{n,m}(\bm z_{n,m})\}
	\nonumber\\ 
       = \pi^{u}_{n,m}(\bm z^*_{n,m}) - \vartheta^{*}_{n,m} \pi^{l}_{n,m}(\bm z^*_{n,m}) = 0.
    \label{eq:NFP-optimal-theta2}
    \end{align}
This is a necessary and sufficient condition.

We can also use Algorithm \ref{alg:NFP} (NFPA) to transform this problem with a slight modification.

\subsubsection{Problem Solving}\label{app:lemma:DataMigration-MigrationAllocation2}
Problems \textbf{P3–5} are challenging to solve since they are still non-convex, but it is simple to identify that they are cubic function problems with a finite range of independent variables taking values. Alternatively put, this is a Lipschitz smooth function with constant $\ell_{n,m}$ satisfying
\begin{equation}
    \ell_{n,m} I \succeq \nabla^{2} \big(\pi^{u}_{n,m}(\bm z_{n,m}) - \vartheta^{*}_{n,m} \pi^{l}_{n,m}(\bm z_{n,m})\big),
\end{equation}
which is like $g^{p}_n$ in Proposition \ref{theorem:NFP_Problem}. 

As a result, we can likewise efficiently handle $\{\bm z_{n,m}\}$ using Algorithm 2 with a few slight adjustments.
\subsubsection{Convergence Analysis}\label{app:proof:DataMigration-MigrationAllocation3}
    A  similar instead of the same result of Convergence Analysis can be obtained by Theorem \ref{theorem:innerConvergence1}.
    If $\rho_n>2\ell_n$, Algorithm \ref{alg:nonconvexadmm} converge to an $\epsilon^{\mathrm{ac}}$-stationary point within $O(1/(p^\mathrm{syn}\epsilon^\mathrm{ac}))$.
    However, in a network with wired communication, the reliability rate of synchronous iterative communication is $p^\mathrm{syn}=1$
    Thus, we derive the NAC algorithm converges within $O(1/\epsilon^{\mathrm{ac}}) in this case$.

\section{A Proof of Theorem \ref{thm:DataMigration-Convergence}}\label{app:thm:DataMigration-Convergence}
    % We prove Theorem \ref{thm:DataMigration-Convergence} here.
    
    First, we see the equivalent optimal solution of the migration $\hat{\bm y^{*}}$ from Theorem \ref{thm:single-or-multi}. Hence, we have these two inequalities as
    \begin{align}
    \begin{split}
        &\frac{1}{N}\sum_{n\in\mathcal{N}} F_n(\bm x^t, \hat{\bm y^{*}}, \bm z^t) \\
        \mathop{=}^{(a)} &\frac{1}{N}\sum_{n\in\mathcal{N}} F^1_n(\bm x^t, \bm z^t) \\
        \mathop{=}^{(b)} &\frac{1}{N}\sum_{n\in\mathcal{N}} (F^2_n(\bm x^{t})+F^3_n(\bm z^{t}, \bm \lambda^{s,t}))\\
        \mathop{\ge}^{(c)} &\frac{1}{N}\sum_{n\in\mathcal{N}} (F^2_n(\bm x^{t+1})+F^3_n(\bm z^{t}, \bm \lambda^{s,t}))\\
        \mathop{\ge}^{(d)} &\frac{1}{N}\sum_{n\in\mathcal{N}} (F^2_n(\bm x^{t+1})+F^3_n(\bm z^{t}, \bm \lambda^{s,t+1}))\\
        \mathop{\ge}^{(e)} &\frac{1}{N}\sum_{n\in\mathcal{N}} (F^2_n(\bm x^{t+1})+F^3_n(\bm z^{t+1}, \bm \lambda^{s,t+1}))\\
        \mathop{=}^{(f)} &\frac{1}{N}\sum_{n\in\mathcal{N}} F^1_n(\bm x^{t+1},\bm z^{t+1})\\
        \mathop{=}^{(g)} &\frac{1}{N}\sum_{n\in\mathcal{N}} F_n(\bm x^{t+1}, \hat{\bm y^{*}}, \bm z^{t+1}).
    \end{split}
    \end{align}
    where 
    (a) and (g) hold according to \eqref{eq:fn=f1n},
    (b) and (f) hold according to Lemma \ref{lemma:Transform-DataMigration-ChannelAllocation},
    (c) holds according to Lemma \ref{lemma:DataMigration-ChannelAllocation},
    (e) holds according to Lemma \ref{lemma:DataMigration-MigrationAllocation2} and the convergence of the NAC algorithm,
    (d) holds and the reason is as follows.
    According to \eqref{eq:DataMigration-AuxiliaryInequality}, we obtain 
    \begin{equation}
    \lambda^{s,t}_{\delta,m}\ge\lambda^{s,t+1}_{\delta,m}, \quad, \forall \delta \in \Delta, \forall m \in \mathcal{M}.
    \end{equation}
    % $\forall \delta \in \Delta, \forall m \in \mathcal{M},~ $.
    In addition, it is easy to get that 
    \begin{equation}
         \frac{\partial F^3_n(\bm z^{t}, \bm \lambda^{s,t})}{\partial \lambda^{s,t}_{\delta,m}}>0,\quad,\forall \delta \in \Delta, \forall m \in \mathcal{M},
    \end{equation}
    in the domain of the definition of $\bm \lambda^s$.
    Thus, we have $F^3_n(\bm z^{t}, \bm \lambda^{s,t})\ge F^3_n(\bm z^{t}, \bm \lambda^{s,t+1})$.

    Based on the inequalities, we obtain
    \begin{align}
        \frac{1}{N}\sum_{n\in\mathcal{N}} F_n(\bm x^t, \hat{\bm y^{*}}, \bm z^t) \ge \frac{1}{N}\sum_{n\in\mathcal{N}} F_n(\bm x^{t+1}, \hat{\bm y^{*}}, \bm z^{t+1})
    \end{align}
    $F_n(\bm x, \bm y, \bm z)$ is a function with positive values and a lower bound.
    Therefore, $F_n(\bm x, \bm y, \bm z)$ monotonically decreases and converges to a unique point.

{\section{Convergence and convergence rate proof}
\label{app:GapFunction-Convergence}
The analysis provided above diverges from the conventional examination of the ADMM (Alternating Direction Method of Multipliers) algorithm, which typically centers on constraining the distance between the current iteration and the optimal solution set. The preceding analysis draws inspiration in part from our prior study of the ADMM's convergence within the context of multi-block convex problems. In that context, the algorithm's advancement is gauged by the combined reduction in specific primal and dual gaps, as detailed in \cite[Theorem 3.1]{hong2017linear}.
However, the nonconvex nature of the problem introduces challenges in estimating either the primal or dual optimality gaps. Hence, in this context, we've opted to employ the reduction of the augmented Lagrangian as a metric to gauge the algorithm's progress.
Moving forward, we delve into the examination of the iteration complexity pertaining to the basic ADMM. In articulating our outcome, we establish the concept of the augmented Lagrangian function's "proximal gradient."
\begin{align*}
    &\tilde{\nabla} L^{p}\left(\left\{\bm x_{n}\right\}, \bm x_{o}, \bm\sigma^t\right)\nonumber\\
   &=\left[\begin{array}{c}
    \bm x_{o} - \prox_{g^p} \left[x_{o} - \nabla_{\bm x_{o}} \left( L^{p}\left(\left\{\bm x_{n}\right\}, \bm x_{o}, \bm\sigma^t\right) - g^p(\bm x_o) \right)\right]  \\
    \nabla_{\bm x_{1}} L^{p}\left(\left\{\bm x_{n}\right\}, \bm x_{o}, \bm\sigma^t\right) \\
    \vdots \\
    \nabla_{\bm x_{N}} L^{p}\left(\left\{\bm x_{n}\right\}, \bm x_{o}, \bm\sigma^t\right)
    \end{array}\right].
\end{align*}
where $\prox_{h}[z]:=\arg\min_{x} h(\bm x)+\frac{1}{2}\|\bm x-\bm z\|^2$ is the proximity operator. 
We will use the following quantity to measure the progress of the algorithm
$\eta(\bm x^t, \bm\sigma^t):=\|\tilde{\nabla}L^p(\{\bm x^t_n\},\bm x^t_0,\bm\sigma^t)\|^2+\sum_{n\in\mathcal{N}}\|\bm x^{t}_n-\bm x^{t}_0\|^2.$
It can be verified that if $\eta(\bm x^t, \bm\sigma^t)\to 0$, then a stationary solution to the problem is obtained. We have the following iteration complexity result:
\begin{theorem}
    Let $T(\epsilon)$ denote an iteration index in which the following inequality is achieved
    $T(\epsilon):=\min\left\{t\mid \eta(\bm x^t, \bm\sigma^t)\le \epsilon, t\ge 0\right\}$
    for some $\epsilon>0$. Then there exists some constant $k^{\Gamma}>0$ such that
    \begin{align}
    \epsilon\le \frac{k^{\Gamma} (L^p(\{\bm x^1_n\},\bm x^1_0,\bm\sigma^1)-\underline{G^p})}{T(\epsilon)}.
    \end{align}
\end{theorem}
\begin{proof}
    We first show that there exists a constant $\kappa_1>0$ such that
    \begin{align}
        \|\tilde{\nabla}L^p(\{\bm x^t_n\},\bm x^t_0,\bm\sigma^t)\|\le \kappa_1 & \left(\|\bm x_0^{t+1}-\bm x_0^t\|+\sum_{n\in\mathcal{N}}\|\bm x^{t+1}_n-\bm x^t_n\|\right), \nonumber\\ 
        &\forall~r\ge 1. \label{eq:sigma}
    \end{align}
    This proof follows similar steps of \cite[Lemma 2.5]{hong2017linear}. From the optimality condition of the $\bm x_0$, we have
    $\bm x^{t+1}_0=\prox_{g^p}\left[ \bm x^{t+1}_0-\sum_{n\in\mathcal{N}}\rho_n\left(\bm x^{t+1}_0-\bm x^t_n-\frac{\bm\sigma^t_n}{\rho_n}\right)  \right].$
    This implies that
    \begin{align}
        &\|\bm x^t_0-\prox_{h}\left[\bm x^t_0-\nabla_{\bm x_0}(L^p(\{\bm x^t_n\},\bm x^t_0,\bm\sigma^t)-g^p(\bm x^t_0))\right]\|\nonumber\\
        &=\left\|\bm x^{t}_0-\bm x^{t+1}_0+\bm x^{t+1}_0-\prox_{g^p}\left[\bm x^t_0-\sum_{n\in\mathcal{N}}\rho_n(\bm x^t_0-\bm x^{t}_n-\frac{\bm\sigma^t_n}{\rho_n})\right]\right\|\nonumber\\
        &\le \|\bm x^{t}_0-\bm x^{t+1}_0\|+\Bigg\|\prox_{g^p}\left[ \bm x^{t+1}_0-\sum_{n\in\mathcal{N}}\rho_n\left(\bm x^{t+1}_0-\bm x^t_n-\frac{\bm\sigma^t_n}{\rho_n}\right)  \right]\nonumber\\
        &\quad\quad\quad\quad-\prox_{g^p}\left[\bm x^t_0-\sum_{n\in\mathcal{N}}\rho_n(\bm x^t_0-\bm x^{t}_n-\frac{\bm\sigma^t_n}{\rho_n})\right]\Bigg\|\nonumber\\
        &\le 2\|\bm x^{t+1}_0-\bm x^t_0\|+\sum_{n\in\mathcal{N}}\rho_n\|\bm x^t_0-\bm x^{t+1}_0\|,
        \label{eq:prox1}
    \end{align}
    where in the last inequality we have used the nonexpansiveness of the proximity operator.
    {
    Similarly, the optimality condition of the $\bm x_n$ subproblem is given by
    $$\nabla g^p_n(\bm x^{t+1}_n)+\rho_n\left(\bm x^{t+1}_n-\bm x^{t+1}_0+\frac{\bm\sigma^{t}_n}{\rho_n}\right)=0.$$
    Therefore, we derive
    \begin{align}
    &\|\nabla_{\bm x_n}L^p(\{\bm x^t_n\},\bm x^t_0,\bm\sigma^t)\|\nonumber\\
    &=\|\nabla g^p_n(\bm x^{t}_n)+\rho_n(\bm x^{t}_n-\bm x^{t}_0+\frac{\bm\sigma^{t}_n}{\rho_n})\|\nonumber\\
    &= \|\left(\nabla g^p_n(\bm x^{t}_n)+\rho_n(\bm x^{t}_n-\bm x^{t}_0+\frac{\bm\sigma^{t}_n}{\rho_n})\right)\nonumber\\
    &- \left(\nabla g^p_n(\bm x^{t+1}_n)+\rho_n(\bm x^{t+1}_n-\bm x^{t+1}_0+\frac{\bm\sigma^{t}_n}{\rho_n})\right)\|\nonumber\\
    &\le (L_n+\rho_n)\|\bm x^t_n-\bm x^{t+1}_n\|+\rho_n\|\bm x^t_0-\bm x^{t+1}_0\|\label{eq:prox2}.
    \end{align}
    }
    Therefore, combining \eqref{eq:prox1} and \eqref{eq:prox2}, we have
    \begin{align}
      \|\tilde{\nabla}L^p(\{\bm x^t_n\},\bm x^t_0,\bm\sigma^t)\|\le \left(2+\sum_{n\in\mathcal{N}}2\rho_n\right)\|\bm x^t_0-\bm x^{t+1}_0\| \nonumber\\
      +\sum_{n\in\mathcal{N}}(L_n+\rho_n)\|\bm x^t_n-\bm x^{t+1}_n\|.\label{eq:estimate:nL}
    \end{align}
    By taking $\kappa_1=\max\left\{(2+ \sum_{n\in\mathcal{N}}2\rho_n), L_1+\rho_1, \cdots, L_n+\rho_n\right\}$, \eqref{eq:sigma} is proved.
    According to \cite[Lemma 2.1]{hong2016convergence}, we obtain
    \begin{align}
    \sum_{n\in\mathcal{N}}\|\bm x^t_n-\bm x^t_0\| = \sum_{n\in\mathcal{N}}\frac{1}{\rho_n}\|\bm\sigma^{t+1}_n-\bm\sigma^t_n\| \le \sum_{n\in\mathcal{N}}\frac{L_n}{\rho_n}\|\bm x_n^{t+1}-\bm x^t_n\|.\label{eq:estimate:x}
    \end{align}
    The inequalities \eqref{eq:estimate:nL} -- \eqref{eq:estimate:x} implies that for some $\kappa_3>0$
    \begin{align}
    &\sum_{n\in\mathcal{N}}\|\bm x^t_n-\bm x^t_0\|^2 + \|\tilde{\nabla}L^p(\{\bm x^t_n\},\bm x^t_0,\bm\sigma^t)\|^2\nonumber\\
    &\le \kappa_3\left(\|\bm x^t_0-\bm x^{t+1}_0\|^2+\sum_{n\in\mathcal{N}}\|\bm x^t_n-\bm x^{t+1}_n\|^2\right). \label{eq:total:square}
    \end{align}
    According to Lemma \cite[Lemma 2.2]{hong2016convergence}, there exists a constant $\kappa_2 = \min\left\{\{ \frac{\gamma_n(\rho_n)}{2}-\frac{L^2_n}{\rho_n}\}_{n\in\mathcal{N}}, \frac{\gamma}{2} \right\}$
    such that
    \begin{align}
    &L^p(\{\bm x^{t}_n\}, x_0^{t}; \bm\sigma^{t})-L^p(\{\bm x^{t+1}_n\}, x_0^{t+1}; \bm\sigma^{t+1})\nonumber\\
    &\ge\kappa_2\left(\sum_{n\in\mathcal{N}}\|\bm x^{t+1}_n-\bm x_{k}^{t}\|^2+\|\bm x_0^{t+1}-\bm x_0^t\|^2\right).\label{eq:estimate:L}
    \end{align}
    Combining \eqref{eq:total:square} and \eqref{eq:estimate:L} we get
    \begin{align}
      &\sum_{n\in\mathcal{N}}\|\bm x^t_n-\bm x^t_0\|^2 + \|\tilde{\nabla}L^p(\{\bm x^t_n\},\bm x^t_0,\bm\sigma^t)\|^2\nonumber\\
      &\le \frac{\kappa_3}{\kappa_2}\left(L^p(\{\bm x^{t}_n\}, x_0^{t}; \bm\sigma^{t})-L^p(\{\bm x^{t+1}_n\}, x_0^{t+1}; \bm\sigma^{t+1})\right)\nonumber.
    \end{align}
    Summing both sides of the inequality above for $t=1,\cdots, r$, we obtain
    \begin{align}
    &\sum_{t=1}^{r}\sum_{n\in\mathcal{N}}\|\bm x^t_n-\bm x^t_0\|^2 + \|\tilde{\nabla}L^p(\{\bm x^t_n\},\bm x^t_0,\bm\sigma^t)\|^2\nonumber\\
      &\le \frac{\kappa_3}{\kappa_2}\left(L^p(\{\bm x^{1}_n\}, x_0^{1}; \bm\sigma^{1})-L^p(\{\bm x^{r+1}_n\}, x_0^{r+1}; \bm\sigma^{r+1})\right)\nonumber\\
    &\le \frac{\kappa_3}{\kappa_2}\left(L^p(\{\bm x^{1}_n\}, x_0^{1}; \bm\sigma^{1})-\underline{G^p}\right).\nonumber
    \end{align}
    Rewriting the final inequality, it becomes evident that we leverage the property that $L^p({\bm x^{r+1}_n}, x_0^{r+1}; \bm\sigma^{r+1})$ exhibits a decreasing trend while remaining above the lower bound $\underline{G^p}$, as previously established in \cite[Lemmas 2.2–2.3]{hong2016convergence}.
    By utilizing the definitions of $T(\epsilon)$ and $\eta(\bm x^t, \bm\sigma^t)$, the above inequality becomes
    \begin{align}
    &T(\epsilon)\epsilon\le \frac{\kappa_3}{\kappa_2}\left(L^p(\{\bm x^{1}_n\}, x_0^{1}; \bm\sigma^{1})-\underline{G^p}\right)
    \end{align}
    By dividing both sides of the equation by $T(\epsilon)$ and choosing $C=\kappa_3/\kappa_2$, the intended outcome is achieved.
\end{proof}
}

{\section{A Proof of NP-hardness}\label{app:NP-hard}
In this case of multi-priority, we need to determine an appropriate offloading policy and service rules to accomplish task scheduling, considering the multi-objective optimization objectives of PAoI and priority task emphasis, which is a classical multi-objective single-machine scheduling problem with sequence-dependent setup times.  
Such a problem can be transformed into a Multi-objective Traveling Salesman Problem (MOTSP)~\cite{yalaoui2003efficient}.
MOTSP is an extended instance of a traveling salesman problem (TSP).
Thus, the case of multi-priority can be reduced to TSP.
% Further, we prove that TSP is a well-known NP-complete.
\begin{theorem}
\label{thm:np-complete}
    The traveling salesman problem is NP-complete.
\end{theorem}
\begin{proof}
\textit{Verification of TSP Membership in NP:}
We first establish that the Traveling Salesman Problem (TSP) belongs to the class of decision problems that can be verified in polynomial time. The verification process uses a certificate, which is a sequence of n vertices representing a tour. The algorithm checks whether this sequence contains each vertex exactly once, computes the sum of edge costs, and verifies that this sum is at most k. This verification process can certainly be performed in polynomial time.\\
\textit{Proving TSP is NP-hard:}
To demonstrate that TSP is NP-hard, we establish a reduction from the Hamiltonian Cycle problem (HAM-CYCLE). 
Let $G = (V, E)$ be an instance of HAM-CYCLE. We construct an instance of TSP as follows: 
We create the complete graph $G' = (V, $E'$)$, 
where $E' = \{(i, j): i, j \in V \& i \neq j\}$, 
and we define the cost function c as follows:
\begin{equation}
    c(i, j) = \begin{cases}
    0 & \text{if }(i, j) \in E \\
    1 & \text{if }(i, j) \notin E
    \end{cases}
\end{equation}
Note that because $G$ is undirected, it has no self-loops, so $c(v, v) = 1$ for all vertices $v \in V$. The instance of TSP is then $(G', c, 0)$, which can be easily created in polynomial time.
Now, we demonstrate that graph $G$ has a Hamiltonian cycle if and only if graph $G'$ has a tour of cost at most 0.
Suppose that graph $G$ has a Hamiltonian cycle $h$. Each edge in $h$ belongs to $E$ and thus has cost 0 in $G'$. Thus, $h$ is a tour in $G'$ with cost 0.
Conversely, suppose that graph $G$' has a tour $h'$ of cost at most 0. Since the costs of the edges in $E'$ are 0 and 1, the cost of the tour $h'$ is exactly 0, and each edge on the tour must have cost 0. Therefore, $h'$ contains only edges in $E$. We conclude that $h'$ is a Hamiltonian cycle in graph $G$.
\end{proof}
According to Theorem \ref{thm:np-complete}, TSP is an NP-complete and NP-hard problem.
Thus, the case of multi-priority is an NP-hard problem.
In summary, unless P = NP, solving such a problem cannot be achieved within polynomial time~\cite{blondel1997np}.}

{\section{Why Multi-Class }\label{app:proof:sp-mp}  
\subsection{The Highest-class priority user vs. priority-free user}\label{app:proof:sp-m:a}
For the highest-class priority user $n^*$, we have
\begin{align}
    &\mathbb{E}[A_{n^*}] - \mathbb{E}[A^{p}_{n^*}] \nonumber\\
    =& (\mathbb{E}[T_{n^*}]+\mathbb{E}[I_{n^*}]+\mathbb{E}[W_{n^*}]+\mathbb{E}[Y_{n^*}])\nonumber\\
    &-(\mathbb{E}[T^{p}_{n^*}]+\mathbb{E}[I^{p}_{n^*}]+\mathbb{E}[W^{p}_{n^*}]+\mathbb{E}[Y^{p}_{n^*}])\nonumber\\
    \mathop{=}^{(a)}& \mathbb{E}[W_{n^*}]-\mathbb{E}[W^{p}_{n^*}]\nonumber\\
    =& \Upsilon(\bm\eta^u) \frac{\zeta_{\Delta}(\bm\eta^u)-\zeta_{\delta(n^*)}(\bm\eta^u)}{(1-\zeta_{\Delta}(\bm\eta^u))(1-\zeta_{\delta(n^*)}(\bm\eta^u))} \nonumber\\
    \mathop{\geq}^{(b)}& 0.
\end{align}
where 
$\zeta_{\delta(n^*)}(\bm\eta^u) =\sum_{\delta\in\Delta(\delta(n^*))}\sum_{n'\in\mathcal{N}^{\delta}}\sum_{c\in\mathcal{C}}p_{n',c}\frac{\eta^u_{n',c}\lambda_{n'}}{\mu_{n'}}$,  
$\zeta_{\Delta}(\bm\eta^u) =\sum_{\delta\in\Delta}\sum_{n'\in\mathcal{N}^{\delta}}\sum_{c\in\mathcal{C}}p_{n',c}\frac{\eta^u_{n',c}\lambda_{n'}}{\mu_{n'}}$,  and
$\Upsilon(\bm\eta^u) = \frac{1}{2}\sum_{\delta\in\Delta}\sum_{n'\in\mathcal{N}^{\delta}}\sum_{c\in\mathcal{C}}p_{n',c}\eta^u_{n',c}\lambda_{n'}\nu_{n'}$.
}

{(a) holds because we are contrasting in the context of the same strategy and at this point $\mathbb{E}[T_{n^*}]=\mathbb{E}[T^p_{n^*}]$, $\mathbb{E}[Y_{n^*}]=\mathbb{E}[Y^p_{n^*}]$ and $\mathbb{E}[I_{n^*}]=\mathbb{E}[I^p_{n^*}]$.
Inequality (b) holds because $\zeta_{\delta(n^*)}(\bm\eta^u)\leq \zeta_{\Delta}(\bm\eta^u)$ and equality sign achieves if and only if $\Delta=\Delta(\delta(n^*))$.
}

{\subsection{The Lowest-class priority user vs. priority-free user}\label{app:proof:sp-m:b}
For the lowest-class priority user $n_*$, we derive
\begin{align}
    &\mathbb{E}[A_{n_*}] - \mathbb{E}[A^{p}_{n_*}] \nonumber\\
    =& (\mathbb{E}[T_{n_*}]+\mathbb{E}[I_{n_*}]+\mathbb{E}[W_{n_*}]+\mathbb{E}[Y_{n_*}])\nonumber\\
    &-(\mathbb{E}[T^{p}_{n_*}]+\mathbb{E}[I^{p}_{n_*}]+\mathbb{E}[W^{p}_{n_*}]+\mathbb{E}[Y^{p}_{n_*}])\nonumber\\
    \mathop{=}^{}& \mathbb{E}[W_{n_*}]-\mathbb{E}[W^{p}_{n_*}]\nonumber\\
    =& \frac{\Upsilon(\bm\eta^u)}{1-\zeta_{\Delta}(\bm\eta^u)}(1-\frac{1}{1-\zeta_{\delta(n)-1}(\bm\eta^u)})\nonumber\\
    \mathop{\leq}^{(c)}& 0.
\end{align}
Inequality (c) holds because $\zeta_{\delta(n^*)}(\bm\eta^u), \zeta_{\Delta}(\bm\eta^u) \in (0,1)$ and equality sign achieves if and only if $\Delta=\Delta(\delta(n^*))$.
\subsection{Comparison between users with different priority levels}\label{app:proof:sp-m:c}
For user $n$ who belong to the priority level from $\delta^0$ to $\delta^0-\hat\delta$, we obtain
\begin{align}
    &    \mathbb{E}[A^{p}_{n}|\delta(n)=\delta^0] - \mathbb{E}[A^{p}_{n}|\delta(n)=\delta^0-\hat\delta] \nonumber\\
    =& \mathbb{E}[W^{p}_{n}|\delta(n)=\delta^0] - \mathbb{E}[W^{p}_{n}|\delta(n)=\delta^0-\hat\delta]\nonumber\\
    =& \frac{\Upsilon(\bm\eta^u)}{(1-\zeta_{\delta^0}(\bm\eta^u))(1-\zeta_{\delta^0 - 1}(\bm\eta^u))}\nonumber\\
    &- \frac{\Upsilon(\bm\eta^u)}{(1-\zeta_{\delta^0 - \hat\delta}(\bm\eta^u))(1-\zeta_{\delta^0 - \hat\delta - 1}(\bm\eta^u))} \nonumber\\
    \ge& \frac{\Upsilon(\bm\eta^u)(\zeta_{\delta^0}(\bm\eta^u)-\zeta_{\delta^0 - \hat\delta - 1}(\bm\eta^u))}{(1-\zeta_{\delta^0}(\bm\eta^u))(1-\zeta_{\delta^0 - 1}(\bm\eta^u))(1-\zeta_{\delta^0 - \hat\delta - 1}(\bm\eta^u))}\nonumber\\
    % &- \frac{\Upsilon(\bm\eta^u)}{(1-\zeta_{\delta^0 - \hat\delta}(\bm\eta^u))(1-\zeta_{\delta^0 - \hat\delta - 1}(\bm\eta^u))} \nonumber\\
    \mathop{\geq}^{(d)}& 0.
\end{align}
where holds because $\zeta_{\delta^0}(\bm\eta^u)\geq\zeta_{\delta^0 - \hat\delta}(\bm\eta^u)$ and $\zeta_{\delta^0 - 1}(\bm\eta^u)\geq\zeta_{\delta^0 - \hat\delta -1}(\bm\eta^u)$.
}

\end{appendices}

\end{document}